\documentclass[reprint,
 superscriptaddress,
nofootinbib,
 aps,
pra,
onecolumn,
]{revtex4-2}

\usepackage[utf8]{inputenc} 
\usepackage{amsmath,amsfonts,amsthm,amssymb} 
\usepackage[T1]{fontenc}    
\usepackage{url}            
\usepackage{xurl}            
\usepackage{multirow}
\usepackage{booktabs}       
\usepackage{nicefrac}       
\usepackage{microtype}      
\usepackage{graphicx}
\usepackage{xcolor, etoolbox}
\usepackage{mathtools}
\usepackage{ragged2e}
\usepackage{braket}      
\usepackage[english]{babel} 
\usepackage{enumitem}
\usepackage{xspace}
\usepackage{dsfont}
\usepackage{bm}
\usepackage{bbm}
\usepackage{dcolumn}
\usepackage{mathrsfs}
\usepackage{tikz}
\usepackage{graphicx} 
\usepackage{lipsum}
\usepackage{lmodern}
\usepackage{leftindex}
\usepackage{comment}
\usepackage{tensor}
\usepackage{blkarray}
\usepackage{hyperref}       
\usepackage[capitalise]{cleveref}
\usepackage{pgfplots}

\newcommand{\HS}{\mathcal{H}}

\DeclarePairedDelimiter\abs{\lvert}{\rvert}

\newcommand{\norm}[1]{\left\|#1\right\|}

\newcommand{\eps}{\varepsilon}

\newcommand{\mc}[1]{{\mathcal{#1}}}

\newcommand{\ketbra}[2]{\ket{#1}\!\bra{#2}}
\newcommand{\proj}[1]{\ketbra{#1}{#1}}
\newcommand{\cq}{\coloneqq}

\newtheorem{theorem}{Theorem}
\newtheorem{lemma}[theorem]{Lemma}
\newtheorem{proposition}[theorem]{Proposition}
\newtheorem{corollary}{Corollary}[theorem]

\newtheorem{definition}{Definition}
\newtheorem{problem}{Problem}

\newtheorem{repeatedproposition}{Proposition}

\newenvironment{repeatProp}[1]
  {\renewcommand{\therepeatedproposition}{\ref{#1}, informal}%
   \begin{repeatedproposition}}
  {\end{repeatedproposition}}

  \newenvironment{repeatProp2}[2]
  {\renewcommand{\therepeatedproposition}{\ref{#1} and \ref{#2}, informal}%
   \begin{repeatedproposition}}
  {\end{repeatedproposition}}

\hypersetup{
    colorlinks,
    linkcolor=blue,
    citecolor=blue,
    urlcolor=blue,
    breaklinks=true 
}

\title{Optimal transducers - QIP}
\date{August 2026}

\begin{document}

\title{Optimal transducers using symmetries}

\author{Benoît Dubus}
    \email{benoit.dubus@ulb.be}
    \affiliation{Centre for Quantum Information and Communication, École polytechnique de Bruxelles, CP 165, Université libre de Bruxelles, 1050 Brussels, Belgium}
\author{Julien Ladeuze}
    \email{julien.ladeuze@ulb.be}
    \affiliation{Centre for Quantum Information and Communication, École polytechnique de Bruxelles, CP 165, Université libre de Bruxelles, 1050 Brussels, Belgium}
\author{Jeremie Roland}
    \email{jeremie.roland@ulb.be}
    \affiliation{Centre for Quantum Information and Communication, École polytechnique de Bruxelles, CP 165, Université libre de Bruxelles, 1050 Brussels, Belgium}

\begin{abstract}
Transducers (Belovs, Jeffery and Yolcu, 2024) are a quantum computing framework describing a quantum algorithm as a unitary converting an input state into a target state using a catalyst, an auxiliary vector that is left unchanged. They are a powerful tool in quantum algorithm design, especially in the context of quantum query complexity: feasible points of the (dual) adversary semidefinite program (SDP) directly translate into transducers and the optimal transduction complexity is equal to the adversary bound, \textit{i.e.} the Las Vegas complexity, which is known to characterize bounded-error quantum query complexity. Moreover, contrary to bounded-error algorithms, transducers compose exactly, which limits overheads due to controlling errors in algorithms constructed by composition. Constructing efficient, let alone optimal, transducers in terms of quantum query complexity nevertheless remains a hard task since it still requires solving the adversary SDP and constructing the unitary to obtain an explicit algorithm. 
In this paper, we show how using the symmetry group of state-conversion problems simplifies both steps. First, using a symmetrization argument, we prove an optimal catalyst can always be chosen covariant under a representation of the symmetry group. Second, we prove that the transducer intertwines two different representations of the group and can thus be chosen block diagonal in the isotypic decomposition of the Hilbert space. Using those methods, we then derive optimal transducers, with optimal constants, for different widely used quantum algorithmic primitives, such as unstructured search, amplitude amplification and amplitude estimation. Our approach extends previous work on the use of representation theory to compute adversary lower bounds (H{\o}yer, Lee, and {\v S}palek, 2007; Ambainis, Magnin, Roetteler and Roland, 2011) to the systematic construction of optimal algorithms.
\end{abstract}
\maketitle
\section{Introduction}
Transducers are a new model for quantum computing recently introduced by Belovs, Jeffery and Yolcu~\cite{belovsTamingQuantumTime2024}. Even though transducers equivalent to bounded-error quantum algorithms, transducers are particularly practical in the design of new quantum algorithms as they allow for errorless composition leading to new algorithms with better complexities for various problems such as the welded tree problem~\cite{belovsGlobalPhaseHelps2024}, electric flow sampling~\cite{apersElfsTransducersQuantum2026}, ground state energy estimation~\cite{jefferyOptimalQuantumAlgorithm2026} or Hamiltonian/Lindbladian simulations~\cite{chenTimeDependentHamiltonianSimulation2026,CGJLWZ26,CGWZ26lind,WY26}, as well as new complexity results~\cite{JW26qma}.

Informally, a transducer is a unitary $S$ acting on the direct sum of two Hilbert spaces $\mc H$ and $\mc L$, respectively, the public and private spaces. We say that it transduces a state $\ket{\xi}\in\mc H$ into a state $\ket{\tau}\in\mc H$, denoted $\ket{\xi}\overset{S}{\rightsquigarrow}\ket{\tau}$ if there exists a vector $\ket{v}\in\mc{L}$, called the catalyst, such that
\begin{equation}
     S\left(\ket\xi\oplus\ket v\right)=\ket\tau\oplus\ket{v}.
\end{equation}
The transduction complexity is then defined as $W=\norm{v}^2$, and Belovs, Jeffery and Yolcu~\cite{belovsTamingQuantumTime2024} show that from such a transducer, it is possible to construct a (standard) quantum algorithm that converts $\ket{\xi}$ into a state $\epsilon$-close to $\ket{\tau}$ using $O(1+W/\epsilon^2)$ calls to the unitary $S$.

In the context of quantum query complexity, Belovs, Jeffery and Yolcu~\cite{belovsTamingQuantumTime2024} define the notion of canonical transducers. Considering the state conversion problem $\{\ket{\xi_i}\in\mc H\}_{i\in\mc I}\mapsto\{\ket{\tau_i}\in\mc H\}_{i\in\mc I}$ for some input set $\mc I$, with access to an oracle from a set $\{O_i\in\mc M\}_{i\in\mc{I}}$, a canonical transducer is given by an input-independent unitary $S^\circ$ and a family of catalysts $\{\ket{v_i}\}_{i\in\mc I}$ living in the private space $\mc L=\mc W\otimes\mc M$ that perform the transduction
\begin{equation}
    S^\circ(\mathds{1}_{\mc H}\oplus \mathds{1}_{\mc W}\otimes O_i)(\ket{\xi_i}\oplus\ket{v_i})=\ket{\tau_i}\oplus\ket{v_i}\qquad\forall i\in \mc I.
\end{equation}
We see that here $S^\circ(\mathds{1}_{\mc H}\oplus \mathds{1}_{\mc W}\otimes O_i)$ plays the role of the transducer. The query complexity for this algorithm is then defined as the worst-case transduction complexity over all inputs: $W=\max_i \norm{v_i}^2.$

This model of computation is particularly promising as it is error-free and satisfies powerful composition properties, allowing to design algorithms for complex problems from simpler building blocks without incurring overheads due to controlling errors. It is therefore essential to derive methods for building efficient transducers and good algorithms for quantum algorithmic primitives, which is precisely the main goal of this manuscript.
In particular, we prove a large set of mathematical propositions leading to methods exploiting the symmetries of a problem for the design of transducers and apply them to find optimal transducers for several standard quantum algorithmic primitives.

In~\cref{Sec:symTra}, we derive consequences of the symmetries of the problem for the efficient construction of transducers and algorithms. First we properly define the symmetry group of a state conversion problem as the subgroup $\mc{G}$ of automorphisms of the input set that transforms the input states, target states and oracles using representations of the symmetry group. We then define two notions of symmetry for a transducer, the weak and strong covariance. A transducing algorithm is weakly covariant if there exists a unitary representation $\phi^{\mc L}$ of $\mc G$ that maps the catalyst $\ket{v_i}$ for the input $i$ onto $\ket{v_{g(i)}}$ for the input $g(i)$ for all elements $g$ of $\mc G$. The algorithm is strongly covariant if that representation is fixed by the problem parameters as described in~\cref{def:strongGov}. In~\cref{sec:symCat}, we derive theorems to help us construct the optimal catalyst. \cref{theorem:symmetryCat} states that any algorithm can be symmetrized into a weakly covariant algorithm with a lower or equal transduction complexity. This immediately implies that there exists an optimal catalyst that is weakly covariant (\cref{corr:optCat}). For many problems, constructing a transducing algorithm will require to impose strong covariance. However, in general there exist problems with no feasible strongly covariant algorithm, or with no optimal strongly covariant algorithm, as exemplified respectively by~\cref{prob:SearchCE1} and~\cref{prob:CosineCE}. \cref{Corr:SpaceRed} states that a weakly covariant catalyst $\ket{v_i}$ necessarily lies in the subspace stabilized by $\mathrm{Stab}(i)$, greatly reducing the possible space. In~\cref{sec:SymTra} we then derive consequences for the construction of the transducing algorithm. \cref{theorem:symmetryTrans} states that $S^\circ$ is an intertwiner and thus we derive~\cref{Corr:Niiiice} that necessarily establishes its block-diagonality on the irreducible representations of the group.

With those results, we are able to construct optimal transducers for widely used quantum algorithmic primitives in~\cref{sec:app}. In particular, we tackle
\begin{itemize}
    \item $\mathrm{Search}^N_M$ and $\mathrm{Search}^N_{\geq M}$: Unstructured search on $N$ elements, with $M$ and at least $M$ marked elements, respectively (\cref{sec:searchM} and~\ref{sec:searchgeqM}). These problems take as input a bit-string $x\in\{0,1\}^N$ and convert the uniform superposition over all indices into the uniform superposition over the marked indices in $x$:
    \begin{equation}
        \ket{\xi_x}=N^{-1/2}\sum_i\ket{i}\longmapsto\ket{\tau_x}=\ket{x}\cq\abs{x}^{-1/2}\sum_{i:x_i=1}\ket{i}.
    \end{equation}
    For search $\mathrm{Search}^N_M$ we consider that $\abs{x}=M$, whereas for $\mathrm{Search}^N_{\geq M}$ we consider that $\abs{x}\geq M$. The symmetry group of those problems is $\mc S_n$. To find the catalyst we first specify the adversary bound to the problem. Then for $\mathrm{Search}^N_M$ we assume that the catalyst is weakly covariant, as allowed by~\cref{corr:optCat}, and then show that among all possible representations, the strong covariance is optimal. Finding the optimal catalyst from there is a simple two-dimensional optimization problem. For $\mathrm{Search}^N_{\geq M}$ we directly assume that the catalyst is strongly covariant and are able to conclude that this is optimal as the resulting complexity matches that of $\mathrm{Search}^N_M$, which constitutes a lower bound that is therefore tight.
    More precisely, we find
    \begin{repeatProp2}{prop:searchM}{prop:searchgeqm}
\begin{equation}
    W\left(\mathrm{Search}^N_{M}\right)=W\left(\mathrm{Search}^N_{\geq M}\right)=\sqrt{\frac{N-M}{4M}}.
\end{equation}
\end{repeatProp2}
Using~\cref{theorem:symmetryTrans}, we are then able to derive explicit optimal transducing algorithms for these primitives (note that their complexity is optimal up to the constant: there is no Big-$O$ in the obtained complexities).
\item $\mathrm{Amp}_\eps$ and $\mathrm{Amp}_{\geq \eps}$: Amplitude amplification with a promise $\eps$ or at least $\eps$ for the overlap on a marked subspace (\cref{sec:AmpFix} and~\ref{sec:AmpPin}). These problems take as input a vector $\ket{\xi}$ and require to project it onto some marked subspace accessed by a projector $\Pi_m$:
    \begin{equation}
        \ket{\xi}\longmapsto\ket{\tau_{\ket{\xi}}}=\frac{\Pi_m\ket{\xi}}{\norm{\Pi_m\ket{\xi}}}.
    \end{equation}
    For $\mathrm{Amp}_\eps$ we have the promise that $\norm{\Pi_m\ket{\xi}}=\eps$, whereas for $\mathrm{Amp}_{\geq \eps}$, the promise is $\norm{\Pi_m\ket{\xi}}\geq\eps$. The symmetry group of those problems is the subgroup of $U(\mc H)$ that commutes with $\Pi_m$. The procedure to find the catalyst is similar to the one for unstructured search: we derive a lower bound for the transduction complexity by considering a subset of inputs, then build a strongly covariant catalyst and show that it saturates the lower bound and is therefore optimal. We find
    \begin{repeatProp2}{prop:AmpFix}{prop:AmpMin}
\begin{equation}
    W\left(\mathrm{Amp}_\eps\right)=W\left(\mathrm{Amp}_{\geq \eps}\right)=\frac{\sqrt{1-\eps^2}}{2\eps}.
\end{equation}
\end{repeatProp2}
We are able then able to derive optimal transducing algorithms using~\cref{theorem:symmetryTrans}. Note that for $\mathrm{Amp}_{\geq \eps}$, the construction is more technical and requires functional analysis as it involves an infinite dimensional Hilbert space, namely $\mathbb{C}\oplus L^2(\mathbb{R})$.
\item $\mathrm{Est}_{g}$: Amplitude estimation (\cref{sec:AmpEstimation}). In this problem, we have access to an oracle reflecting through a state $\ket{\theta}=\cos\theta\ket{0}+\sin\theta\ket{1}$, where $\theta$ is unknown, and the goal is to convert a fixed state $\ket{\xi_0}$ into a state $\ket{f_\theta}$ that depends on $\theta$ (intuitively, the state could encode information about $\theta$):
\begin{equation}
    \ket{\xi_0}\longmapsto\ket{f_\theta}.
\end{equation}
We impose additionally that $\ket{f_\theta}$ is $\pi$-periodic, and that the inner product $\braket{f_\theta|f_{\theta'}}$ is a function $g(\theta-\theta')$. This definition might seem peculiar but subsumes a large class of problems, in particular the more standard problem of determining an estimate of the amplitude $\cos\theta$, as argued at the start of the section. To derive the optimal catalyst for this particular problem, we identify the symmetry group of the problem as $U(1)$, we impose (without loss of generality, as follows from our results) that the optimal catalyst is weakly covariant and that it has a component for all characters of the group. Finding each of those components individually then reduces to a two-dimensional optimization problem that we can directly solve, leading to
\begin{repeatProp}{prop:ampEst}
    \begin{equation}
        W(\mathrm{Est}_{g})=\sum_{n\in2\mathbb Z+1}\abs{\hat h_n},\quad\text{with }h(\omega)=\lim_{\omega'\rightarrow\omega}\frac{1-g(\omega')}{2\sin\omega'}=\sum_{n\in2\mathbb Z+1}\hat h_ne^{i\omega n}
    \end{equation}
\end{repeatProp}
The relation with the Fourier series emerges as the characters of $U(1)$ form the Fourier basis.
\end{itemize}

Note that while the bounded-error quantum query complexity of such primitives was already known, the novelty of this work lies in the explicit expression of transducers doing the transduction optimally, and the machinery developed to construct catalysts efficiently, which can be reused to tackle any symmetric problem.

The algorithms presented in this work are optimal or efficient in terms of quantum query complexity, but we do not study time or space complexity: as constructed, the transducing algorithms that we obtain might involve unitaries that are hard to implement. We leave as open problem for future work how to make these transducing algorithms time or space efficient. We note that the symmetrization process itself, while a key element to construct algorithms that are optimal in terms of query complexity, requires to enlarge the Hilbert space, possibly to infinite dimension as symmetrizing over a Lie group requires to introduce a Hilbert space that is isomorphic to this Lie group. For that reason, amplitude estimation $\mathrm{Est}_{g}$ requires a register $\ell^2(\mathbb{N})\otimes \mathbb{C}^2$ and amplitude amplification $\mathrm{Amp}_{\geq \eps}$ a register $\mathbb{C}\oplus L^2(\mathbb{R})$. All those Hilbert spaces are isomorphic to $\ell^2(\mathbb{N})$.
An open question concerns the conversion of the transducer over those infinite Hilbert spaces into practical quantum algorithms. Either we will need to consider quantum computing architectures able to represent those Hilbert spaces faithfully (there exists many physical systems realizing infinite dimensional Hilbert spaces, such as bosonic systems where a mode is precisely described by Hilbert space $\ell^2(\mathbb{N})$), or we could design techniques to embed the necessary unitaries inside finite dimensional Hilbert spaces, either exactly or via truncation, which could result in trade-offs between the space complexity and the error introduced by the truncation.

\newpage
\section{Preliminary}
In this section we present the mathematical notions used in the manuscript. 
\subsection{State Conversion Problems and Their Equivalence}\label{sec:ScpEq}
The main object of study of this manuscript is the quantum query complexity of \emph{state conversion problems}~\cite{LMRSS11,AMRR11,Bel15var}. Given an input set $\mc{I}$, the goal is to convert an input state $\ket{\xi_i}$ into a target state $\ket{\tau_i}$ for all indices $i\in\mc I$, with access to a unitary oracle $O_i$ that encodes information about $i$ using a single algorithm that has no knowledge of $i$ beyond $O_i$. This generalizes the more standard problem of computing a function $f$, where $\ket{\xi_i}=\ket{0}$ and $\ket{\tau_i}=\ket{f(i)}$). Formally:
\begin{definition}[State conversion problem]\label{def:state-con}
    A \emph{state conversion problem} $P$ (or simply a \emph{problem}) over an input set $\mc{I}$ is a triple $P=(\mc{X}, \mc{T}, \mc{O})$ 
    where $\mc{X}, \mc{T}$ and $\mc{O}$ are, respectively, the indexed families of input states, target states and oracles:
        \begin{equation}
            \mc{X} \cq \{\ket{\xi_i}\}_{i\in \mc{I}}\quad 
            \mc{T} \cq  \{\ket{\tau_i} \}_{i\in \mc{I}}\quad
            \mc{O} \cq \{O_i \}_{ i \in \mc{I}}\qquad\text{with}\qquad\ket{\xi_i},\ket{\tau_i}\in\mc{H}\quad \norm{\xi_i}=\norm{\tau_i}=1\quad O_i\in U(\mc{M}),
        \end{equation}
        for some Hilbert spaces $\mc{H}$ and $\mc{M}$, with $U(\mc{M})$ denoting unitary group on $\mc{M}$.\footnote{Note that transducers, as defined in~\cref{Subsec:PreTra}) will act on the orthogonal direct sum $\mc H\oplus\mc W\otimes \mc M$, where $\mc W$ is some work space.}
\end{definition}

In the context of quantum query complexity, input-independent unitaries are free, so applying a fixed unitary to $\mc{X}$, $\mc{T}$, or $\mc{O}$ does not change the structure of the problem, in particular, it does not change its query complexity. We therefore define the following equivalence relation.
\begin{definition}[Unitary equivalence] Let $P=(\mc{X}, \mc{T}, \mc{O})$ and $P'=(\mc{X}', \mc{T}', \mc{O}')$ be two problems over the same input set $\mc{I}$ and Hilbert spaces $\mc H$ and $\mc M$. We say that $P$ and $P'$ are \emph{unitarily equivalent}, denoted $P\cong P'$, if
    \begin{equation}
       \exists U^\xi,U^\tau\in U(\mc{H}),~\exists U^L,U^R\in U(\mc{M}):\forall i\in\mc{I}:\begin{cases}\ket{\xi_i'}=U^\xi\ket{\xi_i}\\
        \ket{\tau_i'}=U^\tau\ket{\tau_i}\\
        O_i'=U^L O_iU^{R\dagger}\\
        \end{cases}.
    \end{equation}
\end{definition}
 Indeed, in that case an algorithm for $P$ can be transformed into an algorithm for $P'$ by pre-processing and post-processing each oracle call with $U^R$ and $U^{L\dagger}$, and pre-processing and post-processing the whole algorithm with $U^{\xi\dagger}$ and $U^\tau$
 
This equivalence relation is thus induced by the group action of $U(\mc H)^2\times U(\mc M)^2$. An algorithm for $P$ can be transformed into an algorithm solving $P'$ by applying $U^{\xi\dagger}$ at the start, then $U^\tau$ at the end, and $U^{R}$ and $U^{L\dagger}$ respectively before and after each query (see \cref{fig:unitary-equiv}).
We can now introduce invariants under this equivalence relation. First, we define the input and target Gram matrices (or kernels if $\mc{I}$ is infinite) $X$ of $\mc{X}$ and $T$ of $\mc T$,
\begin{equation}
    X(i,j)\cq\braket{\xi_i|\xi_j}\qquad T(i,j)\cq\braket{\tau_i|\tau_j}.
\end{equation}
Second, a natural candidate on the oracle side is the relative oracle operator-valued kernel
\begin{equation}
    K:\mc{I}\times\mc{I}\to U(\mc{M}): (i,j)\mapsto K(i,j)=O_i^\dagger O_j.
\end{equation}
However, although $K$ is invariant under left multiplication, since $U^L$ cancels out, it is not invariant under the right multiplication $O_i\mapsto O_iU^{R\dagger}$
\begin{equation}
    \left[O'_i=U^L O_iU^{R\dagger}\right]\Rightarrow \left[K'(i,j)={O'}_i^{\dagger}O_j'=U^RO_i^\dagger U^{L\dagger}U^LO_jU^{R\dagger}=U^RO_i^\dagger O_jU^{R\dagger}=U^R K(i,j)U^{R\dagger}\right].
\end{equation}
We instead consider $[K]$ the orbit of $K$ under simultaneous conjugation
\begin{equation}\label{Eq:EqKer}
    [K]\cq \left\{(i,j)\longmapsto U^R K(i,j)U^{R\dagger}:U^R\in U(\mc M)\right\}.
\end{equation}
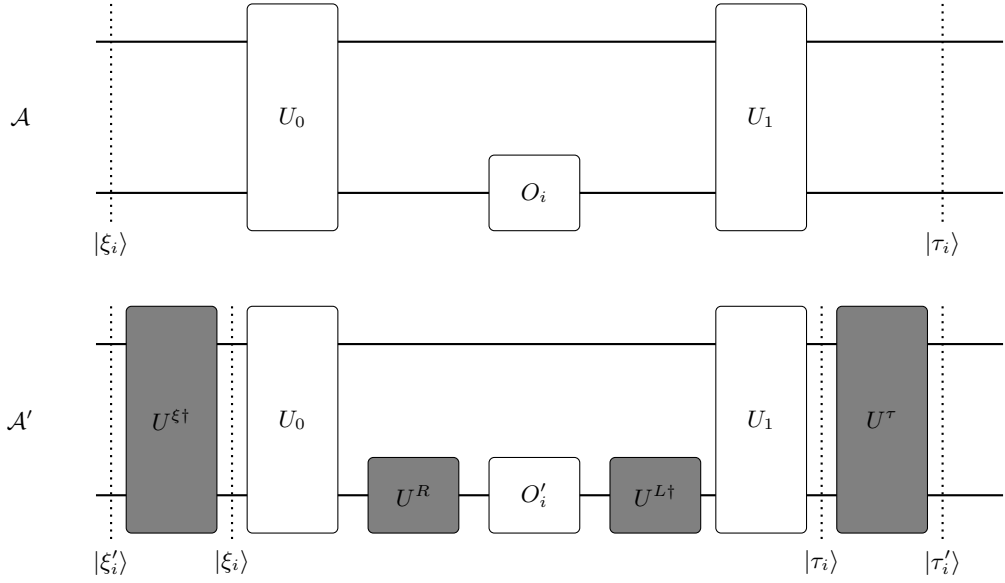
\begin{figure}
    \begin{tikzpicture}
    \tikzset{
    gate/.style={
        draw=black,
        fill=white,
        rounded corners=2pt,
        minimum width=1.2cm,
        minimum height=3cm
    },
    oracle/.style={
        draw=black,
        fill=white,
        rounded corners=2pt,
        minimum width=1.2cm,
        minimum height=1cm
    }
}
\draw[thick] (0,4) -- (12,4);
\draw[thick] (0,2) -- (12,2);

\node at (-1,3) {$\mathcal{A}$};

\node[gate] at (2.6,3) {$U_0$};
\node[oracle] at (5.8, 2) {$O_i$};
\node[gate] at (8.8, 3) {$U_1$};

\draw[dotted, thick] (0.2, 4.5) -- (0.2, 1.6) node[below] {$\ket{\xi_i}$};
\draw[dotted, thick] (11.2, 4.5) -- (11.2, 1.6) node[below] {$\ket{\tau_i}$};


\draw[thick] (0,0) -- (12,0);
\draw[thick] (0,-2) -- (12,-2);

\node at (-1,-1) {$\mathcal{A'}$};

\node[gate, fill=gray] at (1,-1) {$U^{\xi\dagger}$}; 
\node[gate] at (2.6,-1) {$U_0$};
\node[oracle, fill=gray] at (4.2, -2) {$U^R$};
\node[oracle] at (5.8, -2) {$O'_i$};
\node[oracle, fill=gray] at (7.4, -2) {$U^{L\dagger}$};
\node[gate] at (8.8, -1) {$U_1$};
\node[gate, fill=gray] at (10.4, -1) {$U^\tau$};

\draw[dotted, thick] (0.2, 0.5) -- (0.2, -2.6) node[below] {$\ket{\xi'_i}$};
\draw[dotted, thick] (1.8, 0.5) -- (1.8, -2.6) node[below] {$\ket{\xi_i}$};

\draw[dotted, thick] (9.6, 0.5) -- (9.6, -2.6) node[below] {$\ket{\tau_i}$};

\draw[dotted, thick] (11.2, 0.5) -- (11.2, -2.6) node[below] {$\ket{\tau'_i}$};

\end{tikzpicture}
\caption{\centering An example of two algorithms solving two problems that are unitarily equivalent. The algorithm $\mathcal{A}$ solves the problem $\ket{\xi_i}\mapsto \ket{\tau_i}$ whereas the algorithm $\mathcal{A}'$ solves the problem $\ket{\xi'_i} \mapsto \ket{\tau'_i}$. This example shows that a single oracle call is needed with the application of two independent unitaries $U_0$ and $U_1$.\label{fig:unitary-equiv}}
\end{figure}

Note that the same unitary $U^R$ acts simultaneously on every entry of the kernel.
This leads to the following proposition,
\begin{proposition}\label{Prop:EqProb}
    Let $P,P'$ be two state conversion problems over the same input set and Hilbert spaces. Assume that $\mc H$ is finite dimensional, or more generally that $\dim\overline{\operatorname{Span}}(\ket{\xi_i})^\perp=\dim\overline{\operatorname{Span}}(\ket{\xi'_i})^\perp$ and similarly for $\ket{\tau_i}$ and $\ket{\tau_i'}$. Then 
    \begin{equation}
        P\cong P'\iff\begin{cases}
            X=X'\\T=T'\\ [K]=[K']
        \end{cases}.
    \end{equation}
     Consequently, the class of unitarily equivalent state conversion problems is determined by and determines the triple $\left(X,T,[K]\right)$.
\end{proposition}
The forward implication is immediate for $X$ and $T$ and follows from the definition of~\cref{Eq:EqKer} for $[K]$. The converse is more involved, so a full proof of this proposition is delayed to~\cref{App:ProofPropEqProb}. Note that the condition $\dim\overline{\mathrm{Span}}\{\ket{\xi_i}\}^\perp=\dim\overline{\mathrm{Span}}\{\ket{\xi_i'}\}^\perp$ is a technicality that has no practical consequence, as it is trivially satisfied for finite dimensional Hilbert spaces, and can always be achieved in the infinite-dimensional case by a simple modification of $P$ and $P'$: it suffices to embed the Hilbert space $\mc{H}$ into $\mc{H}\otimes \ell^2(\mathbb{N})$, and to replace states $\ket{\xi_i}$ by $\ket{\xi_i}\otimes\ket{0}$ where $\ket{0}$ is some default state in $\ell^2(\mathbb{N})$ (similarly for $\ket{\xi_i'}$, $\ket{\tau_i}$ and $\ket{\tau_i'}$).

\subsection{Quantum Query Complexity}

Having defined state conversion problems, we will define two models of computation to solve such problems: quantum query algorithms (in this section), as well as transducers (in the next sections). We will also define the corresponding notions of complexity, quantum query complexity and transduction complexity, and show that these notions are closely related.

\begin{definition}[Quantum Query Algorithm] A \emph{quantum query algorithm} $\mc A$, or in short, a query algorithm, for a state conversion problem $P$ is a sequence of unitary $\left\{U_k\right\}_{0\leq k\leq n}$ with $U_k\in U\left(\mc H\oplus\mc W\otimes \mc M\right)$ such that
\begin{equation}
    \ket{\tau_i}=\prod_{k=1}^n\left(U_k\left(\mathds{1}_{\mc{H}}\oplus\mathds{1}_{\mc W}\otimes O_i\right)\right)U_0\ket{\xi_i}\qquad\forall i\in \mc I
\end{equation}
    The query complexity of $\mc A$ is the number of oracle calls $n$.
\end{definition}

Given the definition of a state conversion problem and of a query algorithm; quantum query complexity, in short query complexity for this paper, is then defined as follows:

\begin{definition}[Exact query complexity]
    The exact query complexity of a state conversion problem $P\cong(X,T,[K])$, denoted $Q(P)$, is the minimum number of oracle calls a query algorithm needs to perform the conversion. If no exact algorithm exists, then $Q(P)=\infty$.
\end{definition}

Using the argument presented in~\cref{sec:ScpEq}, one can note that the query complexity is well defined on equivalence classes of problems, justifying the equivalence $P\cong(X,T,[K])$. In many cases, however, exact conversion is impossible or extremely expensive, and one therefore usually considers the bounded-error query complexity.
\begin{definition}[Bounded-error query complexity]
    The bounded-error query complexity of a state conversion problem $P\cong(X,T,[K])$ written $Q_\eps(P)$ is the minimum number of oracle calls a query algorithm needs to perform the conversion approximately, with error at most $\eps\geq0$
    \begin{equation}
        Q_\eps(P)=\inf_{\substack{\{\ket{\tau'_i}\in\mc H\}_{i\in\mc I}\\\norm{\ket{\tau_i'}-\ket{\tau_i}}\leq\eps~\forall i\\T'(i,j)=\braket{\tau'_i|\tau'_j}}}Q(X,T',[K]).
    \end{equation}
\end{definition}
This definition is hard to use directly when bounding the complexity, justifying the introduction of a quantity called the adversary bound~\cite{Amb02,hoyerNegativeWeightsMake2007a}, that is the objective value of an optimization program, that characterizes this complexity. We present here what is historically called the dual of the adversary bound~\cite{hoyerNegativeWeightsMake2007a,LMRSS11}, as in this paper we will mainly use this formulation since its feasible points can be converted into transducers~\cite{belovsTamingQuantumTime2024}.
\begin{definition}[Adversary bound]
    The adversary bound of a state-conversion problem $P$, written $\mathrm{Adv}(P)$, is the optimal value of the following optimisation program
    {\small\begin{equation}\label{Eq:DefAdv}
        \mathrm{Adv}(P)=\inf_{\ket{v_i}\in\mathcal{L}}\sup_{i\in \mc{I}} \norm{v_i}^2 \quad \text{s.t.}\quad\left({\bra{\xi_i}}\oplus\bra{v_i}(\mathds{1}_\mathcal{W}\otimes O_i^\dagger)\right) \left(\ket{\xi_j}\oplus(\mathds{1}_\mathcal{W}\otimes O_j)\ket{v_j}\right)=\left({\bra{\tau_i}}\oplus{\bra{v_i}}\right)\left(\ket{\tau_j}\oplus\ket{v_j}\right)\quad\forall i,j\in\mc{I},
    \end{equation}}
    where $\mathcal{L}=\mc W\otimes\mc M$, with the infimum also taken over the Hilbert space $\mc{W}$, or equivalently
    \begin{equation}
    \mathrm{Adv}(P)=\inf_{\ket{v_i}\in\mathcal{L}}\sup_{i\in \mc{I}} \norm{v_i}^2 \qquad \text{s.t.} \quad X(i,j)-T(i,j)=\braket{v_i|v_j}-\braket{v_i|\mathds{1}\otimes K(i,j)|v_j}\quad\forall i,j\in\mc{I}.
    \end{equation}
\end{definition}
Once again, the adversary bound is well defined on equivalence classes of problems, as $U^\xi, U^\tau$ and $U^L$ automatically cancel in~\cref{Eq:DefAdv} since they preserve the inner products, and $U^R$ can be absorbed by $\ket{v_i}\mapsto(\mathds{1}_{\mc W}\otimes U^R)\ket{v_i}$. Belovs and Yolcu~\cite{belovsOneWayTicketVegas2023} proved that the adversary bound is equal to the Las Vegas query complexity of a problem, that represents the expected number of queries for a stochastic algorithm. Similarly to query complexity, we define the bounded-error adversary bound:
\begin{definition}[Bounded-error adversary bound]
    The bounded-error adversary bound is the minimum over all approximate conversions  with error at most $\eps\geq0$
    \begin{equation}
    \mathrm{Adv}_\eps(P)=\inf_{\substack{\{\ket{\tau'_i}\in\mc H\}_{i\in\mc I}\\\norm{\ket{\tau_i'}-\ket{\tau_i}}\leq\eps~\forall i\\T'(i,j)=\braket{\tau'_i|\tau'_j}}}\mathrm{Adv}(X,T',[K]).
    \end{equation}
\end{definition}
Extending the state conversion bounds of~\cite{LMRSS11}, Belovs and Yolcu~\cite{Bel15var,belovsOneWayTicketVegas2023} proved that the adversary bound characterizes the quantum query complexity of a problem, up to a $1/\eps^2$ factor:
\begin{equation}
    \mathrm{Adv}_\eps(P)\leq Q_\eps(P)\leq O\left(\frac{\mathrm{Adv}_{0}(P)}{\eps^2}\right).
\end{equation}
\subsection{Transducers}\label{Subsec:PreTra}
Transducers belong to a line of work that emerges from the adversary bound~\cite{Amb02}, later expressed as a semidefinite program (SDP)~\cite{hoyerNegativeWeightsMake2007a}. Span programs~\cite{RS12} were then developed in an effort to convert feasible points of the dual of the adversary bound for functions into algorithms, a connection later shown to yield nearly tight algorithms~\cite{Rei09}; learning graphs~\cite{Bel12sp} brought a complementary constructive approach. The scope soon widened from Boolean functions to arbitrary state conversion~\cite{LMRSS11} and to general input oracles~\cite{Bel15var}. A parallel effort sought to account for time rather than queries alone, through electric networks, quantum walks and span programs~\cite{Bel13elec,JZ23,CJOP20}, culminating in a thrifty, but not exact, composition of quantum subroutines~\cite{Jef22}. More recently, the adversary bound has also been equated with the Las Vegas query complexity of the problem~\cite{belovsOneWayTicketVegas2023}; transducers~\cite{belovsTamingQuantumTime2024} unify the two threads, obtaining both thriftiness and exactness in time.

Formally, we define:
\begin{definition}[Transducers and catalysts]
    The transducer model of computation represents an algorithm by a pair $\mc A=(S,\ket{v})$ where $S\in U(\mc{H}\oplus\mc{L})$ is the \emph{transducer}, a unitary operator acting on the direct sum of $\HS$ and $\mathcal{L}$, the public and private spaces and $\ket{v}\in\mc{L}$ is the \emph{catalyst}. The transducer maps the input state $\ket{\xi}\in \HS$ onto the target state $\ket{\tau} \in \HS$, denoted $\ket\xi\overset S\rightsquigarrow\ket\tau$ with the \emph{help} of the catalyst. The catalyst is accessed through the direct-sum structure and is not modified by $S$. Explicitly, we have
\begin{equation}
    S: \HS \oplus \mathcal{L} \rightarrow \HS \oplus \mathcal{L}: \ket{\xi} \oplus \ket{v} \mapsto \ket{\tau} \oplus \ket{v}.
\end{equation}
The transduction complexity is then 
\begin{equation}
    W\left(\mc A\right)=\norm{v}^2.
\end{equation}
\end{definition}
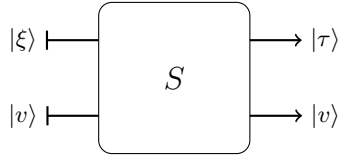
\begin{figure}
   \begin{tikzpicture}

  \draw[rounded corners=5pt]
    (1, -1) rectangle (3, 1);
  \node[font=\large] at (2, 0) {$S$};

  \node (input_state)    at (0,  0.5) {$\ket{\xi}$};
  \node (input_catalyst) at (0, -0.5) {$\ket{v}$};
  \node (output_state)   at (4,  0.5) {$\ket{\tau}$};
  \node (output_catalyst) at (4, -0.5) {$\ket{v}$};

  \draw[|-, thick] (input_state)    -- (1,  0.5);
  \draw[|-, thick] (input_catalyst) -- (1, -0.5);
  \draw[->, thick] (3,  0.5) -- (output_state);
  \draw[->, thick] (3, -0.5) -- (output_catalyst);
    
\end{tikzpicture}
    \caption{Graphical representation of a transducer applied to some vector $\ket{\xi}\oplus \ket{v}$ to produce a vector $\ket{\tau}\oplus \ket{v}$.}
    \label{fig:transducer}
\end{figure}

In this manuscript, we will only consider transducers in the context of quantum query complexity, where one can focus on so-called \textit{canonical transducers} (see \cref{fig:canonical-transducer}). In this case, the transducers $\{S_i\}_{i\in\mc I}$ and the catalysts $\{\ket{v_i}\}_{i\in\mc I}$ are families indexed by the input set, and the transducer can be decomposed as the product of an input-independent unitary and an oracle call, acting only on a subspace of the private space. 

\begin{definition}[Canonical transducer]
     Let $\HS$ be the public space, $\mathcal{L} = \mathcal{W} \otimes \mathcal{M}$ the private space, with $\mc{M}$ the space where the oracle call is applied, and $\mc{W}$ an additional register. Given the problem $P$ defined over the input set $\mc{I}$, with $\ket{\xi_i},\ket{\tau_i}\in\mc{H}$ and $O_i\in U(\mc{M})$, a canonical transducer is the family of unitary operators $S_i\in U(\mc{H}\oplus\mc{L})$ such that 
     \begin{equation}
          S_i = S^\circ(\mathds{1}_\HS\oplus\mathds{1}_{\mathcal{W}} \otimes O_i),
     \end{equation}
     with an input-independent unitary operator $S^\circ\in U(\HS \oplus \mathcal{L})$ that realizes the conversion 
     \begin{equation}
         S_i (\ket{\xi_i}\oplus \ket{v_i})=\ket{\tau_i}\oplus\ket{v_i}\quad\forall i\in\mc I.
     \end{equation}
     Formally, we can represent a transducing algorithm for a problem $P$ by the pair $\mc{A}=(S^\circ,\{\ket{v_i}\}_{i\in\mc{I}})$. The transduction complexity for a specific input is defined as 
     \begin{equation}
         W_i(\mc A)=\norm{v_i}^2,
     \end{equation}
     and for the full input set as
     \begin{equation}
         W(\mc A)=\sup_i W_i(\mc A).
     \end{equation}
\end{definition}
In turn, wec can define the transduction complexity of a problem.
\begin{definition}{Transduction complexity of a state conversion problem}.
Let $P=(\mc{X}, \mc{T}, \mc{O})$ be a state conversion problem. The transduction complexity of $P$, denoted $W(P)$ is defined as the infimum over all transducing algorithms for $P$.
\end{definition}

Given a family of catalysts, by the same argument\footnote{Note that the argument here requires the same condition concerning $\dim\overline{\operatorname{Span}}\left(\ket{\xi_i}\oplus\left(\mathds{1}_{\mc{W}}\otimes O_i\right)\ket{v_i}\right)^\perp$ as discussed as the end of~\cref{sec:ScpEq}, which can be satisfied by extending $\mc{L}$ with an additional $\ell^2(\mathbb{N})$ register.} as in the proof of~\cref{Prop:EqProb}, there exists an input-independent unitary $S^\circ$ converting $\ket{\xi_i}\oplus(\mathds{1}_{\mc W}\otimes O_i)\ket{v_i}$ into $\ket{\tau_i}\oplus\ket{v_i}$ for all $i\in\mc I$ if and only if the corresponding Gram matrices match, that is, if
\begin{equation}
    \left(\bra{\xi_i}\oplus\bra{v_i}(\mathds{1}_{\mc W}\otimes O^\dagger_i)\right)\left(\ket{\xi_j}\oplus(\mathds{1}_{\mc W}\otimes O_j)\ket{v_j}\right)=\left(\bra{\tau_i}\oplus\bra{v_i}\right)\left(\ket{\tau_j}\oplus\ket{v_j}\right)\quad\forall i,j\in\mc I,
\end{equation}
which is precisely the condition of the adversary bound in~\cref{Eq:DefAdv}. This observation implies the following fundamental relation~\cite{belovsTamingQuantumTime2024}, connection quantum query complexity (via the adversary bound $\operatorname{Adv}(P)$) to transduction complexity:
\begin{proposition}
For any state conversion problem $P$, we have
\begin{equation}
    \operatorname{Adv}(P)=W(P).
\end{equation}
\end{proposition}

\begin{figure}[h]
    \begin{tikzpicture}


\tikzset{
    box/.style={draw, rounded corners=5pt},
    state/.style={thick},
}

\node at (-0.4,0.5) {$\ket{\xi_i}$};
\node at (-0.4,-0.5) {$\ket{v_i}$};
\node at (6.9,0.5) {$\ket{\tau_i}$};
\node at (6.9,-0.5) {$\ket{v_i}$};

\node[box, minimum width=1.5cm, minimum height=1cm]
    at (1.75,-0.5) {$\mathds{1}_{\mathcal{W}}\otimes O_i$};

\node[box, minimum width=2cm, minimum height=2cm]
    at (4.5,0) {$S^\circ$};

\draw[|-, state] (0,-0.5) -- (1,-0.5);

\draw[|-, state] (0,0.5) -- (3.5,0.5);
\draw[state] (2.5,-0.5) -- (3.5,-0.5);

\draw[->,state] (5.5,0.5) -- (6.5,0.5);
\draw[->,state] (5.5,-0.5) -- (6.5,-0.5);

\end{tikzpicture}

    \caption{\centering
        Graphical representation of a canonical transducer applied to some vectors $\ket{\xi_i} \oplus \ket{v_i}$ to produce a vector $\ket{\tau_i} \oplus \ket{v_i}$ regarding a certain input $i\in \mathcal{I}$.\label{fig:canonical-transducer}
}
    
\end{figure}
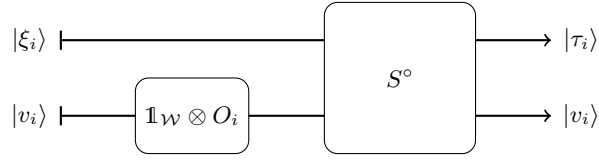

\subsection{Elements of Representation Theory}\label{sec:PreSym}
Symmetries play a major role in our construction, mostly through representation theory: they constrain the form of optimal catalysts and offer a block diagonalization of the transducer. This section presents the notations and standard results we need in the form in which we use them, without proof. For further references, we direct readers to~\cite{serreLinearRepresentationsFinite1977} for the informations concerning finite group,~\cite{brockerRepresentationsCompactLie1985} for Lie groups,~\cite{simonRepresentationsFiniteCompact1995} specifically for the questions concerning infinite dimensional representations and multiplicity spaces and~\cite{follandCourseAbstractHarmonic1994} for the construction of the Group Hilbert space.

Throughout this paper, $\mc{G}$ represents either a finite group or a compact Lie group, and all spaces are complex Hilbert spaces. 
\paragraph{Representations}
Representations map group elements and their action on a vector space. Formally, we first need to define the map between groups, called homomorphisms, and then specify them for representation.
\begin{definition}[Group homomorphism] A group homomorphism from a group $\mc{G}$ equipped with a composition law $\circ$, to a group $\mc{H}$ equipped with $*$, is a function $\phi:\mc{G}\mapsto\mc{H}$ that preserves the composition
\begin{equation}
    \phi(g\circ g')=\phi(g)*\phi(g')\qquad\forall g,g'\in\mc{G}.
\end{equation}
We denote $\operatorname{Hom}(\mc G,\mc H)$ the set of all homomorphisms from $\mc G$ to $\mc H$.
\end{definition}
Then we can define a representation the following way:
\begin{definition}[Representation of $\mc G$]
    A representation of $\mc{G}$ on a Hilbert space $V$  is a group homomorphism
    \begin{equation}
        \phi : \mc{G} \rightarrow {GL}(V)
    \end{equation}
    where $GL(V)$ is the group of bounded endomorphisms of $V$, invertible under composition. For finite Hilbert spaces, $GL(V)$ is simply the general linear group. $V$ is called the representation space. We denote $\operatorname{Hom}(\mc G,GL(V))$ the set of all representations.
\end{definition}
For instance, the simplest representation is the \emph{trivial representation} is the homomorphism $\phi:g\mapsto\phi_g=\mathds{1}_V~\forall g\in\mc G$. Representations $\phi$ and $\phi'$ acting respectively on $V$ and $V'$ can be combined in various ways, and two of the most prominent, we will use extensively in this paper are:
\begin{itemize}
    \item the \emph{direct sum} $(\phi\oplus\phi')_g=\phi_g\oplus\phi'_g$ on $V\oplus V'$,
    \item the \emph{tensor product} $(\phi\otimes\phi')_g=\phi_g\otimes\phi'_g$ on $V\otimes V'$.
\end{itemize}
In this manuscript, we will only consider unitary representations, that is, homomorphism not to the whole $GL(V)$ group, but only to the unitary subgroup $U(V)$. 
Such representations then belong to $\operatorname{Hom}(\mc G,U(V))$. Further, we will also impose that the representations are (strongly) continuous for compact Lie groups, as this assumption will be required for the Haar averaging in~\cref{proof:SymmetryCat} and for complete reducibility~\cref{Th:CompletRed}.
\paragraph{Square summable function over $\mc G$ and regular representation}
In~\cref{proof:SymmetryCat} we will introduce $\mc{H}_{\mc{G}}\cq L^2(\mc G,dg)$, the vector spaces of square summable functions over the group, or in short, the group Hilbert space, with $dg$ a fixed Haar measure on $\mc{G}$, with total weight $\mu(\mc{G})$. We denote by $\{\ket{g}:g\in\mathcal G\}$ the canonical group-element basis, satisfying
\begin{equation}\label{eq:canGroupBasis}
    \braket{g|h}=\delta_{\mc{G}}(g^{-1}h),
\end{equation}
where $\delta_{\mc{G}}$ is the Dirac delta on the group topology\footnote{In short for any function $f$ on $\mc{G}$, $\int_{\mc{G}}f(g)\delta_{\mc G}(g)dg=f(e)$ where $e$ denotes the identity.}.
In that context the left regular representation $L_g$ is a unitary representation acting as the left translation/composition
\begin{equation}
    L_gf(h)\cq f(g\circ h),\qquad L_g\ket{h}=\ket{g\circ h}.
\end{equation}
We could have similarly defined the right regular representation.
For finite groups, we consider $\int_{\mc G}dg$ as the sum $\sum_{\mc G}$ with $\mu(\mc G)=\abs{G}$ and $\delta_{\mc{G}}$ is replaced by the Kronecker delta.
\paragraph{Intertwiners} 
Intertwiners are linear maps commuting with the group action. They are a central part of this manuscript as in~\cref{theorem:symmetryTrans} we will prove that the input-independent unitary is an intertwiner.
\begin{definition}[Intertwiner and equivalence]
    Let $\phi^V$ and $\phi^W$ be two representations of $\mc{G}$ acting respectively on $V$ and $W$. A bounded linear map $S:V\mapsto W$ is an \emph{intertwiner} if
    \begin{equation}
        S\phi^V_g=\phi^W_g S\qquad\forall g\in\mc{G}.
    \end{equation}
    An invertible intertwiner is an \emph{equivalence}. If such an equivalence exists, $\phi^V$ and $\phi^W$ are \emph{equivalent}, written $\phi^V\cong\phi^W$, and by extension $V$ and $W$ are also equivalent, written $V\cong_{\mc{G}} W$.
\end{definition}
Shortly, two equivalent representations describe the same group action viewed from two different bases. If two unitary representations are equivalent, then there always exists a unitary equivalence.
\paragraph{Irreducible Representations and Schur's Lemma}
\begin{definition}[$\mc{G}$-invariant subspace, and irreducibility]
Let $\phi$ be a representation of $\mc{G}$ acting on $V$.
A closed subspace $W$ of $V$ is \emph{$\mc{G}$-invariant} if
\begin{equation}
    \phi_gW\subseteq W\qquad\forall g\in\mc{G}.
\end{equation}
The map $\phi|_W:\mc{G}\mapsto GL(W)$ is then a representation on $W$ and called a \emph{subrepresentation} of $\phi$. The representation $\phi$, and by extension $V$, are irreducible if $V\neq\{\vec0\}$ and its only $\mc{G}$-invariant are $\{\vec0\}$ and $V$ itself.
\end{definition}
When considering unitary representation, the orthogonal complement of an invariant subspace is necessary invariant as well. One of the most important results in representation theory is Schur's lemma. It can be stated in terms of intertwiners between irreducible representations.
\begin{lemma}[Schur's lemma]\label{lem:Schur}
    If $S$ is an intertwiner between two irreducible representations $\phi$ and $\phi'$ then
    \begin{itemize}
        \item either $S=0$, or $S$ is an equivalence,
        \item if, moreover, $\phi=\phi'$, then $S=c\mathds{1}$ for some complex number $c$.
    \end{itemize}
\end{lemma}
\paragraph{Isotypic Decomposition}
Irreducible representations of compact groups are necessarily finite dimensional. We can classify all irreducible representations in equivalence classes denoted by a label $\lambda$, then for each equivalence class we can choose an arbitrary representation $\phi^\lambda$, acting on a representation space $V_\lambda$. It is then possible to decompose any representation into irreducible representations in the following way:
\begin{theorem}[Complete reducibility and irreducible decomposition]\label{Th:CompletRed}
Let $\phi$ be a representation\footnote{We require this representation to be strongly continuous for Lie groups.} on $V$. There exists a unique decomposition 
\begin{equation}
    V=\bigoplus_\lambda V^{(\lambda)}\qquad \text{with}\qquad V^{(\lambda)}\cong_{\mc{G}} V_\lambda\otimes \mathbb{C}^{m_\lambda}.
\end{equation}
such that each $V_\lambda$ is an irreducible representation.
Shortly, $V^{(\lambda)}$, the closure of all invariant subspaces $W\cong_{\mc{G}} V_\lambda$ with $\phi|_W\cong\phi^\lambda$, is called the \emph{isotypic component}, and  $m_\lambda$ is the multiplicity of the corresponding irreducible representation $\phi^\lambda$. Both are uniquely determined by $\phi$. Globally we have that
\begin{equation}
    \phi\cong\bigoplus_\lambda\phi^\lambda\otimes \mathds{1}_{m_\lambda}.
\end{equation}
As a consequence, two representations are equivalent if they have the same multiplicities for each irreducible subrepresentation. When considering finite dimensional representations, this theorem is usually called complete reducibility.
\end{theorem}
When considering infinite-dimensional vector spaces, $m_\lambda$ can become infinite, and then one should replace $\mathbb{C}^{m_\lambda}$ by $\ell^2(\mathbb{N})$, and $\mathds{1}_{m_\lambda}$ by $\mathds{1}_{\ell^2(\mathbb N)}$. We can combine the isotypic decomposition with Schur's lemma on each block to derive a corollary that we will use to derive~\cref{Corr:Niiiice}.
\begin{corollary}[Structure of intertwiners]\label{Corr:StrucIntert}
    Given two representations $\phi\cong\bigoplus_\lambda \phi_\lambda\otimes\mathds{1}_{m_\lambda}$ and $\phi\cong\bigoplus_\lambda \phi_\lambda\otimes\mathds{1}_{m'_\lambda}$ acting respectively on $V$ and $V'$, any intertwiner $S:V\to V'$ has the following structure in the isotypic decomposition:
    \begin{equation}
        S=\bigoplus_\lambda \mathds{1}_{V_\lambda}\otimes S_\lambda\qquad S_\lambda:\mathbb{C}^{m_\lambda}\to\mathbb{C}^{m'_\lambda}.
    \end{equation}
    $S$ is an equivalence if and only if $\phi\cong\phi'$ (hence $m_\lambda=n_\lambda$), and if every $S_\lambda$ is invertible. Naturally, $S$ is unitary if and only if every $S_\lambda$ is a unitary.
\end{corollary}

\section{Symmetric Transducers}\label{Sec:symTra}
In this manuscript we will focus on how the symmetry properties of a problem help in finding optimal transducers. To that end, we first need to formally define the symmetry group of the problem:
\begin{definition}[Symmetry group of $P$]\label{def:symmetry-group}
The \emph{symmetry group $\mc{G}$ of the problem $P$} is the group of automorphisms of $(\mc{X},\mc{T},\mc{O})$. Explicitly, it is the maximal subgroup of bijections $g$ in $\mathcal{I}$ such that there exists four unitary representations $\phi^\xi,\phi^\tau\in\mathrm{Hom}(\mc{G},U(\mc{H}))$ and $\phi^R,\phi^L\in\mathrm{Hom}(\mc{G},U(\mc{M}))$ satisfying
\begin{equation}\label{eq:DefRepProb}
\begin{cases}
    \ket{\xi_{g(i)}}=\phi^\xi_g\ket{\xi_i}\\
    \ket{\tau_{g(i)}}=\phi^\tau_g\ket{\tau_i}\\
    O_{g(i)}=\phi^L_gO_i\phi_g^{R\dagger}\\
\end{cases}\forall i\in \mathcal{I}.
\end{equation}
This condition implies
\begin{equation}
    \begin{cases}
    X(g(i),g(j))=X(i,j)\\
    T(g(i),g(j))=T(i,j)\\
    K(g(i),g(j))=\phi^R_gK(i,j)\phi_g^{R\dagger}
    \end{cases}\forall i,j\in \mathcal{I},
\end{equation}
which means that the symmetry group is a subgroup of the automorphisms that leave the triplet $(X,T,[K])$ invariant.
\begin{equation}
    \begin{cases}
        X\circ(g\times g)=X\\
        T\circ(g\times g)=T\\
        [K\circ(g\times g)]=[K]\\
    \end{cases}
\end{equation}
\end{definition}

\subsection{Symmetries of the Catalyst}\label{sec:symCat}
In this section, we consider a fixed state conversion problem $P$ with symmetry group $\mc{G}$, which by \cref{def:symmetry-group} implies the existence of unitary representations $\phi^\xi,\phi^\tau\in\mathrm{Hom}(\mc{G},U(\mc{H}))$ and $\phi^R,\phi^L\in\mathrm{Hom}(\mc{G},U(\mc{M}))$ satisfying \cref{eq:DefRepProb}. In what follows, we consider $\phi^\xi,\phi^\tau,\phi^R$ and $\phi^L$ to be fixed as well.

We can expect that the symmetry group $\mc{G}$ of $P$ induces symmetry properties for the transducing algorithms for $P$. We will prove that the catalysts satisfy a property that we call weak covariance.
\begin{definition}[Weak covariance]\label{def:weakCov}
    A family of catalysts $\{\ket{v_i}\}_{i\in\mc{I}}$ is weakly covariant if there exists a unitary representation $\phi^{\mc{W}}$ of $\mc{G}$ in $\mc{W}$ , that is strongly continuous if $\mc G$ is a compact Lie group,such that
    \begin{equation}
        \ket{v_{g(i)}}=\phi^{\mc{L}}_g\ket{v_i}\qquad \forall i\in\mc{I},~\forall g\in\mc{G},\qquad\text{with}\qquad\phi^{\mc{L}}_g=\phi^{\mc{W}}_g\otimes\phi^R_g.
    \end{equation}
    By extension, a transducing algorithm $\mc A=(S^\circ,\{\ket{v_i}\}_{i\in\mc{I}})$ is weakly covariant if the family of catalysts is weakly covariant.
\end{definition}
In many cases, we will require a stronger condition on the catalyst to easily find explicit solutions, which we call strong covariance:
\begin{definition}[Strong covariance]\label{def:strongGov}
 A family of catalysts $\{\ket{v_i}\}_{i\in\mc{I}}$ is strongly covariant if
    \begin{equation}
        \ket{v_{g(i)}}=\phi^{\mc{L}}_g\ket{v_i}\qquad \forall i\in\mc{I},~\forall g\in\mc{G},\qquad\text{with}\qquad\phi^{\mc{L}}_g=\mathds{1}_{\mc{W}}\otimes\phi^R_g.
    \end{equation}
    By extension, a transducing algorithm $\mc{A}=(S^\circ,\{\ket{v_i}\}_{i\in\mc{I}})$ is strongly covariant if the family of catalysts is strongly covariant.
\end{definition}

However, we can find examples of problems $P$ where the optimal transducing algorithm is not strongly covariant, or that do not admit (even suboptimal) strongly covariant transducing algorithms.

The first main theorem of this manuscript states that for any algorithm $\mc A$, there exists a weakly covariant algorithm $\tilde{\mc A}$ with a lower or equal transduction complexity. This new algorithm is obtained by symmetrization of the catalyst. For each group elements, we can define a rotated family of catalysts with the same transduction complexity as the original, then build a symmetrized catalyst as a controlled uniform superposition over all group elements. Then, the bound on average implies that this new family of catalysts, and thus the corresponding algorithm, has a lower or equal complexity. A formal version of the proof is given in~\cref{proof:SymmetryCat}.
\begin{theorem}\label{theorem:symmetryCat}
Let $P$ be a problem with symmetry group $\mc{G}$. For any transducing algorithm $\mc A=(S^\circ,\{\ket{v_i}\})$ solving $P$, there exists a weakly covariant algorithm $\tilde{A}=(\tilde{S}^\circ,\{\ket{\tilde{v}_i}\})$ with lower or equal transduction complexity:
\begin{equation}
    W(\tilde{\mc A})\leq W(\mc A).
\end{equation}
\end{theorem}
And we can directly find the following corollary:
\begin{corollary}\label{corr:optCat}
For any problem $P$ with symmetry group $\mc{G}$, there exists an optimal transducing algorithm that is weakly covariant.
\end{corollary}
One can now ask whether there always exists an optimal transducing algorithm that is strongly covariant. However, we can find counter-examples:
\begin{itemize}
    \item no strongly covariant algorithm solves \cref{prob:SearchCE1}, it is a variation of unstructured search with a single element marked by a weak oracle. One of the main ingredients is that given the structure of the oracle, the symmetry group of the problem is not the symmetric group $\mc{S}_n$ anymore, but the cyclic group $\mc C_n$. Note that while we expressed the problem as a state conversion problem for consistency, it can easily be transformed into a functional problem.
    \item \cref{prob:CosineCE} can be solved by a strongly covariant algorithm, but it is not optimal. For this problem, the oracle is redundant in the way it provides information, but imposing strong covariance implies some symmetrization that causes a wasteful, and therefore sub-optimal, exploitation of this redundancy.
\end{itemize}

The symmetry of the catalyst greatly reduces the space in which it can live, and thus simplifies the search for optimal transducing algorithms.
\begin{proposition}\label{Corr:SpaceRed}
A weakly covariant family of catalysts is left invariant by any group element that stabilizes the corresponding input, that is,
    \begin{equation}
        \forall i\in\mc I:\quad \ket{v_i}=\phi_g^{\mc{L}}\ket{v_i}\quad\forall g\in\mathrm{Stab}(i).
    \end{equation}
Equivalently, the catalyst lives in the subspace of $\mc{L}$ stabilized by the stabilizator of its relative input
\begin{equation}
    \ket{v_i}\in\Pi_{\mathrm{Stab(i)}}\mc{L}.
\end{equation}
\end{proposition}
The proof is quite immediate from the definition of weak covariance as if $g(i)=i$, then simultaneously $\ket{v_{g(i)}}=\phi_g^{\mc{L}}\ket{v_i}$ and $\ket{v_{g(i)}}=\ket{v_i}$.

\subsection{Symmetry Theorem for the Trandsucer}\label{sec:SymTra}

The joint action of the group on the input and target states, the oracle and the catalyst implies the following consequence on the input-independent unitary $S^\circ$ of a transducing algorithm:
\begin{theorem}\label{theorem:symmetryTrans}
    If a transducing algorithm is weakly covariant, then its input-independent unitary $S^\circ$ is an intertwiner from $\mc{X}=\mathrm{Span}\left(\ket{\xi_i}\oplus\left(\mathds{1}_{\mc{W}}\otimes O_i\right)\ket{v_i},\forall i\in\mc{I}\right)$ to $\mc{T}=\mathrm{Span}\left(\ket{\tau_i}\oplus\ket{v_i},\forall i\in\mc{I}\right)=S\mc{X}$:
        \begin{equation}
            S^\circ\left.\left(\phi^\xi\oplus\phi^{\mc{W}}\otimes\phi^L\right)\right|_{\mc{X}}=\left.\left(\phi^\tau\oplus\phi^{\mc{W}}\otimes\phi^R\right)\right|_{\mc{T}}S^\circ,
        \end{equation}
        and therefore
        \begin{equation}
            \mc X\cong_{\mc G}\mc T\quad\text{and}\quad \left.\left(\phi^\xi\oplus\phi^{\mc{W}}\otimes\phi^L\right)\right|_{\mc{X}}\cong\left.\left(\phi^\tau\oplus\phi^{\mc{W}}\otimes\phi^R\right)\right|_{\mc{T}}.
        \end{equation}
\end{theorem}
The proof presented in~\cref{proof:SymmetryTrans} is a simple application of the definition of the symmetry group of the problem.
Combining that result with~\cref{Corr:StrucIntert}, we can deduce that the input-independent unitary possesses a block-diagonal decomposition in the isoptypic decomposition of the Hilbert space.
\begin{corollary}\label{Corr:Niiiice} With increasing levels of assumption, we can deduce the following structure for $S^\circ$:
    \begin{enumerate}
        \item As $S^\circ$ is an intertwiner, it admits a block decomposition in the isotypic blocks $\mc X^{(\lambda)}$ and $\mc{T}^{(\lambda)}$ of $\mc X$ and $\mc T$ respectively:
    \begin{equation}
        S^\circ|_\mc{X}=\bigoplus_{\lambda}U_\lambda\otimes S^\circ_\lambda\qquad\text{with}\qquad U_\lambda\mc X^{(\lambda)}=\mc T^{(\lambda)}\text{ and }S^\circ_\lambda\in U\left(\mc{X}^{(\lambda)}\right).
    \end{equation}
    \item If the orthogonal complements of $\mc X$ and $\mc T$, $\mc X^\perp$ and $\mc T^\perp$, are also equivalent, $\mc X^\perp\cong_{\mc G}\mc T^\perp$, then $S^\circ$ can be chosen as an intertwiner on those subspaces too, and thus as a global intertwiner. It then admits a global decomposition
    \begin{equation}
        S^\circ=\bigoplus_\lambda U_\lambda\otimes S^\circ_\lambda\qquad \text{with}\qquad U_\lambda V^{(\lambda)}=W^{(\lambda)}\text{ and }S^\circ_\lambda\in U\left(V^{(\lambda)}\right)
    \end{equation}
    where $V^{(\lambda)}$ and $W^{(\lambda)}$ are the isotypic decomposition of the total Hilbert space $\mc H\oplus\mc L$ regarding respectively the two representations $\phi^\xi\oplus\phi^{\mc{W}}\otimes\phi^L$ and $\phi^\tau\oplus\phi^{\mc{W}}\otimes\phi^R$. 
    \item If both representations are equal, $\phi^\xi\oplus\phi^{\mc{W}}\otimes\phi^L=\phi^\tau\oplus\phi^{\mc{W}}\otimes\phi^R$, then we can choose
\begin{equation}
    S^\circ=\bigoplus_\lambda\mathds{1}_{(\lambda)}\otimes S^\circ_\lambda,\qquad \text{with}\qquad S^\circ_\lambda\in U\left(V^{(\lambda)}\right)
\end{equation}
where $\mathds{1}_{(\lambda)}$ is the projector onto $V^{(\lambda)}$.
    \end{enumerate}
\end{corollary}

\section{Practical applications}\label{sec:app}

\subsection{Unstructured Search and Amplitude Amplification on Single Orbit}
The first problems we will consider are the unstructured search with a fixed number of marked elements, which consists in generating the uniform superposition over all marked indices of a bit string with $M$ marked elements, and amplitude amplification with fixed amplitude. Both problems are very similar and we will solve them in a similar fashion, with a subtle difference. In both cases we will first specify the adversary bound to the setup of the problem. Then, for unstructured search, we will only impose weak covariance (which can always be enforced) then optimize over all possible representations and find that the optimal family of catalysts happens to be strongly covariant, while for amplitude amplification we will compute a lower bound for the complexity, then directly impose strong covariance and verify later that the obtained family of catalysts saturates the lower bound. Finally, we will deduce from~\cref{Corr:Niiiice} the structure of the input-independent unitary and from there give an explicit transducing algorithm.

\subsubsection{Unstructured search with \texorpdfstring{$M$}{M} marked elements}\label{sec:searchM}
Let us first consider the unstructured search with $M$ marked elements. This problem is a variation of the one proposed by Grover~\cite{Gro96} for $M=1$ and generalized by Boyer Brassard and H{\o}yer and Tapp~\cite{BBHT98}.
Formally
\begin{problem}[Unstructured search with $M$ marked elements]\label{prob:Fixed_search}
Given $N\in\mathbb{N}$ the number of elements, and $1\leq M\leq N$ a promise on the number of marked elements, $\mathrm{Search}^N_M$ is defined as the triple
\begin{equation}
    \mathrm{Search}^N_M\cq\left(\mc{X}_M^N,\mc{T}_M^N,\mc{O}_M^N\right),
\end{equation}
over the input set $\mc{I}=\mc{B}^N_{M}\cq\{x\in\{0,1\}^N:\abs{x}=M\}$, and the Hilbert spaces $\mc{H}=\mc{M}=\mathbb{C}^N$; where the indexed families of input states, target states and oracles are defined as
{\footnotesize\begin{equation}
\mc{X}_M^N \cq \left\{\ket{\xi_x}=\ket{\psi}\cq\frac{1}{\sqrt{N}}\sum_{i=1}^N\ket{i}\right\}_{x\in \mc{B}^N_{M}}
            \mc{T}_M^N \cq \left \{\ket{\tau_x} = \ket{x}\cq \frac{1}{\sqrt{M}} \sum_{i:x_i=1} \ket{i}\right \}_{x\in \mc{B}^N_{M}}
    \mc{O}_M^N \cq \left\{O_x=\mathds{1}_{\mc M}-2\sum_{i:x_i=1}\proj{i}\right\}_{x\in \mc{B}^N_{M}},
\end{equation}}
where $\{\ket i\}_{1\leq i\leq N}$ is the canonical basis of $\mathbb C^N$.
\end{problem}
Essentially, an algorithm solving $\mathrm{Search}^N_M$ should convert the uniform superposition into the superposition over the marked elements of a bitstring $x$, with access to the phase oracle
\begin{equation}
    O_x\ket{i}=(-1)^{x_i}\ket{i},
\end{equation}
that apply a negative phase on the marked indices. In this section we will prove the following proposition:
\begin{proposition}\label{prop:searchM}
    Unstructured search with $M$ marked elements has a transduction complexity of 
    \begin{equation}
        W\left(\mathrm{Search}^N_M\right)=\sqrt{\frac{N-M}{4M}},
    \end{equation}
    and this complexity can be realized with $\mc{W}=\mathbb C$, the following catalyst and input-independent unitary
    \begin{equation}
        \ket{v_x}=\frac{(n-M)^{1/4}}{2M^{1/4}}\left(\ket x+\ket{\neg x}\right),\qquad S^\circ=\proj \psi\otimes S^\circ_\parallel+\left(\mathds 1-\proj \psi\right)\otimes S^\circ_\perp,
    \end{equation}
    where $S\circ_\parallel$ and $S^\circ_\perp$ are $2\times2$ unitaries, fixed up to a phase.
\end{proposition}
In term of complexity we recover the result of Bennett, Bernstein, Brassard and Vazirani~\cite{BBBV97}. We have a smaller constant than the one proved by Zalka~\cite{Zal99}, the difference is that we consider Las Vegas complexity whereas he consider the minimum fixed case of oracles.
\begin{proof}
~\paragraph{Reformulation of the problem}
The initial and final Gram matrices, and the operator-valued kernel respectively are
\begin{equation}\label{Eq:KerSearcM}
    \begin{cases}
        X(x,y)=1\\
        T(x,y)=\frac{\abs{x\land y}}{M}\\
        K(x,y)=O_x^\dagger O_y=O_{x\oplus y}
    \end{cases}\qquad\forall x,y\in\mc{B}^N_{M}.
\end{equation}
Using those expressions, we can specify the optimization program~(\cref{Eq:DefAdv}) for this particular problem
\begin{equation}
    \operatorname{W}=\min_{{v_x}\in\mc{L}}\max_{x\in\mc{B}^N_{M}}\norm{v_x}^2\qquad\text{s.t.}\quad 1-\frac{\abs{x\land y}}{M}=\braket{v_x|\mathds{1}_{\mc{L}}-\mathds{1}_{\mc{W}}\otimes O_{x\oplus y}|v_y}\quad\forall x,y\in\mc{B}^N_{M}
\end{equation}
We can decompose the catalyst in the usual basis of $\mc{M}$ such that $\ket{v_x}=\sum_{i=1}^N\ket{v_{x,i}}\ket{i}$, with $\ket{v_{x,i}}\in\mc{W}$ and simplify further the equation
\begin{equation}
    \operatorname{W}=\min_{{v_{x,i}}\in\mc{W}}\max_{x\in\mc{B}^N_{ M}}\sum_{i=1}^N\norm{v_{x,i}}^2\qquad\text{s.t.}\quad 1-\frac{\abs{x\land y}}{M}=2\sum_{i:x_i\oplus y_i=1}\braket{v_{x,i}|v_{y,i}}\quad\forall x,y\in\mc{B}^N_{M}
\end{equation}
\paragraph{Symmetry group of the problem}The symmetry group of the problem is the symmetric group $\mc{S}_N$. Its action on the bit strings is the standard permutation of the indices, and the representations $\phi^\xi, \phi^\tau, \phi^L$ and $\phi^R$ are all the natural permutation representation $\phi^N$ \footnote{There is a bit of ambiguity for $\phi^\xi$ that could be the trivial representation since $\ket{\xi_x}=\ket{\psi}$ does not depend on $x$. This will not matter much further. }
\begin{equation}
    \phi^N_g\ket{i}=\ket{g(i)}.
\end{equation}
\paragraph{Optimal catalyst}By theorem~\cref{theorem:symmetryCat} we know that there exists a weakly covariant solution, therefore
\begin{equation}
    \ket{v_{g(x)}}=\phi^{\mc{W}}_g\otimes\phi^N_g\ket{v_x}\iff\ket{v_{g(x),g(i)}}=\phi_g^{\mc{W}}\ket{v_{x,i}}
\end{equation}
Note that at this point we do not know $\phi_g^{\mc{W}}$. We can decompose the set of pairs $(x,i)\in\mc{P}^{N,M}=\mc{B}^N_{M}\times[N]$ into two sets $\mc{P}^{N,M}_0=\{(x,i)\in\mc{B}^N_{M}\times[N]:x_i=0\}$ and $\mc{P}^{N,M}_1=\{(x,i)\in\mc{B}^N_{M}\times[N]:x_i=1\}$ that can be seen as the the orbits of the two tuples $p^{N,M}_0=(0^{N-M}1^M,1)$ and $p^{N,M}_1=(1^M0^{N-M},1)$ such that $\mc{P}^{N,M}_\alpha =\mc{S}_Np^{N,M}_\alpha$. We can therefore relate the different components of the catalyst to only two vectors
\begin{equation}
\ket{v^{N,M}_\alpha}\equiv\ket{v_{p^{N,M}_\alpha}}\Rightarrow\ket{v_{x,i}}=\phi^{\mc{W}}_g\ket{v^{N,M}_{x_i}}\quad\forall (x,i)\in \mc{P}^{N,M}.\label{Eq:FixedGroverCataReduc}
\end{equation}
The objective value can now be reduced to an optimization over only the two vectors $\ket{v^{N,M}_\alpha}$ and all the possible representations $\phi^{\mc{W}}$
{\small\begin{equation}
    \operatorname{W}=\min_{\phi^{\mc{W}}\in\mathrm{Hom}(\mc{S},U(\mc{W}))}\min_{\ket{v^{N,M}_\alpha}\in\mc{W}}M\norm{v_1^{N,M}}^2+(N-M)\norm{v_0^{N,M}}^2\qquad\text{s.t.}\quad 1=4M\Re\left(\braket{v_0^{N,M}|\phi^{\mc{W}}_g|v_1^{N,M}}\right)\quad\forall g\in\mathrm{Stab}(1)
\end{equation}}
Dealing with the optimization over the realizations does not seem trivial; however, one can see that the subproblem for the  trivial representation is a relaxation of all the other and the trivial representation is therefore optimal. Indeed, for the trivial representation, the constraint for all permutation collapses to 
\begin{equation}
    1=4M\Re\left(\braket{v_0^{N,M}|v_1^{N,M}}\right),
\end{equation}
that is also a constraint for all other representations, as it is the constraint for the trivial permutation. Note that this proves for the single fixed Hamming weight orbit that the optimal solution is strongly covariant. We can therefore reduce the optimization program to the following 
\begin{equation}
    \operatorname{W}=\min_{\ket{v^{N,M}_\alpha}\in\mc{W}}M\norm{v_1^{N,M}}^2+(N-M)\norm{v_0^{N,M}}^2\qquad\text{s.t.}\quad 1=4M\Re\left(\braket{v_0^{N,M}|v_1^{N,M}}\right)
\end{equation}
That is a two dimensional problem whose optimal solution is given by the two vectors
\begin{equation}
    \begin{cases}
        \ket{v_1^{N,M}}=\frac{(N-M)^{1/4}M^{1/4}}{2M}\ket{0}\\
        \ket{v_0^{N,M}}=\frac{1}{2(N-M)^{1/4}M^{1/4}}\ket{0}\\
    \end{cases}
\end{equation}
for an arbitrary unit vector $\ket{0}$. The objective value is $\operatorname{W}=\sqrt{\frac{N-M}{4M}}$. This is also the optimal transduction complexity for unstructured search with $M$ marked elements. We can reconstruct the catalyst for the initial problem using~\cref{Eq:FixedGroverCataReduc}:
\begin{equation}
    \ket{v_x}
    =\sum_{i:x_i=1}\ket{v_1^{N,M}}\ket{i} +\sum_{i:x_i=0}\ket{v_0^{N,M}}\ket{i}
    =\frac{(N-M)^{1/4}}{2M^{1/4}}\ket{0}\left(\ket{x}+\ket{\neg x}\right)
\end{equation}
where we introduced $\ket{\neg x}=(N-M)^{-1/2}\sum_{i:x_i=0}\ket{i}$. As it does not play any role we can trace out the $\mc{W}$ register. 

\paragraph{Recovering the transducer} We now need to find a unitary $S^\circ$ that transforms
\begin{equation}
    \ket{\xi_x}\oplus\left(\mathds{1}_{\mc{W}}\otimes O_x\right)\ket{v_x}=\ket{\psi}\ket{0}+\frac{(N-M)^{1/4}}{2M^{1/4}}\left(-\ket{x}+\ket{\neg x}\right)\ket{1}\xrightarrow{S^\circ}\ket{\tau_x}\oplus\ket{v_x}=\ket{x}\ket{0}+\frac{(N-M)^{1/4}}{2M^{1/4}}\left(\ket{x}+\ket{\neg x}\right)\ket{1},
\end{equation}
where we used an additional register to distinguish the public and private spaces. On both sides the global group representation is the direct sum of the natural representation with itself
\begin{equation}
     \phi^\xi\oplus\phi^{\mc{W}}\otimes\phi^L=\phi^\tau\oplus\phi^{\mc{W}}\otimes \phi^R=\phi^N\oplus\phi^N=\phi^N\otimes\mathds{1}_2.
\end{equation}
As the natural representation is the direct sum of trivial representation and the standard decomposition, we can easily find the isotypic decomposition of the vector space to be the symmetric subspace and its orthogonal complement.
\begin{equation}
    \mc{H}\oplus\mc{L}=\left(\mathrm{Span}(\ket{\psi})\otimes \mathrm{Span}(\ket{0},\ket{1})\right)\oplus \left(\mathrm{Span}(\ket{\psi})^\perp\otimes \mathrm{Span}(\ket{0},\ket{1})\right)
\end{equation}
Then by~\cref{Corr:Niiiice}, we know that the transducer can be decomposed into two parts, acting on the different isotypic components. 
\begin{equation}\label{Eq:StructSSearch}
    S^\circ=\proj{\psi}\otimes S^\circ_\parallel+(\mathds{1}-\proj{\psi})\otimes S^\circ_\perp
\end{equation}
For a graphical representation, see \cref{fig:controlled-op-grov}.

Introducing the vector $\sqrt{N}\ket{\psi^\perp_x}\equiv \sqrt{N-M}\ket{x}-\sqrt{M}\ket{\neg x}$, and the scalar $\kappa=\frac{(N-M)^{1/4}}{2M^{1/4}},\mu=\sqrt{\frac{M}{N}},\nu=\sqrt{\frac{N-M}{N}}$, we can translate both side into their components in the isotypic decomposition
\begin{equation}
    \left\{\begin{aligned}
        \ket{\xi_x}\oplus\left(\mathds{1}_{\mc{W}}\otimes O_x\right)\ket{v_x}&=\ket{\psi}\otimes\begin{pmatrix}1\\\kappa(\nu-\mu)\end{pmatrix}+\ket{\psi_x^\perp}\otimes\begin{pmatrix}0\\-\kappa(\nu+\mu)\end{pmatrix}\\
        \xrightarrow{S^\circ}\ket{\tau_x}\oplus\ket{v_x}&=\ket{\psi}\otimes\begin{pmatrix}\mu\\\kappa(\mu+\nu)\end{pmatrix}+\ket{\psi_x^\perp}\otimes\begin{pmatrix}\nu\\\kappa(\nu-\mu)\end{pmatrix}
    \end{aligned}\right.\Rightarrow\left\{\begin{aligned}
        S^\circ_\parallel\begin{pmatrix}1\\\kappa(\nu-\mu)\end{pmatrix}&=\begin{pmatrix}\mu\\\kappa(\mu+\nu)\end{pmatrix}\\
        S^\circ_\perp\begin{pmatrix}0\\-\kappa(\nu+\mu)\end{pmatrix}&=\begin{pmatrix}\nu\\\kappa(\nu-\mu)\end{pmatrix}
    \end{aligned}\right.
\end{equation}
One can check that for the $2$-dimensional block the norms match, which ensures that $S^\circ_\parallel$ and $S^\circ_\perp$ exist. These $2\times2$ unitaries are fixed up to a phase and can be built as follows, defining the four vectors
\begin{equation}\label{Eq:DecSSearchM}
        \ket{v_\parallel}=\begin{pmatrix}1\\\kappa(\nu-\mu)\end{pmatrix},\quad\ket{w_\parallel}=\begin{pmatrix}\mu\\\kappa(\nu+\mu)\end{pmatrix},\quad\ket{v_\perp}=\begin{pmatrix}0\\-\kappa(\nu+\mu)\end{pmatrix},\quad\ket{w_\perp}=\begin{pmatrix}\nu\\\kappa(\nu-\mu)\end{pmatrix}.
\end{equation}
and $\ket{v_\parallel^\perp},\ket{w_\parallel^\perp},\ket{v_\perp^\perp},\ket{w_\perp^\perp}$ vectors orthogonal to each of them respectively, all fixed up to the phase and the norm, then the valid blocks are:
\begin{equation}
    S^\circ_\parallel=\frac{\ketbra{w_\parallel}{v_\parallel}}{\norm{w_\parallel}\norm{v_\parallel}}+\frac{\ketbra{w^\perp_\parallel}{v^\perp_\parallel}}{\norm{w^\perp_\parallel}\norm{v^\perp_\parallel}},\qquad S^\circ_\perp=\frac{\ketbra{w_\perp}{v_\perp}}{\norm{w_\perp}\norm{v_\perp}}+\frac{\ketbra{w^\perp_\perp}{v^\perp_\perp}}{\norm{w^\perp_\perp}\norm{v^\perp_\perp}}.
\end{equation}
Among these classes there is a single reflection
\begin{equation}
    S^\circ_\parallel=\mathds 1_2-2\frac{\left(\ket{v_\parallel}-\ket{w_\parallel}\right)\left(\bra{v_\parallel}-\bra{w_\parallel}\right)}{\norm{\ket{v_\parallel}-\ket{w_\parallel}}^2},\qquad S^\circ_\perp=\mathds 1_2-2\frac{\left(\ket{v_\perp}-\ket{w_\perp}\right)\left(\bra{v_\perp}-\bra{w_\perp}\right)}{\norm{\ket{v_\perp}-\ket{w_\perp}}^2},
\end{equation}
and a single rotation mapping $\ket{v_\parallel}\mapsto\ket{w_\parallel}$.
\end{proof}

\begin{figure}[!h]
    \begin{tikzpicture}
    \node (starting_line_h) at (0, 0) {};
    \node (ending_line_h) at (6.2, 0) {};
    \node (starting_line_p) at (0, -2) {};
    \node (ending_line_p) at (6.2, -2) {};

    \draw[thick] (starting_line_h) -- (ending_line_h);
    \draw[thick] (starting_line_p) -- (ending_line_p);

    \node (hilb) at (-0.5, 0) {$\HS$};
    \node (P) at (-0.5, -2) {$\mathcal{P}$};

    \node[circle, draw=black, fill=white] (controlled_u) at (4.1, 0){$\psi/\psi^\perp$};

    \node[circle, draw=black, fill=white] (controlled_oracle) at (2.1, -2){$1$};

    \draw[draw=black, fill=white, rounded corners=2pt] (3.6, -2.5) rectangle (4.6, -1.5);

    \draw[draw=black, fill=white, rounded corners=2pt] (1.6, -0.5) rectangle (2.6, 0.5);
    
    \node (control) at (4.1, -2) {$S^\circ_\parallel/S^\circ_\perp$};
    
    \node (oracle) at (2.1, 0) {$O_x$};

    \draw[thick] (controlled_u) -- (4.1,-1.5);
    \draw[thick] (controlled_oracle) -- (2.1,-0.5);

    \draw[thick, dotted] (0.7, 0.6) -- (0.7, -2.4) node[below] {$\ket{\psi_{in}}$};
    \draw[thick, dotted] (5.5, 0.6) -- (5.5, -2.4) node[below] {$\ket{\psi_{out}}$};

\end{tikzpicture}
    \caption{ Graphical representation of the action of the canonical transducer $S = (\mathds{1}_{\HS}\oplus (\mathds{1}_\mathcal{W} \otimes O_x))S^\circ$ where $S^\circ$ is given in \cref{Eq:StructSSearch}. The input vector is $\ket{\psi_{in}} = \ket{\psi}_\HS\ket{0}_{\mathcal{P}} +\ket{v_x}_{\HS}\ket{1}_{\mathcal{P}}$. Both operations act as controlled-unitaries on different subspaces of $\HS\otimes \mathcal{P}$. This is represented on the figure by the control on the register $\mathcal{P}$ (\textit{e.g.} the control qubit is $\ket{1}$ for the oracle) and the application on $\mathcal{H}$ (\textit{e.g.} the applied operator is $O_x$). The action of the transducer results in the desired output vector $\ket{\psi_{out}} = \ket{\tau_{x}}_{\HS}\ket{0}_{\mathcal{P}} + \ket{v_{x}}_{\HS}\ket{1}_{\mathcal{P}}$.}
    \label{fig:controlled-op-grov}
\end{figure}
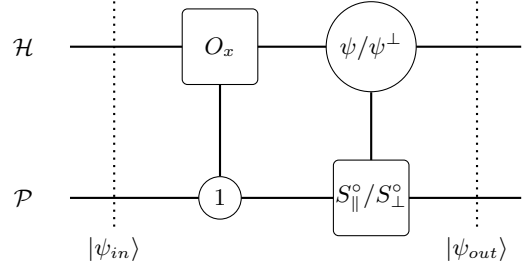

\subsubsection{Amplitude amplification with a fixed amplitude}\label{sec:AmpFix}
In this section, we consider the amplitude amplification problem where the initial amplitude is known to be some constant $\varepsilon$. This problem has been formalized by Brassard, H{\o}yer, Mosca and Tapp~\cite{BHMT02}.

\begin{problem}[Amplitude amplification with initial amplitude $\varepsilon$]\label{prob:fixed-amp}
    Given $\HS$ the public space, $\Pi_m$ the projector on the marked space $\mc H_m\cq \Pi_m\mc H$, and $\eps\in (0,1]$ a promise on the amplitude of the overlap over the marked space, $\mathrm{Amp}_\eps$ is defined as the triple
    \begin{equation}
    \mathrm{Amp}_\eps\cq\left(\mc{X}_\eps,\mc{T}_\eps,\mc{O}_\eps\right),
\end{equation}
    over the input set $\mathcal{I}=\mc B_\eps \cq \{\ket{\xi}\in\mc{H}:\norm{\xi}=1, \norm{\Pi_m\ket{\xi}}=\varepsilon\}$ with $\mc M=\mc H$; where the where the indexed families of input states, target states and oracles are defined as
{\begin{equation}
    \mc{X}_\eps \cq\mc{B}_\eps\qquad
            \mc{T}_\eps \cq \left \{\ket{\tau_{\ket\xi}} = \eps^{-1}\Pi_m \ket{\xi}\right \}_{\ket{\xi}\in\mc{B}_\eps}\qquad
    \mc{O}_\eps \cq \left\{O_{\ket{\xi}}=\mathds{1}_{\mc M}-2\proj{\xi}\right\}_{\ket{\xi}\in\mc{B}_\eps}.
\end{equation}}
\end{problem}
Informally, \textbf{amplitude amplification with initial amplitude $\varepsilon$} is the conversion of a state $\ket{\xi}$ into its projection on a marked subspace $\mc H_m$, with at our disposal the access to a reflection around $\ket{\xi}$. Usually one would consider the reflection with respect to the marked space as an oracle to, however in term of query complexity we will consider that this oracle is free as the complexity is computed the number of calls for a single oracle. One could define an alternative version of amplitude amplification where the marked space is the input and the reflection on $\ket{\xi}$ is free and recover the same complexity. Intuitively this will also appear in our construction, as the input-independent unitary~\cref{eq:controlled-unitary} is essentially a reflection controlled by wether or not the state is in the marked space as showned in~\cref{fig:controlled-op}. We could have explicitly built decomposed the unitary as a call to the reflection on the marked space oracle setting a controlling qbit, then the reflections.

Note that even tough $\mc H$ and $\Pi_m$ are parameters of the problem, as they have no impact on the complexity we omit them in the expression $\mathrm{Amp}_\eps$ by abuse of notation.

In this section we will prove the following proposition:
\begin{proposition}\label{prop:AmpFix}
    Amplitude amplification with initial amplitude $\varepsilon$ has a transduction complexity of 
    \begin{equation}
        W(\mathrm{Amp}_\eps)=\sqrt{\frac{1-\eps^2}{4\eps^2}},
    \end{equation}
    and this complexity can be realized with $\mc{W}=\mathbb C$ and the following catalyst and input-independent unitary
    \begin{equation}\label{eq:reflMarkSpa}
        \ket{v_{\ket\xi}}=\frac{(1-\varepsilon^2)^{1/4}}{2\sqrt{\varepsilon}}\left(\ket \xi+\ket{\xi^\perp}\right),\qquad S^\circ=\Pi_m\otimes S_m+\Pi_u\otimes S^\circ_u,
    \end{equation}
    where $\Pi_u\cq\mathds{1}_{\mc H}-\Pi_m$,
    \begin{equation}
    \ket{\xi^\perp}\cq \left(-\frac{\beta}{\eps}\Pi_m + \frac{\eps}{\beta}\Pi_u\right)\ket{\xi}.
\end{equation}
    and $S^\circ_u$ and $S^\circ_m$ are $2\times2$ unitaries, fixed up to a phase.
\end{proposition}
This complexity is similar to the one found by Brassard, H{\o}yer, Mosca and Tapp~\cite{BHMT02} or Long~\cite{Lon01}. The constant is different as we consider Las Vegas complexity. Concurrently, this catalyst was also discovered by Cornelissen, Apers and Gribling~\cite{cornelissenImprovedQuantumVolume2026} and an alternative one has been built by Apers, Roland and Zhang~\cite{apersElfsTransducersQuantum2026}, with a higher complexity.
\begin{proof}
~\paragraph{Reformulation of the problem} In the following, we will extensively use an alternative formulation of the input set, where we decompose each vector into their marked and unmarked parts $\mc B_\eps\approx\{ (\ket{m},\ket{u}) \in \HS_m \times \HS_u\}$, where $\mc H_u\cq\mc H_m^\perp$ and $\Pi_u\cq\mathds 1_{\mc H}-\Pi_m$; with the natural equivalence
\begin{equation}
    \begin{cases}
        \ket{m}\cq \varepsilon^{-1}\Pi_m \ket{\xi}\\
        \ket{u}\cq \beta^{-1}\Pi_u \ket{\xi}
    \end{cases}\iff\ket{\xi} = \varepsilon\ket{m} + \beta\ket{u},
\end{equation}
where $\beta=\sqrt{1-\varepsilon^2}$. In the following we will use this bijection implicitly and consider, for a pair of input vectors $\ket{\xi},\ket{\xi'}$, that
\begin{equation}
        \ket{\xi}=\eps\ket{m}+\beta\ket{u} 
        \quad\text{and}\quad 
        \ket{\xi'}=\eps\ket{m'}+\beta\ket{u'}.
\end{equation}
We can develop the initial and final Gram kernels\footnote{We will not use the word matrix in this context as the input set is not discrete, its elements take value over a manifold.} and the operator-valued kernel using this representation to get
\begin{equation}
    \begin{cases}
        X(\xi,\xi')=\varepsilon^2 \braket{m|m'} + \beta^2\braket{u|u'}\\
        T(\xi,\xi')=\braket{m|m'}\\
        K(\xi,\xi')=\mathds{1}- 2(\proj{\xi} +\proj{\xi'}) + 4X(\xi,\xi') \ketbra{\xi}{\xi'} 
    \end{cases}
    \quad \forall \ket{\xi},\ket{\xi'} \in \mathcal{I}.
\end{equation}
If we denote the catalyst of a given input state $\ket{\xi}$ as $\ket{v_\xi}$, we find the following formulation of the problem
\begin{equation}
    \operatorname{W}=\min_{\ket{v_{\ket{\xi}}}\in\mc{L}}\max_{\ket{\xi}\in\mc B_\eps}\norm{v_{\ket{\xi}}}^2\qquad\text{S.T.}\qquad \beta^2(\braket{u|u'} - \braket{m|m'}) =  \braket{v_{\ket{\xi}} | \mathds{1} - (\mathds{1} \otimes O^\dagger_{\ket{\xi}} O_{\ket{\xi}})|v_{\ket{\xi'}}}\quad \forall \ket{\xi},\ket{\xi'} \in \mathcal{I}.
\end{equation}
\paragraph{Lower bound on the complexity} We could proceed as for search and now look for an optimal catalyst assuming it is weakly covariant, but we will take a different approach. We will find a lower bound on the complexity of the problem by characterizing a subproblem over a subset of inputs, then directly assume strong covariance on the catalyst and verify that it saturates the lower bound. Let us thus consider the reduced input set
\begin{equation}
    \mc{I}_2=\left\{\ket{+} = \varepsilon\ket{m} + \beta\ket{u}, \quad \ket{-} =\varepsilon\ket{m} - \beta\ket{u}\right\},
\end{equation}
for some arbitrary $(\ket{m},\ket{u})\in\mc B_\eps$. On this subset of inputs the kernels are
\begin{equation}
        X=\begin{pmatrix}
            1 & 1 - 2\beta^2\\
            1 - 2\beta^2 & 1
        \end{pmatrix} \qquad
        T =\begin{pmatrix}
            1 &1\\
            1 & 1
        \end{pmatrix}\qquad
        K(+,-)= -4\varepsilon \beta\begin{pmatrix}
            2\varepsilon\beta & \varepsilon^2-\beta^2\\
            \beta^2-\varepsilon^2 & 2\varepsilon \beta
       \end{pmatrix}+\mathds{1}
\end{equation}
and the problem reduces to
\begin{equation}
    \operatorname{W}=\min_{\ket{v_+},\ket{v_-},\in\mc{L}}\max_{\pm}\norm{v_\pm}^2\qquad\text{S.T.}\qquad -2\beta^2=  \braket{v_{+} | \mathds{1} -K(+,-)|v_{-}}.
\end{equation}
One can directly see that the spectral norm is
\begin{equation}
    \norm{\mathds{1} -K(+,-)}_\infty=4\beta \varepsilon
\end{equation}
and therefore we can deduce from the constraint the following inequality
\begin{equation}
    2\beta^2\leq 4\beta \varepsilon\norm{v_+}\norm{v_-}\Rightarrow\norm{v_+}\norm{v_-}\geq\frac{\beta}{2\varepsilon}
\end{equation}
by weak covariance we know that $\ket{v_+}$ and $\ket{v_-}$ are related by a unitary representing the group $\mc{S}_2$, and therefore have the same norm $\norm{v_+}=\norm{v_-}=\norm{v}$, and thus
\begin{equation}
    \operatorname{W}=\norm{v}^2\geq\sqrt{\frac{1-\varepsilon^2}{4\varepsilon^2}}.
\end{equation}
This lower bound matches the one for search and is the one we expect.

\paragraph{Symmetry group of the problem} The symmetry group of the problem consists in the group of operators that commutes with $\Pi_m$, and thus the group of unitaries acting separately on $\HS_m$ and $\HS_u$, explicitly  $\mc{V} = \{V \in U(\mathcal{M}): [V, \Pi_m] = 0\}=U(\mc{H}_m)\oplus U(\mc{H}_u)$ acting with the defining representation:
\begin{equation}\label{Eq:SymAmpAmp}
    \phi^\xi=\phi^\tau=\phi^R=\phi^L=\phi^d\quad\text{S.T.} \quad\phi^d_U\ket{\xi}=U\ket{\xi}.
\end{equation}
We note the catalyst of a given input state $\ket{\xi}$ of the problem as $\ket{v_{\ket{\xi}}}$ and we will directly assume it is strongly covariant. This means for all $\ket{\xi} \in \mathcal{I}$ and all element $V \in \mathcal{V}$,
\begin{equation}
    \ket{v_{V\ket{\xi}}} = \mathds{1} \otimes V\ket{v_{\ket{\xi}}}.
\end{equation}
\paragraph{Optimal catalyst} By \cref{Corr:SpaceRed}, we know that the catalyst should live in the space stabilized by the stabilizer of $\ket{\xi}$ in $\mc{V}$, which is the direct sum of the stabilizers of $\ket{m}$ in $U(\mc{H}_m)$ and $u$ in $U(\mc{H}_u)$. The subspace of $\mc{L}$ that is stabilized by this subgroup is trivially $\mc{W}\otimes\mathrm{Span}(\ket{m},\ket{u})$. We can thus decompose 
\begin{equation}\label{Eq:AmpAmpDecCat}
    \ket{v_{\ket{\xi}}} = \ket{v^{\parallel}_{\ket{\xi}}}\ket{\xi} + \ket{v_{\ket{\xi}}^\perp} \ket{\xi^\perp}
\end{equation}
for some vectors $\ket{v^{\parallel}_{\ket{\xi}}}, \ket{v_{\ket{\xi}}^\perp} \in \mathcal{W}$ with
\begin{equation}\label{Eq:defXiPerp}
    \ket{\xi^\perp}\cq -\beta\ket{m} + \varepsilon\ket{u}.
\end{equation}
As, for all $\ket{\xi},\ket{\xi'}\in\mc B_\eps$, there exists a unitary $V\in\mc{V}$ such that $\ket{\xi'}=V\ket{\xi}$, the input set is a single orbit and we can use strong covariance to relate the two parts $\ket{v^{\parallel}_{\ket{\xi}}}$ and $ \ket{v_{\ket{\xi}}^\perp}$ and find
\begin{equation}
    \begin{cases}
        \ket{v^\parallel}\cq\ket{v^{\parallel}_{\ket{\xi}}}=\ket{v^{\parallel}_{\ket{\xi'}}}\\
        \ket{v^\perp}\cq\ket{v^{\perp}_{\ket{\xi}}}=\ket{v^{\perp}_{\ket{\xi'}}}\
    \end{cases}\Rightarrow\ket{v_{\ket{\xi}}}=\ket{v^{\parallel}}\ket{\xi}+\ket{v^\perp}\ket{\xi^\perp}.
\end{equation}
The action of the oracle on the catalyst consists of a reflection around the input state $\ket{\xi}$. Thus, we have 
\begin{equation}
O_\xi\ket{v_\xi} = -\ket{v^\parallel}\ket{\xi} + \ket{v^\perp} \ket{\xi^\perp}
\end{equation}
And as a result we can find a simplified expression for the problem
\begin{equation}
     \operatorname{W}=\min_{\ket{v^\parallel},\ket{v^\perp}\in\mc{W}}\norm{v^\perp}^2+\norm{v^\parallel}^2\qquad\text{S.T.}\qquad \beta^2(\braket{u|u'} - \braket{m|m'}) =  2(\braket{v^\parallel| v^\perp} \braket{\xi|\xi'^\perp} + \braket{v^{\perp}| v^{\parallel}} \braket{  \xi^\perp | \xi'}).
\end{equation}
Knowing the scalar product
\begin{equation}
        \braket{\xi | \xi'^\perp}=\braket{\xi^\perp | \xi'} = \beta \varepsilon
        \left(\braket{u|u'}-\braket{m|m'} \right)
\end{equation}
the problem is now
\begin{equation}
    \operatorname{W}=\min_{\ket{v^\parallel},\ket{v^\perp}\in\mc{W}}\norm{v^\perp}^2+\norm{v^\parallel}^2\qquad\text{S.T.}\qquad\frac{\beta}{4\varepsilon} = \Re(\braket{v^\parallel|v^\perp}).
\end{equation}
This is a two-dimensional optimization problem whose solution is
\begin{equation}
        \ket{v^\parallel}=\ket{v^\perp}=\sqrt{\frac{\beta}{4\varepsilon}}\ket{0},
\end{equation}
for some arbitrary normal vector $\ket{0}$. The corresponding optimal value is $\operatorname{W}=\frac{\beta}{2\varepsilon}=\sqrt{\frac{1-\eps^2}{4\eps^2}}$, which is the awaited complexity for amplitude amplification, and the optimal catalyst is
\begin{equation}\label{eq:cata-amp}
    \ket{v_{\ket{\xi}}} = \sqrt{\frac{\beta}{4\varepsilon}} (\ket{\xi} + \ket{\xi^\perp}),
\end{equation}
dropping the useless register $\ket{0}$.

\paragraph{Recovering the transducer} We now need to find a unitary $S^\circ$ that converts
\begin{equation}
    \ket{\xi}\oplus\left(\mathds{1}_{\mc{W}}\otimes O_{\ket{\xi}}\right)\ket{v_{\ket{\xi}}}=\ket{\xi}\ket{0}+\sqrt{\frac{\beta}{4\varepsilon}}\left(-\ket{\xi}+\ket{\xi^\perp}\right)\ket{1}\xrightarrow{S^\circ}\ket{\tau_{\ket{\xi}}}\oplus\ket{v_{\ket{\xi}}}=\ket{m}\ket{0}+\sqrt{\frac{\beta}{4\varepsilon}}\left(\ket{\xi}+\ket{\xi^\perp}\right)\ket{1},
\end{equation}
where we used an ancilla register $\mc P$to distinguish the public and private spaces. Similarly to search, we once again see that on both sides the representations are the same, the doubling of the defining representation:
\begin{equation}
     \phi^\xi\oplus\phi^{\mc{W}}\otimes\phi^L=\phi^\tau\oplus\phi^{\mc{W}}\otimes \phi^R=\phi^N\oplus\phi^N=\phi^d\otimes\mathds{1}_2.
\end{equation}
The isotypic decomposition of the Hilbert space simply separates the marked and unmarked spaces
\begin{equation}
    \mc{H}\oplus\mc{L}=\left(\mc{H}_m\otimes \mathrm{Span}(\ket{0},\ket{1})\right)\oplus \left(\mc{H}_u(\ket{\psi})\otimes \mathrm{Span}(\ket{0},\ket{1})\right)
\end{equation}
Once again, using~\cref{Corr:Niiiice}, the transducer (see \cref{fig:controlled-op}) can be decomposed into a part acting on each side
\begin{equation}\label{eq:controlled-unitary}
    S^\circ=\Pi_m\otimes S^\circ_m+\Pi_u\otimes S^\circ_u
\end{equation}
Introducing the scalar $\kappa=\sqrt{\frac{\beta}{4\varepsilon}}$, we can explicitly separate the action of the unitary
\begin{equation}
    \left\{\begin{aligned}
        \ket{\xi}\oplus\left(\mathds{1}_{\mc{W}}\otimes O_{\ket{\xi}}\right)\ket{v_{\ket{\xi}}}&=\ket{m}\otimes\begin{pmatrix}\varepsilon\\-\kappa(\beta+\varepsilon)\end{pmatrix}+\ket{u}\otimes\begin{pmatrix}\beta\\\kappa(\varepsilon-\beta)\end{pmatrix}\\
        \xrightarrow{S^\circ}\ket{\tau_{\ket{\xi}}}\oplus\ket{v_{\ket{\xi}}}&=\ket{m}\otimes\begin{pmatrix}1\\\kappa(\varepsilon-\beta)\end{pmatrix}+\ket{u}\otimes\begin{pmatrix}0\\\kappa(\varepsilon+\beta)\end{pmatrix}
    \end{aligned}\right.\Rightarrow\left\{\begin{aligned}
        S^\circ_m\begin{pmatrix}\varepsilon\\-\kappa(\beta+\varepsilon)\end{pmatrix}&=\begin{pmatrix}1\\\kappa(\varepsilon-\beta)\end{pmatrix}\\
        S^\circ_u\begin{pmatrix}\beta\\\kappa(\varepsilon-\beta)\end{pmatrix}&=\begin{pmatrix}0\\\kappa(\varepsilon+\beta)\end{pmatrix}
    \end{aligned}\right.
\end{equation}
These equations are totally equivalent to those found for search, \cref{Eq:DecSSearchM}, and we will therefore be able to directly translate the conclusions. One can check that for the $2$-dimensional block the norms match, which ensures that $S^\circ_m$ and $S^\circ_u$ exist. Once again these $2\times2$ unitaries are fixed up to a phase and we can build them explicitly by defining the vectors
\begin{equation}
        \ket{v_m}=\begin{pmatrix}s\\-\kappa(\beta+\varepsilon)\end{pmatrix},\quad\ket{w_m}=\begin{pmatrix}1\\\kappa(\varepsilon-\beta)\end{pmatrix},\quad\ket{v_u}=\begin{pmatrix}\beta\\\kappa(\varepsilon-\beta)\end{pmatrix},\quad\ket{w_u}=\begin{pmatrix}0\\\kappa(\varepsilon+\beta)\end{pmatrix}.
\end{equation}
and $\ket{v_m^\perp},\ket{w_m^\perp},\ket{v_u^\perp},\ket{w_u^\perp}$ vectors orthogonal to each of them respectively, all fixed up to the phase and the norm, then the valid block unitaries are:
\begin{equation}
    S^\circ_m=\frac{\ketbra{w_m}{v_m}}{\norm{w_m}\norm{v_m}}+\frac{\ketbra{w^\perp_m}{v^\perp_m}}{\norm{w^\perp_m}\norm{v^\perp_m}},\qquad S^\circ_u=\frac{\ketbra{w_u}{v_u}}{\norm{w_u}\norm{v_u}}+\frac{\ketbra{w^\perp_u}{v^\perp_u}}{\norm{w^\perp_u}\norm{v^\perp_u}}.
\end{equation}
Among these classes there is a single reflection
\begin{equation}
    S^\circ_m=\mathds 1_2-2\frac{\left(\ket{v_m}-\ket{w_m}\right)\left(\bra{v_m}-\bra{w_m}\right)}{\norm{\ket{v_m}-\ket{w_m}}^2},\qquad S^\circ_u=\mathds 1_2-2\frac{\left(\ket{v_u}-\ket{w_u}\right)\left(\bra{v_u}-\bra{w_u}\right)}{\norm{\ket{v_u}-\ket{w_u}}^2},
\end{equation}
and a single rotation mapping $\ket{v_m}\mapsto\ket{w_m}$.
\end{proof}

\begin{figure}[!h]
    \begin{tikzpicture}
    \node (starting_line_h) at (0, 0) {};
    \node (ending_line_h) at (6.2, 0) {};
    \node (starting_line_p) at (0, -2) {};
    \node (ending_line_p) at (6.2, -2) {};

    \draw[thick] (starting_line_h) -- (ending_line_h);
    \draw[thick] (starting_line_p) -- (ending_line_p);

    \node (hilb) at (-0.5, 0) {$\HS$};
    \node (P) at (-0.5, -2) {$\mathcal{P}$};

    \node[circle, draw=black, fill=white] (controlled_u) at (4.1, 0){$m/u$};

    \node[circle, draw=black, fill=white] (controlled_oracle) at (2.1, -2){$1$};

    \draw[draw=black, fill=white, rounded corners=2pt] (3.6, -2.5) rectangle (4.6, -1.5);

    \draw[draw=black, fill=white, rounded corners=2pt] (1.6, -0.5) rectangle (2.6, 0.5);
    
    \node (control) at (4.1, -2) {$S^\circ_m/S^\circ_u$};
    
    \node (oracle) at (2.1, 0) {$O_x$};

    \draw[thick] (controlled_u) -- (4.1,-1.5);
    \draw[thick] (controlled_oracle) -- (2.1,-0.5);

    \draw[thick, dotted] (0.7, 0.6) -- (0.7, -2.4) node[below] {$\ket{\psi_{in}}$};
    \draw[thick, dotted] (5.5, 0.6) -- (5.5, -2.4) node[below] {$\ket{\psi_{out}}$};

\end{tikzpicture}
    \caption{Graphical representation of the action of the canonical transducer $S = (\mathds{1}_{\HS}\oplus (\mathds{1}_\mathcal{W} \otimes O_{\ket{\xi}}))S^\circ$ where $S^\circ$ is given in \cref{eq:controlled-unitary}. The input vector is $\ket{\psi_{in}} = \ket{\xi}_\HS\ket{0}_{\mathcal{P}} +\ket{v_{\ket{\xi}}}_{\HS}\ket{1}_{\mathcal{P}}$. Both operations act as controlled-unitaries on different subspaces of $\HS\otimes \mathcal{P}$. This is represented in the figure by the control on the register $\mathcal{P}$ (\textit{e.g.} the control qubit is $\ket{1}$ for the oracle) and the application on $\mathcal{H}$ (\textit{e.g.} the applied operator is $O_x$). The action of the transducer results in the desired output vector $\ket{\psi_{out}} = \ket{\tau_{\ket{\xi}}}_{\HS}\ket{0}_{\mathcal{P}} + \ket{v_{\ket{\xi}}}_{\HS}\ket{1}_{\mathcal{P}}$.}
    \label{fig:controlled-op}
\end{figure}
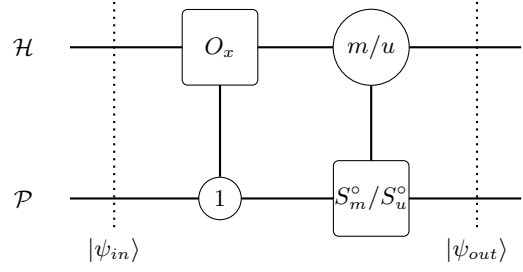

\subsection{Unstructured Search and Amplitude Amplification on Multiple Orbits}
In this section we consider variations of the two problems presented in the last section, where we replace the promises of respectively a fixed number of marked elements and the amplitude on the marked space, by a lower bound. We simplify algebraically the optimization problem given the particular values of the kernels. We then identify the symmetry group of the problem and describe its actions, in particular we will identify that the input set can be seen as the sum of different orbits of the group. 

\subsubsection{Unstructured search with \texorpdfstring{$\geq M$}{>= M} marked elements}\label{sec:searchgeqM}
Unstructured search with $\geq M$ marked elements is defined as
\begin{problem}[Unstructured search with $\geq M$ marked elements]\label{prob:Unfixed_search}
Given $N\in\mathbb{N}$ the number of elements, and $1\leq M\leq N$ a promise on the number of marked elements, $\mathrm{Search}^N_{\geq M}$ is defined as the triple
\begin{equation}
    \mathrm{Search}^N_{\geq M}\cq\left(\mc{X}_{\geq M}^N,\mc{T}_{\geq M}^N,\mc{O}_{\geq M}^N\right),
\end{equation}
over the input set $\mc{I}=\mc{B}^N_{\geq M}\cq\{x\in\{0,1\}^N:\abs{x}\geq M\}$, and the Hilbert spaces $\mc{H}=\mc{M}=\mathbb{C}^N$; where the indexed families of input states, target states and oracles are defined as
{\footnotesize\begin{equation}
\mc{X}_{\geq M}^N \cq \left\{\ket{\xi_x}=\ket{\psi}\cq\frac{1}{\sqrt{N}}\sum_{i=1}^N\ket{i}\right\}_{x\in \mc{B}^N_{\geq M}}
            \mc{T}_{\geq M}^N \cq \left \{\ket{\tau_x} = \ket{x}\cq \frac{1}{\sqrt{\abs{x}}} \sum_{i:x_i=1} \ket{i}\right \}_{x\in \mc{B}^N_{\geq M}}
    \mc{O}_{\geq M}^N \cq \left\{O_x=\mathds{1}_{\mc M}-2\sum_{i:x_i=1}\proj{i}\right\}_{x\in \mc{B}^N_{\geq M}},
\end{equation}}
where $\{\ket i\}_{1\leq i\leq N}$ is the canonical basis of $\mathbb C^N$.
\end{problem}
It is essentially the same problem as $\mathrm{Search}_M^N$, but now the promise is that at least $M$ elements are marked. We will also recover the same query complexity, and the catalyst will match when the Hamming weight is equal to $M$.
\begin{proposition}\label{prop:searchgeqm}
    Unstructured search with $\geq M$ marked elements has a transduction complexity of 
    \begin{equation}
        W\left(\mathrm{Search}^N_{\geq M}\right)=\sqrt{\frac{N-M}{4M}},
    \end{equation}
    and this complexity can be realized with $\mc{W}=\mathbb C^{N+1-M}$, the following catalyst and input-independent unitary
    \begin{equation}
        \ket{v_x}= \frac{(N-M)^{1/4}}{2|x|^{1/4}}\ket{u_{|x|}}\otimes \left(\ket{x} + \sqrt{\frac{N-|x|}{N-M}} \ket{\neg x}\right),
        \qquad
         S^\circ=\proj{\psi}\otimes S^\circ_\parallel+\left(\mathds{1}-\proj{\psi}\right)\otimes S^\circ_\perp,
    \end{equation}
    where $S^\circ_{\parallel}$, $S^\circ_{\perp}$ are $(N+2-M)\times (N+2-M)$ unitaries fixed respectively up to a phase and a $2\times 2$ unitary, and $\ket{u_{|x|}}$ are unit vectors such that
    \begin{equation}
        \braket{u_{h}|u_{h'}} = \frac{2h^{1/4}{h'}^{1/4}}{\sqrt{h} + \sqrt{h'}}, \quad \forall M\leq h,h'\leq N.
    \end{equation}
\end{proposition}

\begin{proof}
~\paragraph{Simplification of the optimization program} The initial and final Gram matrices, and the operator-valued kernel respectively are
\begin{equation}
    \begin{cases}
        X(x,y)=1\\
        T(x,y)=\frac{\abs{x\land y}}{\sqrt{\abs{x}\abs{y}}}\\
        K(x,y)=O_x^\dagger O_y=O_{x\oplus y}
    \end{cases}
\end{equation}
such that one can write the optimization program~\eqref{Eq:DefAdv} for this particular problem as
\begin{equation}
    \operatorname{W}=\min_{\ket{v_x}\in\mc{L}}\max_{x\in\mc{B}^N_{\geq M}}\norm{v_x}^2\qquad\text{s.t.}\quad 1-\frac{\abs{x\land y}}{\sqrt{\abs{x}\abs{y}}}=\braket{v_x|\mathds{1}_{\mc{L}}-\mathds{1}_{\mc{W}}\otimes O_{x\oplus y}|v_y}\quad\forall x,y\in\mc{B}^N_{\geq M}
\end{equation}
We can decompose the catalyst in the usual basis of $\mc{M}$ such that $\ket{v_x}=\sum_{i=1}^N\ket{v_{x,i}}\ket{i}$, with $\ket{v_{x,i}}\in\mc{W}$ and simplify further the equation
\begin{equation}
    \operatorname{W}=\min_{\ket{v_{x,i}}\in\mc{W}}\max_{x\in\mc{B}^N_{\geq M}}\sum_{i=1}^N\norm{v_{x,i}}^2\qquad\text{s.t.}\quad 1-\frac{\abs{x\land y}}{\sqrt{\abs{x}\abs{y}}}=2\sum_{i:x_i\oplus y_i=1}\braket{v_{x,i}|v_{y,i}}\quad\forall x,y\in\mc{B}^N_{\geq M}
\end{equation}

\paragraph{The symmetry group} of the problem is the symmetric group $\mc{S}_N$. Its action on the bit strings is the standard permutation of the indices, and the representations $\phi^\xi, \phi^\tau, \phi^L$ and $\phi^R$ are all the standard representation\footnote{There is a bit of ambiguity for $\phi^\xi$ that could instead be taken as trivial representation since $\ket{\xi_x}=\ket{\psi}$ does not depend on $x$. This will not affect what follows.}
\begin{equation}
    \phi^S_g\ket{i}=\ket{g(i)}.
\end{equation}
The input set can be divided into orbits representing the bitstrings with fixed Hamming weight
\begin{equation}
    \mc{I}=\mc{B}_{\geq M}^N=\bigcup_{h\geq M}\mc{B}_h^N
\end{equation}

\paragraph{Lower bound on the complexity.} Restricting the input set to the set $\mc{B}_M^N$ of bit strings of $n$ bits with Hamming weight $M$ leads to \cref{prob:Fixed_search}. This implies that the optimal transduction complexity is lower bounded by that of \cref{prob:Fixed_search}, that is,
\begin{equation}
        W\left(\mathrm{Search}^N_{\geq M}\right)\geq W\left(\mathrm{Search}^N_{ M}\right)=\sqrt{\frac{N-M}{4M}}.
\end{equation}
This implies that any catalyst that saturates this bound is necessarily optimal, which will allow us to avoid the symmetrization of the catalyst.
\paragraph{Optimization of the catalyst.} Let us assume that our candidate catalyst is strongly covariant. 
\begin{equation}
    \ket{v_x} \in \mathcal{W} \otimes \mathcal{M} \text{ such that for all permutations $g\in \mc{S}_N$, } \ket{v_{g(x)}} = \sum_{i} \ket{v_{x,i}}\ket{g(i)} 
\end{equation}
Going a bit further, let us define the catalyst in terms of the orbit of $p_0^{N,h}$ and $p_1^{N,h}$ for some constant $h\in \{M,\cdots, N\}$ (as previously done in~\cref{Eq:FixedGroverCataReduc}). Hence, for all $(x,i)\in \mc{B}_h^N$, we have
\begin{equation}
    \ket{v_{x,i}} = \begin{cases}
        \ket{v_{1}^{N,h}} \text{ if } x_i = 1\\
        \ket{v_{0}^{N,h}} \text{ if } x_i = 0 
    \end{cases}
\end{equation}
The program then becomes
\begin{equation}
    \operatorname{W}=\min_{\ket{v_{1}^{N,h}},\ket{v_{0}^{N,h}}\in \mathcal{W}}\max_{N\ge h\ge M} h \norm{v_{1}^{N,h}}^2 + (N-h) \norm{v_{0}^{N,h}}^2\qquad
    \text{s.t.} \begin{cases}
        \frac{1}{2} = h\braket{v_{1}^{N,h}|v_{0}^{N,h'}} + h'\braket{v_{0}^{N,h}|v_{1}^{N,h'}}\\
        \frac{1}{2} = \sqrt{hh'}\braket{v_{1}^{N,h}|v_{0}^{N,h'}} +  \sqrt{hh'}\braket{v_{0}^{N,h}|v_{1}^{N,h'}}
    \end{cases}
\end{equation}
The constraint is a system of equations whose unique solution is
\begin{equation}
        \begin{cases}\braket{v_{1}^{N,h}|v_{0}^{N,h'}} = \frac{1}{2\sqrt{h}\left(\sqrt{h} + \sqrt{h'}\right)}&h\neq h'\\
        \Re\braket{v_{1}^{N,h}|v_{0}^{N,h}} = \frac{1}{4h}\end{cases}.
\end{equation}
    Assuming that $\ket{v_0^{N,h}}$ and $\ket{v_1^{N,h}}$ are parallel and that their scalar product is real, we can decompose them as\footnote{This choice of vector is an assumption, however we will guarantee the optimality of the transduction complexity by comparing to the lower bound expressed earlier.}:
    \begin{equation}
        \begin{cases}
            \ket{v_1^{N,h}} = \frac{\alpha}{2h^{3/4}} \ket{u_h}\\
            \ket{v_0^{N,h}} = \frac{1}{2\alpha h^{1/4}} \ket{u_h}\\
        \end{cases}
    \end{equation}
for some constant $\alpha \in \mathbb{R}^+$, not a function of $h$, and a unit vector $\ket{u_h}$. Note that these are assumptions for now, but they are justified: if they were not aligned or if their scalar product had a non-zero imaginary part, the objective function would just increase. Further, as we will verify that this solution saturates the lower bound, it will guarantee the optimality globally. It may be possible to slightly improve the transduction complexity for some inputs. The constraint now becomes that $\braket{u_h|u_{h'}}=\frac{2h^{1/4}{h'}^{1/4}}{\sqrt{h}+\sqrt{h'}}$, and one could wonder if this is a valid Gram matrix given that $\ket{u_h}$ is supposed to be of unit norm. A simple way to verify it is to find a realization of this Gram matrix:  considering the set $L^2(\mathbb{R}^+)$ of square-integrable functions over $\mathbb{R}^+$ with the usual inner product, the vector $\ket{u_h}$ can be defined as the function $u_h\in L^2(\mathbb{R}^+)$ such that $u_h(t):= \sqrt{2}h^{1/4} e^{-\sqrt{h}t}$ which leads to the appropriate kernel.
With this, the objective function becomes
\begin{equation}
    \operatorname{W}=\min_{\alpha \in \mathbb{R}^+}\max_{N\ge h\ge M} \frac{\alpha^2}{4\sqrt{h}}+\frac{N-h}{4 \alpha^2\sqrt{h}} \min_{\alpha \in \mathbb{R}^+}\frac{\alpha^2}{4\sqrt{M}}+\frac{N-M}{4 \alpha^2\sqrt{M}} = \frac{\sqrt{N-M}}{2\sqrt{M}},
\end{equation}
where we used that the objective function was strictly decreasing with $h$. As expected, our solution has the same transduction complexity as the optimal solution for the problem with a fixed Hamming weight. The best possible value for $\alpha$ is $(N-M)^{1/4}$, which implies the following catalyst as the optimal solution.
\begin{equation}
        \ket{v_x} = \sum_{i} \ket{v_{x,i}}\ket{i}
                  = \sum_{i:x_i=1} \ket{v_1^{N,h}}\ket{i} + \sum_{i:x_i=0} \ket{v_0^{N,h}}\ket{i} 
                  = \frac{(N-M)^{1/4}}{2h^{1/4}}\ket{u_h}\otimes \left(\ket{x} + \sqrt{\frac{N-h}{N-M}} \ket{\neg x}\right)
\end{equation}
where we fixed $\mc W=\operatorname{Span}\left(\ket{u_h},h\in \{M,\dots,N\}\right)$.
\paragraph{Recovering the unitary.}  We now need to find a unitary $S^\circ$ that transforms
\begin{equation}\begin{split}
    \ket{\xi_x}\oplus\left(\mathds{1}_{\mc{W}}\otimes O_x\right)\ket{v_x}=\ket{\psi}\ket{0}+\kappa_h\left(-\ket{x}+ \lambda_h\ket{\neg x}\right)\ket{u_h}\\\xrightarrow{S^\circ}\ket{\tau_x}\oplus\ket{v_x}=\ket{x}\ket{0}+\kappa_h\left(\ket{x}+ \lambda_h\ket{\neg x}\right)\ket{u_h},
\end{split}\end{equation}
where
\[
    \kappa_h=\frac{(N-M)^{1/4}}{2h^{1/4}},
    \qquad \lambda_h=\sqrt{\frac{N-h}{N-M}},
\]
and we swapped $\mc W$ and $\mc M$, and used the vector $\ket{0}$ to distinguish the public state. Note that as effectively $\mc H=\mc M$, the global space now is $\mc{H}\otimes(\mathbb{C}\oplus\mc W).$ We can directly recover the argument that lead us to~\cref{Eq:StructSSearch}, and find that the unitary should decompose as
\begin{equation}
    S^\circ=\proj{\psi}\otimes S^\circ_\parallel+\left(\mathds{1}-\proj{\psi}\right)\otimes S^\circ_\perp,
\end{equation}
where $S^\circ_\parallel$ and $S^\circ_\perp$ are both unitary on $\mathbb{C}\oplus\mc W$. For $M\leq h<N$, we define
\[
    \qquad \mu_h=\sqrt{\frac{h}{N}},
    \qquad \nu_h=\sqrt{\frac{N-h}{N}},
\]
and $\ket{\psi_x^\perp}=\nu_h\ket{x}-\mu_h\ket{\neg x}$.
Moreover, we define the following vectors in $\mathbb C\oplus\mc W$:
\begin{equation}
\begin{aligned}
    \ket{v^\parallel_h}
        &=\begin{pmatrix}
            1\\
            \kappa_h(\lambda_h\nu_h-\mu_h)\ket{u_h}
          \end{pmatrix},
&
    \ket{w^\parallel_h}
        &=\begin{pmatrix}
            \mu_h\\
            \kappa_h(\lambda_h\nu_h+\mu_h)\ket{u_h}
          \end{pmatrix},
\\
    \ket{v^\perp_h}
        &=\begin{pmatrix}
            0\\
            -\kappa_h(\nu_h+\lambda_h\mu_h)\ket{u_h}
          \end{pmatrix},
&
    \ket{w^\perp_h}
        &=\begin{pmatrix}
            \nu_h\\
            \kappa_h(\nu_h-\lambda_h\mu_h)\ket{u_h}
          \end{pmatrix}.
\end{aligned}
\end{equation}
The required transformation then becomes
\begin{equation}
\begin{aligned}
    \ket{\xi_x}\oplus
    (\mathds 1_{\mc W}\otimes O_x)\ket{v_x}
       &=\ket{\psi}\otimes\ket{v^\parallel_h}
         +\ket{\psi_x^\perp}\otimes\ket{v^\perp_h},
\\
    \ket{\tau_x}\oplus\ket{v_x}
       &=\ket{\psi}\otimes\ket{w^\parallel_h}
         +\ket{\psi_x^\perp}\otimes\ket{w^\perp_h}.
\end{aligned}
\end{equation}
Consequently, it suffices to find two unitaries satisfying
\begin{equation}
    \begin{cases}
        S^\circ_\parallel\ket{v^\parallel_h}
            =\ket{w^\parallel_h},\\
        S^\circ_\perp\ket{v^\perp_h}
            =\ket{w^\perp_h},
    \end{cases}
    \quad\forall h.
\end{equation}
For $h=N$, only the parallel equation is needed; its vectors
are given by the same formulas with $\mu_N=1$ and
$\nu_N=\lambda_N=0$. Checking the Gram matrix confirms these unitaries exist. Further, by a dimensional argument, $S^\circ_\parallel$ (respectively $S^\circ_\perp$) is fixed on $N-M+1$ (respectively $N-M$) dimensions out of a $N-M+2$-dimensional Hilbert space, and they are free on the orthogonal complement of $\operatorname{Span}\left(\ket{v^\sigma_h}\right),\sigma\in\{\parallel,\perp\}$. If $V^\sigma$ and $W^\sigma$ are the matrices whose columns respectively are $\ket{v^\sigma_h}$ and $\ket{w^\sigma_h}$, $Q_V^\sigma$ and $Q_W^\sigma$ are matrices encoding an orthonormal basis of $\operatorname{Span}\left(\ket{v^\sigma_h}\right)^\perp$ and $\operatorname{Span}\left(\ket{w^\sigma_h}\right)^\perp$, respectively, then the unitaries necessarily take the form
\begin{equation}
    S^\circ_\sigma=W^\sigma V^{\sigma+}+Q_W^\sigma U^\sigma Q_V^{\sigma\dagger}
\end{equation}
where $+$ denotes Moore-Penrose pseudoinverse, and $U^\sigma$ is a unitary encoding the degrees of freedom.
\end{proof}
\subsubsection{Amplitude amplification with unknown initial amplitude}\label{sec:AmpPin}
In this section, we consider the amplitude amplification problem where the amplitude of the marked state is unknown.

\begin{problem}[Amplitude amplification with initial amplitude $\geq\varepsilon$]\label{prob:unkown-amp}
    Given $\HS$ the public space, $\Pi_m$ the projector on the marked space $\mc H_m\cq \Pi_m\mc H$, and $\eps\in (0,1]$ a promise on the amplitude of the overlap over the marked space, $\mathrm{Amp}_{\geq\eps}$ is defined as the triple
    \begin{equation}
    \mathrm{Amp}_{\geq\eps}\cq\left(\mc{X}_{\geq\eps},\mc{T}_{\geq\eps},\mc{O}_{\geq\eps}\right),
\end{equation}
    over the input set $\mathcal{I}=\mc B_{\geq\eps} \cq \{\ket{\xi}\in\mc{H}:\norm{\xi}=1, \norm{\Pi_m\ket{\xi}}\geq\varepsilon\}$ with $\mc M=\mc H$; where the indexed families of input states, target states and oracles are defined as
{\begin{equation}
    \mc{X}_{\geq\eps} \cq  \mc{B}_{\geq\eps}\qquad
            \mc{T}_{\geq\eps} \cq \left \{\ket{\tau_{\ket\xi}} = \frac{\Pi_m \ket{\xi}}{\norm{\Pi_m \ket{\xi}}}\right \}_{\ket{\xi}\in\mc{B}_{\geq\eps}}\qquad
    \mc{O}_{\geq\eps} \cq \left\{O_{\ket{\xi}}=\mathds{1}_{\mc M}-2\proj{\xi}\right\}_{\ket{\xi}\in\mc{B}_{\geq\eps}}.
\end{equation}}
\end{problem}

This state conversion problem is similar to $\mathrm{Amp}_{\eps}$, with the single difference that the promise now becomes a lower bound; we will however recover the same complexity. In this section we will prove the following proposition:
\begin{proposition}\label{prop:AmpMin}
    Amplitude amplification with initial amplitude $\geq\varepsilon$ has a transduction complexity of 
    \begin{equation}
        W(\mathrm{Amp}_{\geq\eps})=\frac{\sqrt{1-\eps^2}}{2\eps},
    \end{equation}
    and this complexity can be realized with $\mc{W}=\mathbb C\oplus L^2(\mathbb R)$ and the following catalyst and input-independent unitary
    \begin{equation}
        \ket{v_{\ket\xi}}=\frac{\ket{g_s}}{2(1-\eps^2)^{1/4} \sqrt{s}}\left(\sqrt{1-\eps^2}\ket{\xi}+\sqrt{1-s^2}\ket{\xi^\perp}\right),
        \qquad
        S^\circ=\Pi_m \otimes S^{\circ}_m + \Pi_u \otimes S^{\circ}_u
    \end{equation}
    where $\Pi_u\cq\mathds{1}_{\mc H}-\Pi_m$,
    \begin{equation}
    \ket{\xi^\perp}\cq \left(-\frac{\sqrt{1-s^2}}{s}\Pi_m + \frac{s}{\sqrt{1-s^2}}\Pi_u\right)\ket{\xi}.
\end{equation}
    $S^{\circ}_m, S^{\circ}_u$ are unitaries $\mathbb C\oplus L^2(\mathbb R)$, fixed on the complex register and on $\mathbb R^+$, and $\ket{g_{s}}$ is a normalized vector in $L^2(\mathbb R)$ such that
    \begin{equation}
        \braket{g_{s}|g_{s'}} = \frac{2\sqrt{ss'}}{s + s'}, \quad \forall s,s' \in [\varepsilon,1].
    \end{equation}
\end{proposition}
\begin{proof}~
\paragraph{Reformulation of the problem}Similarly to~\cref{sec:AmpFix}, we will consider an alternative formulation of the input set, where we decompose the inputs into the normalized marked and unmarked projections, and consider an additional parameter that represents the amplitude: $\mathcal{I} \equiv \{(\ket{m}, \ket{u}, {s}) \in \HS_m \times \HS_u \times [\varepsilon,1]:\norm{m}=\norm{u}=1\}$. These two representations of the inputs are related in the following way
\begin{equation}
    \begin{cases}
        s\cq\norm{\Pi_m\ket\xi},\\
        c\cq\norm{\Pi_u\ket\xi}=\sqrt{1-s^2},\\
        \ket m\cq s^{-1}\Pi_m\ket \xi,\\
        \ket{u}\cq c^{-1}\Pi_u\ket{\xi}
    \end{cases}\iff\ket{\xi}=s\ket{m}+c\ket{u},
\end{equation}
In the following we will use this bijection implicitly and consider, for a pair of input vectors $\ket{\xi},\ket{\xi'}$, that
\begin{equation}
        \ket{\xi}=s\ket{m}+c\ket{u} 
        \quad\text{and}\quad 
        \ket{\xi'}=s'\ket{m'}+c'\ket{u'}.
\end{equation}
Given these definitions, we can rewrite the Gram and operator-valued kernels defining the problem as
\begin{equation}
    \begin{cases}
        X(\xi, \xi') =ss' \braket{m| m'} + cc'\braket{u| u'}\\
        T(\xi, \xi') = \braket{m | m'}\\
        K(\xi, \xi') = \mathds{1} - 2(\ketbra{\xi}{\xi} + \ketbra{\xi'}{\xi'})+ 4X(\xi,\xi')\ketbra{\xi}{\xi'}
    \end{cases},
\end{equation}
and, writing $\ket{v_\xi}$ the catalyst for the input $\ket{\xi}$, the optimization program that will give us the complexity and the catalyst is
\begin{equation}\label{eq:progAmpAmpnf}
    \operatorname{W}=\min_{\ket{v_{\ket{\xi}}}\in\mc{L}}\max_{\ket{\xi}\in\mc{I}}\norm{v_{\ket{\xi}}}^2\qquad\text{s.t.}\qquad cc'\braket{u|u'} +(ss'-1) \braket{m|m'} =  \braket{v_{\ket{\xi}} | \mathds{1} - (\mathds{1} \otimes O^\dagger_{\ket{\xi}} O_{\ket{\xi'}})|v_{\ket{\xi'}}}\quad \forall \ket{\xi},\ket{\xi'} \in \mathcal{I}.
\end{equation}
\paragraph{The symmetry group of the problem} The symmetry group of the problem is the same as the one for~\cref{prob:fixed-amp}, $\mc{V} = \{V \in U(\mathcal{M}): [V, \Pi_m] = 0\}=U(\mc{H}_m)\oplus U(\mc{H}_u)$, with the same representations presented in~\cref{Eq:SymAmpAmp}.

\paragraph{Optimal catalyst.} Similarly to our analysis of~\cref{prob:Unfixed_search}, we will directly look for a strongly covariant catalyst and verify that it matches the complexity of~\cref{eq:cata-amp}. That is because amplitude amplification with initial amplitude $\varepsilon$ is a subproblem of amplitude amplification with initial amplitude $\geq\varepsilon$, so its transduction complexity is necessarily a lower bound on that of the latter. Strong covariance then implies that
\begin{equation}
    \ket{v_{V\ket\xi}}=\mathds{1}_{\mc{W}}\otimes V\ket{v_{\ket\xi}}\qquad\forall V\in\mc V.
\end{equation}
With the same argument we presented for~\cref{Eq:AmpAmpDecCat}, we can apply~\cref{Corr:SpaceRed} and find
\begin{equation}
    \ket{v_{\ket{\xi}}} = \ket{v^{\parallel}_{\ket{\xi}}}\ket{\xi} + \ket{v^{\perp}_{\ket{\xi}}} \ket{\xi^\perp}
\end{equation} 
for some vectors $\ket{v^{\parallel}_{\ket{\xi}}}, \ket{v^{\perp}_{\ket{\xi}}} \in \mathcal{W}$ with $\ket{\xi^\perp}$ defined in~\cref{Eq:defXiPerp}. If two input vectors, $\ket{\xi}$ and $\ket{\xi'}$ are such that $s=s'$, then we can build a unitary $V_{\xi,\xi'}$ that maps $\ket\xi$ to $\ket\xi'$, that is, $V_{\xi,\xi'}\ket{\xi}=\ket{\xi'}$, with $V_{\xi,\xi'}\ket{u}=\ket{u'}$ and $V_{\xi,\xi'}\ket{m}=\ket{m'}$, so that it also maps $\ket{\xi^\perp}$ to $\ket{{\xi'}^\perp}$, that is, $V_{\xi,\xi'}\ket{\xi^\perp}=\ket{{\xi'}^\perp}$. Using strong covariance, we directly get that the parallel and orthogonal components of their catalysts are necessarily equal. We can thus relate the respective decomposition of their catalysts:
\begin{equation}
    \begin{cases}
        \ket{v^\parallel_s}\cq\ket{v_{\ket\xi}^\parallel}=\ket{v_{\ket{\xi'}}^\parallel}\\
        \ket{v^\perp_s}\cq\ket{v_{\ket\xi}^\perp}=\ket{v_{\ket{\xi'}}^\perp}\\
    \end{cases}\Rightarrow\ket{v_{\ket{\xi}}}=\ket{v_s^\parallel}\ket{\xi}+\ket{v_s^\perp}\ket{\xi^\perp}.
\end{equation}
One can directly verify that $O_\xi\ket{v_{\ket{\xi}}}=-\ket{v_s^\parallel}\ket{\xi}+\ket{v_s^\perp}\ket{\xi^\perp}$ and expand the right-hand side of the condition in~\cref{eq:progAmpAmpnf}
\begin{equation}
    \braket{v_{\ket{\xi}}| \mathds{1} \otimes (1 - O^\dagger_\xi O_{\xi'})|v_{\ket{\xi'}}} = 2\left(\braket{v^\parallel_s | v^\perp_{s'}} \braket{\xi|\xi'^{\perp}} +\braket{v^\perp_s | v^\parallel_{s'}} \braket{\xi^\perp|\xi'} \right).
\end{equation}
We can explicitly state the scalar products between $\ket{\xi}$ and $\ket{{\xi'}^\perp}$:
\begin{equation}
    \begin{cases}
        \braket{\xi^\perp|\xi'}=-cs'\braket{m|m'}+sc'\braket{u|u'}\\
        \braket{\xi|\xi'^\perp}=-sc'\braket{m|m'}+cs'\braket{u|u'},
    \end{cases}
\end{equation}
which leads to
\begin{equation}
    \braket{v_{\ket{\xi}}| \mathds{1} \otimes (1 - O^\dagger_\xi O_{\xi'})|v_{\ket{\xi'}}} = 2\left(-\left(sc'\braket{v^\parallel_s | v^\perp_{s'}}+cs'\braket{v^\perp_s | v^\parallel_{s'}}\right)\braket{m|m'} +\left(cs'\braket{v^\parallel_s | v^\perp_{s'}}+sc'\braket{v^\perp_s | v^\parallel_{s'}}\right)\braket{u|u'}  \right).
\end{equation}
Plugging into~\cref{eq:progAmpAmpnf}, we can rewrite the constraint as
\begin{equation}
    \left(\frac{cc'}{2}-cs'\braket{v^\parallel_s | v^\perp_{s'}}-sc'\braket{v^\perp_s | v^\parallel_{s'}}\right)\braket{u|u'}=\left(\frac{1-ss'}{2}-sc'\braket{v^\parallel_s | v^\perp_{s'}}-cs'\braket{v^\perp_s | v^\parallel_{s'}}\right)\braket{m|m'}.
\end{equation}
In this equation, since we have decomposed an input into a triplet $(\ket m,\ket u,{s})$, by considering different pairs of inputs we can freely make the scalar products $\braket{u|u'}$ and $\braket{m|m'}$ vary independently while keeping the parentheses sitting in front of them, that depend only on ${s}$ and ${s}'$, constant. This implies that this condition can only be satisfied for any pair of inputs if these parentheses are identically zero for any pair ${s},{s}'$. We therefore obtain the following system of equations (considering that the left parenthesis is zero, and that the second plus $ss'/cc'$ times the first is zero too):
\begin{equation}
    \begin{cases}
        \frac{1}{2} = \frac{s}{c}\braket{v^\perp_s | v^\parallel_{s'}}   + \frac{s'}{c'}\braket{v^\parallel_s | v^\perp_{s'}}\\
        \frac{1}{2} = \frac{s'}{c}\braket{v^\perp_s | v^\parallel_{s'}}   + \frac{s}{c'}\braket{v^\parallel_s | v^\perp_{s'}}
    \end{cases},
\end{equation}
For ${s}\neq{s}'$, this system admits as unique solution 
\begin{equation}
    \braket{v_s^\perp|v_{s'}^\parallel}=\frac{c}{2(s+s')},
\end{equation}
while for ${s}={s}'$ we get
\begin{equation}
    \Re(\braket{v_s^\perp|v_{{s}}^\parallel})=\frac{c}{4s}.
\end{equation}
From that last equation and assuming that $\ket{v_s^\parallel}$ and $\ket{v_s^\perp}$ are aligned and their scalar product is real\footnote{As discussed for search, these assumptions make sense as a complex phase or an orthogonal component would increase the objective function. They will also be verified as we saturate the lower bound.}, we can re-express $\ket{v_s^\parallel}$ and $\ket{v_s^\perp}$ as 
\begin{equation}
    \begin{cases}
        \ket{v_s^\parallel}=\frac{\lambda}{2\sqrt{s}}\ket{g_s}\\
        \ket{v_s^\perp}=\frac{c}{2\lambda\sqrt{s}}\ket{g_s}
    \end{cases},
\end{equation}
for a family of normalized vectors $\{\ket{g_s}\}_{{s}\in[\eps,1]}$, parallel to $\ket{v^\parallel_s}$, and a positive number $\lambda$. The constraint now becomes $\braket{g_s|g_{s'}}=\frac{2\sqrt{ss'}}{s+s'}$. The question is the same as for search, whether this is a valid Gram kernel or not, and the same example of functions realizes it, that is, we can use the functions $g_s\in L^2(\mathbb{R}^+)$ such that $g_s(t):=\sqrt{2s}e^{-st}$. With this definition, $\ket{g_s}$ disappears from the optimization program and we can thus rewrite $\operatorname{W}$ as:
\begin{equation}
    \operatorname{W}=\min_{\lambda}\max_{{s}\in[\eps,1]}\frac{\lambda^2}{4s}+\frac{c^2}{4s\lambda^2}=\min_{\lambda}\frac{\lambda^2}{4\eps}+\frac{1-\eps^2}{4\eps\lambda^2}=\frac{\sqrt{1-\eps^2}}{2\eps},
\end{equation}
where we used that the objective function was monotonically decreasing with ${s}$. We therefore recover the lower bound and the expected complexity of $W(P)=\frac{\sqrt{1-\eps^2}}{2\eps}$ for $\lambda=(1-\eps^2)^{1/4}$, and find that the optimal catalyst is
\begin{equation}
    \ket{v_{\ket\xi}}=\frac{\ket{g_s}}{2\sqrt{c_\eps s}}\left(c_\eps\ket{\xi}+c\ket{\xi^\perp}\right),
\end{equation}
with $c_\eps=\sqrt{1-\eps^2}$ and $\mc W=L^2(\mathbb{R}^+).$ One could want to consider only $\operatorname{Span}(\ket{g_s},{s}\in[\eps,1])$; however this is not a proper Hilbert space as it is not closed; thus one would want to consider $\overline{\operatorname{Span}}(\ket{g_s},{s}\in[\eps,1])$ but precisely we will show that $\overline{\operatorname{Span}}(\ket{g_s},{s}\in[\eps,1])=L^2(\mathbb{R}^+)$. Indeed, let us assume toward contradiction that $\mc{W}=\overline{\operatorname{Span}}(\ket{g_s},{s}\in[\eps,1])\subsetneq L^2(\mathbb{R}^+)$. Then, there exists some nonzero $f\in L^2(\mathbb R^+)$ such that $\Pi_{\mc W}f=0$. This condition of no overlap with $\mc{W}$ can be translated in the Laplace domain as a condition on the unilateral Laplace transform $\mc{L}_u(f)=F$, indeed, for any $f\in L^2(\mathbb{R}^+)$ we have
\begin{equation}\label{Eq:EqIPLT}
    \braket{g_s|f}=\sqrt{2s}\int_0^\infty e^{-s t}f(t) dt=\sqrt{2s}F(s),\qquad s>0,
\end{equation}
and thus
\begin{equation}
    \Pi_{\mc W}f=0\iff F(s)=0\quad\forall{s}\in[\eps,1].
\end{equation}
However, the Laplace transform of a square-integrable function is holomorphic on $\Re(z)>0$; since it vanishes on the segment $[\eps,1]$, which has accumulation points, the identity theorem gives $F\equiv0$ on $\Re(z)>0$. By injectivity of the Laplace transform this implies that $f=0$, which is a contradiction, ergo $\mc{W}=\overline{\operatorname{Span}}(\ket{g_s},{s}\in[\eps,1])= L^2(\mathbb{R}^+)$.

\paragraph{Recovering the transducer.}
We now need to find a unitary $S^\circ$ that transforms
\begin{equation}
\begin{aligned}
    \ket{\xi}\oplus
    \left(\mathds{1}_{\mc W}\otimes O_{\ket{\xi}}\right)
    \ket{v_{\ket{\xi}}}
    &=\ket{\xi}\ket{\bot}
      +\kappa\left(-c_\eps\ket{\xi}
      +c\ket{\xi^\perp}\right)\ket{g_s}
    \\
    \xrightarrow{S^\circ}\ket{\tau_\xi}\oplus\ket{v_{\ket{\xi}}}
    &=\ket{m}\ket{\bot}
      +\kappa\left(c_\eps\ket{\xi}
      +c\ket{\xi^\perp}\right)\ket{g_s},
\end{aligned}
\end{equation}
where we introduced $\kappa=\frac{1}{2\sqrt{c_\eps s}}$. We interchanged the $\mc W$ and $\mc H$ registers and used the vector $\ket{\bot}$ to distinguish the public space, so that the global space is now $\mc H\otimes(\mathbb C\oplus\mc W)$.

As the amplitude varies continuously, we use the Hilbert-space realization of the Gram kernel mentioned earlier:
\begin{equation}
    g_s(t)=\sqrt{2s}\,e^{-st},\quad t\geq0, \qquad \mc W=\overline{\operatorname{Span}} \left\{\ket{g_s}:{s}\in[\eps,1]\right\}= L^2(\mathbb R^+).
\end{equation}
On both sides, the group acts by its defining representation on $\mc H$ and trivially on $\mathbb C\oplus\mc W$. The isotypic decomposition therefore separates the marked and unmarked spaces:
\begin{equation}
    \mc H\oplus\mc L \cong \left(\mc H_m\otimes(\mathbb C\oplus\mc W)\right) \oplus \left(\mc H_u\otimes(\mathbb C\oplus\mc W)\right).
\end{equation}
Using~\cref{Corr:Niiiice}, we can thus look for a transducer of the form
\begin{equation}\label{Eq:StructSAmpUnfixed}
    S^\circ=\Pi_m\otimes S^\circ_m+\Pi_u\otimes S^\circ_u,
\end{equation}
where $S^\circ_m$ and $S^\circ_u$ are unitaries on $\mathbb C\oplus\mc W$ (up to an extension of $\mc W$ that we discuss below). Let us define the following four vectors in this space:
\begin{equation}\label{Eq:DecSAmpUnfixed}
\begin{aligned}
    \ket{v^m_s}
        &=\begin{pmatrix}
            s\\-\kappa(c_\eps s+c^2)\ket{g_s}
          \end{pmatrix},
&
    \ket{w^m_s}
        &=\begin{pmatrix}
            1\\\kappa(c_\eps s-c^2)\ket{g_s}
          \end{pmatrix},
\\
    \ket{v^u_s}
        &=\begin{pmatrix}
            c\\\kappa c(s-c_\eps)\ket{g_s}
          \end{pmatrix},
&
    \ket{w^u_s}
        &=\begin{pmatrix}
            0\\\kappa c(s+c_\eps)\ket{g_s}
          \end{pmatrix}.
\end{aligned}
\end{equation}
The required transformation then becomes
\begin{equation}
\begin{aligned}
    \ket{\xi}\oplus
    \left(\mathds{1}_{\mc W}\otimes O_{\ket{\xi}}\right)
    \ket{v_{\ket{\xi}}}
       &=\ket{m}\otimes\ket{v^m_s}
         +\ket{u}\otimes\ket{v^u_s},
\\
    \ket{\tau_\xi}\oplus\ket{v_{\ket{\xi}}}
       &=\ket{m}\otimes\ket{w^m_s}
         +\ket{u}\otimes\ket{w^u_s}.
\end{aligned}
\end{equation}
Consequently, it suffices to find two unitaries satisfying
\begin{equation}\label{eq:ICantAnymore}
    \begin{cases}
        S^\circ_m\ket{v^m_s}=\ket{w^m_s},\\
        S^\circ_u\ket{v^u_s}=\ket{w^u_s},
    \end{cases}
    \qquad\forall{s}\in[\eps,1].
\end{equation}
One can check that the Gram kernels match; however, this is not enough in this case to ensure the existence of the unitaries\footnote{Without equal codimension of the closed span the blocks are only isometries.} since when we are working in infinite dimensions we need to verify that the dimensions of the orthogonal complements of the span match. Let us do it for the input state of the marked part: let us assume a vector $(a,f)\in\mathbb{C}\oplus L^2(\mathbb{R}^+)$ is orthogonal to all vectors $\ket{v^m_s}$, i.e.\ $\braket{v^m_s|(a,f)}=0$. Using~\cref{Eq:EqIPLT}, we can transform this orthogonality condition into an equation in the Laplace domain
\begin{equation}
    as-\kappa(c_\eps s+c^2)\sqrt{2s}F(s)=0\qquad\forall{s}\in[\eps,1].
\end{equation}
By analyticity of $F$ over $\Re(s)>0$, we can extend this constraint to any ${s}>0$, and using the definition of $\kappa$ and $c^2=1-s^2$, this constraint implies
\begin{equation}
    \sqrt{2c_\eps}as+(s^2-c_\eps s-1)F(s)=0\qquad \forall{s}>0.
\end{equation} The polynomial $(s^2-c_\eps s-1)$ possesses a root $s_+>0$, i.e. at ${s}=s_+>0$; evaluating at $s=s_+$ the constraint implies $a=0$. Plugging it in the constraint it becomes $F(s)=0$ for all $s>0$, thus $f=0$. So the only vector $(a,f)$ orthogonal to all $\ket{v^m_s}$ is $(0,0)$. The vectors $\ket{v^m_s}$ therefore span $\mathbb{C}\oplus \mc W$, and the orthogonal complement has dimension $0$. The same argument holds for $\ket{v^u_s},\ket{w^m_s}$, however, for $\ket{w^u_s}$ one can trivially see that the vector $(a,0)$ belongs to the orthogonal complement, so it is not zero-dimensional.
\begin{equation}
\begin{split}
    \overline{\operatorname{Span}}(\ket{v_s^m})=\overline{\operatorname{Span}}(\ket{w_s^m})=\overline{\operatorname{Span}}(\ket{v_s^u})=\mathbb{C}\oplus\mc W,\qquad\overline{\operatorname{Span}}(\ket{w_s^u})=\{0\}\oplus\mc W\\
    \rightarrow \dim \overline{\operatorname{Span}}(\ket{v_s^m})^\perp=\dim\overline{\operatorname{Span}}(\ket{w_s^m})^\perp=\dim\overline{\operatorname{Span}}(\ket{v_s^u})^\perp=0\qquad\dim\overline{\operatorname{Span}}(\ket{w_s^u})^\perp=1
\end{split}.
\end{equation}
Whereas for the marked space everything is fine, it is not the case for the unmarked space where there is a mismatch in the dimension of the complement. It is thus necessary to extend the workspace as discussed in~\cref{App:ProofPropEqProb},
to have infinite-dimensional complements. We could consider adding a register representing a counter $\ell^2(\mathbb{N})$, but we propose the following alternative: we consider the enlarged Hilbert space $\mc W'=L^2(\mathbb{R})=L^2(\mathbb{R}^+)\oplus L^2(\mathbb{R}^-)$ with the usual inner product and find unitaries. One should now consider that $g_s\in L^2(\mathbb{R})$ such that $g_s(t):=\sqrt{2s}e^{-st}h(t)$ where $h(t)$ is the Heaviside step function, so that the catalyst has no component in $L^2(\mathbb{R}^-)$.

The unitaries $S^\circ_\sigma$ with $\sigma\in\{u,m\}$ will have the following block decomposition
\begin{equation}
    S^\circ_\sigma\begin{pmatrix}
        a\\f
    \end{pmatrix}=\begin{pmatrix}
        da+l(f)\\ag(t)+Df
    \end{pmatrix}\Rightarrow S^\circ_\sigma\approx\begin{pmatrix}
        d&l\\g&D
    \end{pmatrix}.
\end{equation} with $d$ complex, $l$ a bounded linear functional on $L^2(\mathbb{R})$, $g\in L^2(\mathbb{R})$ and $D$ a bounded operator on $L^2(\mathbb{R})$. As for search, they are fixed on the span of the $\ket{v^\sigma_s}$, and free on the orthogonal complement. For the scalar part of the marked space ($\sigma=m$), the condition becomes:
\begin{equation}\label{eq:againalabel}
    ds+\frac{s^2-c_\eps s-1}{\sqrt{2c_\eps}}l(e^{-s t}h(t))=1\quad\forall {s}\in[\eps,1].
\end{equation}By analyticity (the map $s\mapsto e^{-st}h(t)$ is holomorphic from $\Re(s)>0$ to $L^2(\mathbb R)$ and $l$ is bounded), this constraint can be extended to all values of ${s}>0$, and in particular, if $s_\pm=\frac{c_\eps\pm\sqrt{c_\eps^2+4}}{2}$ are the two roots of the polynomial, we can take the equation at $s=s_+$ and find $d=s_+^{-1}=-s_-$. Note that $s_+>1$, so this extension beyond $[\eps,1]$ is essential. Plugging that in~\cref{eq:againalabel} we get that:
\begin{equation}
    l(e^{-s t}h(t))=\frac{\sqrt{2c_\eps}(1-ds)}{s^2-c_\eps s-1}=\frac{\sqrt{2c_\eps}s_-(s-s_+)}{(s-s_+)(s-s_-)}=\frac{\sqrt{2c_\eps}s_-}{(s-s_-)}=\sqrt{2c_\eps}s_-\int_0^\infty e^{s_-t}e^{-st}h(t)dt
\end{equation}
where for the last equality one needs to recognize $\int_0^\infty e^{-\alpha t}e^{-\beta t}dt=\frac{1}{\alpha+\beta}$, and that $s_-$ is negative; and thus for functions on the positive reals we have
\begin{equation}
    l(f)=\sqrt{2c_\eps}s_-\int_0^\infty e^{s_-t}f(t) dt.
\end{equation}
On the function part, the constraint is 
\begin{equation}\label{Eq:ConstDmark}
    sg(t)+\frac{s^2-c_\eps s-1}{\sqrt{2c_\eps}}D(e^{-s t}h(t))=\frac{s^2+c_\eps s-1}{\sqrt{2c_\eps}}e^{-s t}h(t)\qquad \forall {s}\in[\eps,1].
\end{equation}
Once again, we can extend the constraint to all positive ${s}$ by analyticity, then evaluating this function at $s=s_+$, we directly find
\begin{equation}
    g(t)=\frac{s_+^2+c_\eps s_+-1}{\sqrt{2c_\eps}s_+}e^{-s_+ t}h(t)=\sqrt{2c_\eps}e^{-s_+ t}h(t).
\end{equation}
Plugging $g$ in~\cref{Eq:ConstDmark}, one can get
\begin{equation}
    D(e^{-s t}h(t))=e^{-s t}h(t)+2c_\eps\frac{(e^{-st}-e^{-s_+t})h(t)}{s-s_+}\frac{s}{(s-s_-)}
\end{equation}
To continue we need to use the two following operators, defined for $s_0>0$:
\begin{equation}
\begin{cases}
    X_{s_0}(f)=e^{s_0t}h(t)\int_t^\infty e^{-s_0\tau}f(\tau)d\tau&\rightarrow X_{s_0}(e^{-st}h(t))=\frac{e^{-st}h(t)}{s+s_0}\rightarrow (\mathds{1}_{L^2(\mathbb R)}-s_0X_{s_0})(e^{-st}h(t))=\frac{se^{-st}h(t)}{s+s_0}\\
    Y_{s_0}(f)=e^{-s_0 t}h(t)\int_0^t e^{s_0\tau}f(\tau)d\tau&\rightarrow Y_{s_0}(e^{-st}h(t))=\frac{e^{-s t}-e^{-s_0t}}{s_0-s}h(t)
\end{cases}
\end{equation}
Both are convolutions with the integrable kernel $e^{-s_0\abs{u}}$ (restricted to $\tau>t$ and $\tau<t$ respectively), hence bounded with $\norm{X_{s_0}},\norm{Y_{s_0}}\leq s_0^{-1}$, and $X_{s_0}=Y_{s_0}^\dagger$. Since $-s_->0$, the multiplication of $e^{-st}h(t)$ by $s/(s-s_-)$ is realized by $\mathds{1}+s_-X_{-s_-}$, and thus we directly find
\begin{equation}
    D=\mathds{1}_{L^2(\mathbb R)}-2c_\eps Y_{s_+}(\mathds{1}_{L^2(\mathbb R)}+s_-X_{-s_-}).
\end{equation}
Since $l$, $X_{-s_-}$ and $Y_{s_+}$ only involve values of $f$ on $\mathbb{R}^+$ and produce functions supported on $\mathbb{R}^+$, $S^\circ_m$ acts as the identity on $L^2(\mathbb{R}^-)$. On $\mathbb{C}\oplus L^2(\mathbb{R}^+)$ it is unitary: the Gram kernels of $\{\ket{v^m_s}\}$ and $\{\ket{w^m_s}\}$ match and both families are total, so the isometry defined on the span extends to a unitary. Hence $S^\circ_m$ is unitary on $\mathbb{C}\oplus L^2(\mathbb{R})$. Globally
\begin{equation}
\boxed{
    S^\circ_m\begin{pmatrix}
        a\\f
    \end{pmatrix}=\begin{pmatrix}
        a/s_++\sqrt{2c_\eps}s_-\int_0^\infty e^{s_-\tau}f(\tau) d\tau\\
        a\sqrt{2c_\eps}e^{-s_+ t}h(t)+\left[\mathds{1}_{L^2(\mathbb R)}-2c_\eps Y_{s_+}(\mathds{1}_{L^2(\mathbb R)}+s_-X_{-s_-})\right]f(t)
    \end{pmatrix}
    }
\end{equation}

For the unmarked space ($\sigma=u$), the constraint reads, after dividing by $c$,
\begin{equation}\label{Eq:ConstUnmarkTime}
    S^\circ_u\begin{pmatrix}1\\\frac{s-c_\eps}{\sqrt{2c_\eps}}e^{-st}h(t)\end{pmatrix}=\begin{pmatrix}0\\\frac{s+c_\eps}{\sqrt{2c_\eps}}e^{-st}h(t)\end{pmatrix}\qquad\forall{s}\in[\eps,1].
\end{equation}
For the scalar part, evaluating at $s=c_\eps$ gives $d=0$, and then $l(e^{-st}h(t))=0$ for all $s>0$, so $l$ vanishes on $L^2(\mathbb{R}^+)$. For the function part, evaluating at $s=c_\eps$ gives
\begin{equation}
    g(t)=\sqrt{2c_\eps}e^{-c_\eps t}h(t)=:e_{c_\eps}(t),
\end{equation}
and plugging it back, $D(e^{-st}h(t))=e^{-st}h(t)-2c_\eps\frac{e^{-st}-e^{-c_\eps t}}{c_\eps-s}h(t)$, that is $D=\mathds{1}-2c_\eps Y_{c_\eps}$ on $L^2(\mathbb{R}^+)$. This part maps $\mathbb{C}\oplus L^2(\mathbb{R}^+)$ unitarily onto $\{0\}\oplus L^2(\mathbb{R}^+)$: in the Laplace domain $\mathds{1}-2c_\eps Y_{c_\eps}$ is the multiplication by the Blaschke factor $\frac{z-c_\eps}{z+c_\eps}$, an isometry whose range is the orthogonal complement of $e_{c_\eps}$. As anticipated, nothing is mapped onto the scalar part, and this is where $L^2(\mathbb{R}^-)$ is needed. On $L^2(\mathbb{R}^-)$ the action is free, and we choose the mirror image of the previous block:
\begin{equation}
    l(f)=\sqrt{2c_\eps}\int_{-\infty}^0e^{c_\eps\tau}f(\tau)d\tau,\qquad D f=f-2c_\eps Y^-_{c_\eps}f\quad\text{for }f\in L^2(\mathbb{R}^-),
\end{equation}
with $Y^-_{s_0}(f)=e^{-s_0t}h(-t)\int_{-\infty}^te^{s_0\tau}f(\tau)d\tau$. Under the reflection $t\mapsto-t$, this block is the adjoint of the block on $\mathbb{C}\oplus L^2(\mathbb{R}^+)$ (indeed $Y^-_{c_\eps}$ becomes $X_{c_\eps}=Y_{c_\eps}^\dagger$), so it maps $L^2(\mathbb{R}^-)$ unitarily onto $\mathbb{C}\oplus L^2(\mathbb{R}^-)$. Altogether
\begin{equation}
    \boxed{
        S^\circ_u\begin{pmatrix}
            a\\f
        \end{pmatrix}=\begin{pmatrix}
            \sqrt{2c_\eps}\int_{-\infty}^0e^{c_\eps\tau}f(\tau)d\tau\\
            a\sqrt{2c_\eps}e^{-c_\eps t}h(t)+\left(\mathds{1}-2c_\eps\left(Y_{c_\eps}+Y^-_{c_\eps}\right)\right)f
        \end{pmatrix}
    }
\end{equation}
which is unitary on $\mathbb{C}\oplus L^2(\mathbb{R})$, as it is the direct sum of a unitary $\mathbb{C}\oplus L^2(\mathbb{R}^+)\to L^2(\mathbb{R}^+)$ and a unitary $L^2(\mathbb{R}^-)\to\mathbb{C}\oplus L^2(\mathbb{R}^-)$, and satisfies~\cref{Eq:ConstUnmarkTime}.

\end{proof}
\subsection{Amplitude estimation}\label{sec:AmpEstimation}

In this section, we consider amplitude estimation on a simple two-dimensional Hilbert space $\mathbb{C}^2$. We will first find the optimal algorithm for what we call amplitude estimation on the full circle defined in~\cref{Prob:AmpEsC}. In this case the algorithm takes as input an angle $\theta\in S^1=\mathbb{R}/2\pi\mathbb Z$, representing the vector $\ket{\theta}=\cos\theta\ket0+\sin\theta\ket 1$, and prepares from a fixed input state $\ket{\xi_0}$ a target state $\ket{f_\theta}$ in some Hilbert space $\mc H$, where the family of states $\{\ket{f_\theta}\}_{\theta\in S^1}$ is such that $\braket{f_\theta|f_{\theta'}}=\braket{f_{\theta-\theta'}|f_0}$. This definition of amplitude estimation may seem peculiar, as it is different from the problem consisting in evaluating $\theta$ directly~\cite{BHMT02} or the decision problem that would ask to prepare a state $\ket{0}$ or $\ket{1}$ depending on whether $\theta$ is above or below some threshold $\theta_t$ for all $\theta$ in some interval. 

However both can be reduced to this problem. For the first one, one can consider the special case where $\ket{f_\theta}$ is a function $f_\theta\in L^2(S^1)$, where $S^1\approx[0,2\pi)$ is the unit circle, corresponding to a square window of width $\delta\in[0,\pi]$, all angles are understood modulo $2\pi$, more precisely
\begin{equation}
    f_\theta(\theta'):=\begin{cases}
        \delta^{-1/2} &\abs{\theta-\theta'}\leq\frac{\delta}{2}\\
        0&\frac{\delta}{2}\leq\abs{\theta-\theta'}\leq\frac{\pi}{2}\\
        f_\theta(\theta'+k\pi)=f_\theta(\theta')&k\in\mathbb Z
    \end{cases}
\end{equation}
In that case the problem is equivalent to determining the angle $\theta$ up to an error $\delta$. To recover the decision problem one should simply prepare $\ket{f_\theta}$, then to determine whether $\theta$ is in some interval $I$ (for instance, the angle larger than a certain threshold), one should simply measure with respect to the projector on functions identically zero outside of $I$. Note that the particular choice of window given in the example would give an infinite complexity due to the discontinuity, but considering smoother windows such as
\begin{equation}
    f_\theta(\theta')=\begin{cases}
        \sqrt{\frac{2}{\delta}}\cos(\pi (\theta-\theta')/\delta)&\abs{\theta-\theta'}\leq\frac{\delta}{2}\\
        0&\frac{\delta}{2}\leq\abs{\theta-\theta'}\leq\frac{\pi}{2}\\
        f_\theta(\theta'+k\pi)=f_\theta(\theta')&k\in\mathbb Z
    \end{cases},
\end{equation}
one will find that the complexity scales as $\Theta(\delta^{-1})$ (see \cref{fig:window_function}). This is different from the standard decision amplitude estimation transducer obtained by the conversion from the traditional Markov algorithm or the construction from Jeffery, Osborne and Pass~\cite{JP26}, that has variable cost over $O((\theta-\theta_t)^{-1})$ that can blow up when the input is close to the threshold, whereas here the cost is a fixed $O(\delta^{-1})$ but has an error when $\abs{\theta-\theta'}\leq\delta/2$.

\begin{figure}[h]

\begin{tikzpicture}
    \begin{axis}[
    axis lines=middle,
    xmin=-0.3, xmax=7,
    ymin=-0.3, ymax=1.4,
    xtick=\empty,
    ytick=\empty,
    xlabel={$\theta'$},
    ylabel={$f_\theta(\theta')$},
]

    \addplot[
        domain=2:4,
        samples=200,
        thick,
    ]
    { sqrt(2/2) * cos(deg(pi*(3 - x)/2)) };

    \addplot[thick, domain=-0.3:2, samples=2] {0};
    \addplot[thick, domain=4:7, samples=2] {0};

    \draw[<->, thick]
        (axis cs:2,1.1) -- (axis cs:4,1.1)
        node[midway, above] {$\delta$};

    \draw[dashed]
        (axis cs:2,1) -- (axis cs:4,1);

    \draw[dashed]
        (axis cs:2,0) -- (axis cs:2,1);

    \draw[dashed]
        (axis cs:4,0) -- (axis cs:4,1);

    \draw[<->, thick]
        (axis cs:4.4,0) -- (axis cs:4.4,1)
        node[midway, right] {$\sqrt{\frac{2}{\delta}}$};

    \node[thick] at (axis cs:-0.1,-0.1) {$0$};

    \draw[thick]
        (axis cs:6,-0.05) -- (axis cs:6,0.05);
    \node at (axis cs:6,-0.1) {$2\pi$};

    \draw[thick]
        (axis cs:3,-0.05) -- (axis cs:3,0.05);
    \node at (axis cs:3,-0.1) {$\theta$};

    \draw[thick]
        (axis cs:2,-0.05) -- (axis cs:2,0.05);
    \node at (axis cs:2,-0.15) {$\theta - \frac{\delta}{2}$};

    \draw[thick]
        (axis cs:4,-0.05) -- (axis cs:4,0.05);
    \node at (axis cs:4,-0.15) {$\theta + \frac{\delta}{2}$};

\end{axis}
\end{tikzpicture}
\caption{Plot of a window function $f_\theta(\theta')$ with parameter $\delta = 2$. As $\delta$ gets large, the function tends to be constant with value $\sqrt{\frac{2}{\delta}}$ on the interval $[\theta - \frac{\delta}{2}, \theta + \frac{\delta}{2}]$ with a periodicity of $\pi$.}
    \label{fig:window_function}
\end{figure}
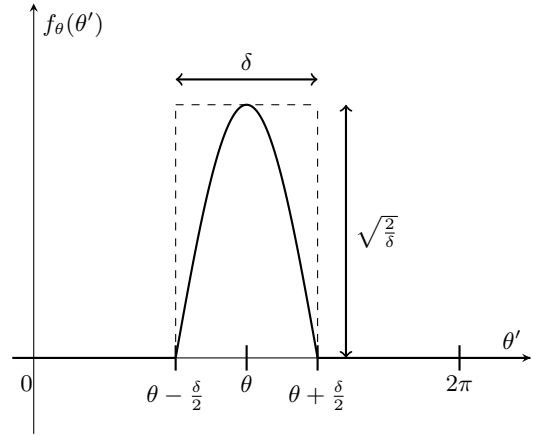

\subsubsection{Amplitude amplification on the full circle}

First let us formally define the problem of amplitude amplification on the full circle. In this section we will consider the vectors and inner products
\begin{equation}\label{eq:BitOfTrig}
    \begin{cases}
        \ket{\theta}=\cos\theta\ket0+\sin\theta\ket1\\
        \ket{\theta_\perp}=-\sin\theta\ket0+\cos\theta\ket1
    \end{cases}\Longrightarrow
\begin{cases}
    \braket{\theta|\theta'}=\cos(\theta-\theta')\\
    \braket{\theta|\theta'_\perp}=\sin(\theta-\theta').
    \end{cases}
    \end{equation}
Note that $\ket{\theta_\perp}=\ket{\theta+\frac{\pi}{2}}=R_{\pi/2}\ket{\theta}$, where $R_\omega$ is the rotation matrix such that $R_\omega\ket{\theta}=\ket{\theta+\omega}$.

\begin{problem}[Amplitude estimation on the full circle]\label{Prob:AmpEsC}
Given $\HS$ the public space, a state $\ket{\xi_0}\in \mc H$ and $\{\ket{f_\theta}\}_{\theta\in S^1}$ a $\pi$-periodic family of states in $\mc H$ that are continuously indexed by $\theta\in S^1$ and such that their inner products satisfy $\braket{f_\theta|f_{\theta'}}=\braket{f_{\theta-\theta'}|f_{0}}=:g(\theta-\theta')$ where $S^1=\mathbb{R}/2\pi\mathbb Z$ is the unit circle, $\mathrm{Est}_g$ is defined as the triple 
\begin{equation}
    \mathrm{Est}_{g}\cq\left(\mc{X}_{g},\mc{T}_{g},\mc{O}_{g}\right),
\end{equation}
    over the input set $\mathcal{I}=S^1$ with $\mc M=\mathbb C^2$; where the where the indexed families of input states, target states and oracles are defined as
{\begin{equation}
    \mc{X}_{g} \cq\left\{\ket{\xi_\theta}\cq\ket{\xi_0}\right\}_{\theta\in S^1}\qquad
            \mc{T}_{g} \cq \{\ket{f_\theta}\}_{\theta\in S^1}\qquad
    \mc{O}_{g} \cq \left\{O_{\theta}\cq\mathds{1}-2\proj\theta\right\}_{\theta\in S^1}.
\end{equation}}
\end{problem}
Note that the $\pi$-periodicity of $f$ is necessary for the problem to be feasible. Otherwise, as $O_\theta=O_{\theta+\pi}$, the algorithm should be able to distinguish two inputs without access to an oracle capable to do so. What motivates the choice of $g$ as an index distinguishing the instances of the problem is that it is the determining component in the complexity of the transducer.

We prove the following proposition.
\begin{proposition}\label{prop:ampEst}
    Amplitude estimation on the full circle has a transduction complexity of 
    \begin{equation}
        W(\mathrm{Est}_{g})=\sum_{n\in2\mathbb Z+1}\abs{\hat h_n},\quad\text{with }h(\omega)=\lim_{\omega'\rightarrow\omega}\frac{1-g(\omega')}{2\sin\omega'}=\sum_{n\in2\mathbb Z+1}\hat h_ne^{i\omega n}
    \end{equation}
    and this complexity can be realized with $\mc{W}=\ell^2(2\mathbb Z+1)$ and the following catalyst 
    \begin{equation}
    \ket{v_\theta}=\sum_{n\in2\mathbb Z+1} \sqrt{\frac{\abs{\hat h_n}}{2}}e^{-in\theta}\ket{n}\left(\ket{\theta}+e^{i\arg(\hat h_n)}\ket{\theta_\perp}\right).
\end{equation} 
\end{proposition}
\begin{proof}~
\paragraph{Reformulation of the problem} 
The kernel of the problem can be expressed as
\begin{equation}
    \begin{cases}
        X(\theta,\theta')=1\\
        T(\theta,\theta')=g(\theta-\theta')\\
        K(\theta,\theta')=\mathds1_2-2\left(\proj\theta+\proj{\theta'}\right)+4\cos(\theta-\theta')\ketbra\theta{\theta'}
    \end{cases}.
\end{equation}
such that we can write the adversary bound as
\begin{equation}\label{Eq:AdvAmpEs}
    \operatorname{W}=\min_{\ket{v_\theta}}\max_\theta\norm{v_\theta}^2\quad\text{s.t.}\quad 1-g(\theta-\theta')=\braket{v_\theta|\mathds{1}_{\mc W}\otimes(\mathds1_2-O_{\theta}^\dagger O_{\theta'})|v_\theta'}.
\end{equation}
\paragraph{Symmetry group of the problem} The symmetry group of the problem is the circle group $U(1)$, acting on the input set as standard rotations $r_\omega\theta=\theta+\omega$. It acts on the input vector, target vector and oracle as
\begin{equation}\label{eq:KerAmpEs}
    \begin{cases}
        \phi^\xi_{\omega}\ket{\xi_0}=\ket{\xi_0}\\
        \phi^\tau_\omega\ket{f_\theta}=\ket{f_{\theta+\omega}}\\
        \phi^L_\omega O_\theta\phi^{R\dagger}_\omega=R_\omega O_\theta R_\omega^\dagger=O_{\theta+\omega}
    \end{cases}
\end{equation}
where $R_\omega$ is the usual rotation matrix. 
\paragraph{Optimal catalyst} A natural way to decompose the catalyst is in the eigenbasis of the oracle of the corresponding input $\theta$, in a way that its action is easy to write
\begin{equation}\label{eq:DecCatAmpE}
    \ket{v_\theta}=\ket{v_\theta^\parallel}\ket{\theta}+\ket{v_\theta^\perp}\ket{\theta_\perp}\longrightarrow O_\theta\ket{v_\theta}=-\ket{v_\theta^\parallel}\ket{\theta}+\ket{v_\theta^\perp}\ket{\theta_\perp},
\end{equation}
and the adversary bound becomes
\begin{equation}\label{Eq:AdvAmpEs2}
    \operatorname{W}=\min_{\ket{v_\theta}}\max_\theta\norm{v_\theta^\parallel}^2+\norm{v_\theta^\perp}^2\quad\text{s.t.}\quad 1-g(\theta-\theta')=2\sin(\theta-\theta')\left(\braket{v_\theta^\parallel|v_{\theta'}^\perp}-\braket{v_\theta^\perp|v_{\theta'}^\parallel}\right).
\end{equation}
where we used the identity given in~\cref{eq:BitOfTrig}. Using~\cref{theorem:symmetryCat}, one can show that 
\begin{equation}
    \begin{cases}
        \ket{v_\theta^\parallel}=\phi_\theta^{\mc W}\ket{v_0^\parallel}\\
        \ket{v_\theta^\perp}=\phi_\theta^{\mc W}\ket{v_0^\perp}\\
    \end{cases}
\end{equation}
for some representation $\phi^{\mc W}$  of $U(1)$, allowing us again to simplify the adversary bound
\begin{equation}
    \operatorname{W}=\min_{\substack{\ket{v_0}\in\mc{W}\\\phi^{\mc{W}}\in\operatorname{Hom(U(1),U(\mc W))}}}\norm{v_0^\parallel}^2+\norm{v_0^\perp}^2\quad\text{s.t.}\quad \frac{1-g(\omega)}{2\sin(\omega)}=\left(\braket{v_0^\parallel|\phi^{\mc W\dagger}_\omega|v_{0}^\perp}-\braket{v_0^\perp|\phi^{\mc W\dagger}_\omega|v_0^\parallel}\right).
\end{equation}
Note that the division by sine in $\omega=0$ and $\omega=\pi$ should be understood by continuity. Indeed as $\ket{f_\theta}$ is continuous and $\ket{f_\theta}=\ket{f_{\theta+\pi}}$, $g(\theta)$ is continuous and $g(0)=g(\pi)=1$. We will thus generally write
\begin{equation}
    h(\omega)=\begin{cases}
        \frac{1-g(\omega)}{2\sin(\omega)}&\omega\neq0,\pi\\
        \frac{-g'(0)}{2}&\omega=0\\
        \frac{g'(\pi)}{2}&\omega=\pi\\
    \end{cases}.
\end{equation}
Further, one can note that if $h(\omega)$ is not constant, we cannot choose the trivial representation $\phi_\omega\mc{W}=\mathds{1}_{\mc{W}}$, or there would be a contradiction in the constraint. The only way for $h$ to be constant is that $g(\omega)=1~\forall\omega$, that is the trivial problem. We can thus conclude that the strong covariance is possible only for the trivial problem.

Now we will in the following section use group theory to demonstrate that the catalyst will encode the Fourier coefficients of $h$. As stated in~\cref{theorem:symmetryCat}, we can safely assume that $\phi^{\mc W}$ is strongly continuous and unitary. We know that the strongly continuous irreducible representations of the group $U(1)$ are quite simple, the representation $\phi^n$ for $n\in\mathbb{Z}$ is represented by the character $\chi_n(\omega)=e^{in\omega}$, and acts on a $1$-dimensional space as the phase $e^{in\omega}$. With this we can derive the isotypic decomposition of $\mc W$ as
\begin{equation}
    \mc{W}=\bigoplus_{n\in\mathbb Z}\mc W^{(n)}\qquad \text{with }\phi_\omega^{\mc W}|_{\mc W^{(n)}}\cong\phi^n_\omega\otimes \mathds{1}_{\mc W^{(n)}}=e^{in\omega}\mathds{1}_{\mc W^{(n)}}.
\end{equation}
We can define the projector onto each of the isotypic components as 
\begin{equation}
    \Pi_n=\frac{1}{2\pi}\oint_{\omega\in[0,2\pi)}d\omega e^{-in\omega}\phi^{\mc W}_\omega.
\end{equation}
Decomposing both parts of the catalyst in that basis, we get
\begin{equation}
    \begin{cases}
        \ket{a_n}=\Pi_n\ket{v_0^\parallel}\\
        \ket{b_n}=\Pi_n\ket{v_0^\perp}
    \end{cases},
\end{equation}
and we can once again reexpress the adversary bound using all our machinery as
\begin{equation}
    \operatorname{W}=\min_{\substack{\ket{v_0}\in\mc{W}\\\phi^{\mc{W}}\in\operatorname{Hom(U(1),U(\mc W))}}}\sum_{n\in\mathbb Z}\norm{a_n}^2+\norm{b_n}^2\quad\text{s.t.}\quad h(\omega)=2i\sum_{n\in\mathbb Z}e^{-i\omega n}\Im\braket{a_n|b_n}.
\end{equation}
Note that by Cauchy-Schwarz inequality the series are absolutely convergent. We have a similar situation as for search or amplitude amplification, and we will want to assume that $\ket{a_n}$ and $\ket{b_n}$ are both parallel, and that their inner product is imaginary. We could use the same trick as presented earlier, but we will present another one that will avoid us to introduce derivatives. Setting $t_n=2\Im\braket{a_n|b_n}$, we can find the following lower bound
\begin{equation}
    \abs{t_n}\leq 2\norm{a_n}\norm{b_n}\leq\norm{a_n}^2+\norm{b_n}^2,
\end{equation}
and we can saturate that bound if we choose 
\begin{equation}
    \ket{a_n}=\sqrt{\frac{\abs{t_n}}{2}}\ket{n}\qquad\ket{b_n}=i\mathrm{Sign(t_n)}\sqrt{\frac{\abs{t_n}}{2}}\ket{n},
\end{equation}
for some unit vectors $\ket{n}$. Thus, without loss of generality we can assume that the multiplicity of each representation is one, and that the latest choice of vector is optimal. Finally the adversary bound becomes
\begin{equation}
    \operatorname{W}=\min_{t_n}\sum_{n\in\mathbb Z}\abs{t_n}\quad\text{s.t.}\quad h(\omega)=i\sum_{n\in\mathbb Z}e^{-i\omega n}t_n.
\end{equation}
One can recognize that the constraint necessarily is a Fourier series and thus the $t_n$ are fixed
\begin{equation}
    \hat h_n=\frac{1}{2\pi}\oint_{\omega\in[0,2\pi)}h(\omega)e^{-in\omega}d\omega
    \Rightarrow
     t_n=-i\hat h_{-n}.
\end{equation}
and as $h(\omega)=-h^*(-\omega)$, the Fourier coefficients of $h$ are purely imaginary, and thus $t_n$ are real as expected. Finally, by noting that $h$ is $\pi$-antiperiodic, \textit{ergo} the odd Fourier coefficients are non-zero
\begin{equation}
    \operatorname{W}=\sum_{n\in\mathbb2 N+1}\abs{\hat h_n}.
\end{equation}
We know that the optimal catalyst lives in $\ell^2(2\mathbb{Z}+1)\otimes\mathbb{C}^2$ and can be written as
\begin{equation}
    \ket{v_\theta}=\sum_{n\in2\mathbb Z+1} \sqrt{\frac{\abs{\hat h_n}}{2}}e^{-in\theta}\ket{n}\left(\ket{\theta}+e^{i\arg(\hat h_n)}\ket{\theta_\perp}\right)\quad\text{with }h(\omega)=\sum_{n\in2\mathbb Z+1}\hat h_ne^{i\omega n}.
\end{equation}    
\end{proof}

\subsubsection{Amplitude amplification on an interval}

\begin{problem}[Amplitude estimation on an interval]
Given the input set $\mc I\subset S^1$ an interval on $S^1$ such that $\abs{\mc I}<\pi$, $\HS$ the public space, a state $\ket{\xi_0}\in \mc H$ and $\{\ket{f_\theta}\}_{\theta\in I}$ a family of states in $\mc H$ that are continuously indexed by $\theta\in S^1$ and such that their inner products satisfy $\braket{f_\theta|f_{\theta'}}=\braket{f_{\theta-\theta'}|f_{0}}=:g(\theta-\theta')$ , $\mathrm{Est}_g$ is defined as the triple 
\begin{equation}
    \mathrm{Est}^{\mc{I}}_{g}\cq\left(\mc{X}_{g},\mc{T}_{g},\mc{O}_{g}\right),
\end{equation}
    with $\mc M=\mathbb C^2$; where the where the indexed families of input states, target states and oracles are defined as
{\begin{equation}
    \mc{X}^{\mc{I}}_{g} \cq\left\{\ket{\xi_\theta}\cq\ket{\xi_0}\right\}_{\theta\in \mc I}\qquad
            \mc{T}^{\mc{I}}_{g} \cq \{\ket{f_\theta}\}_{\theta\in \mc I}\qquad
    \mc{O}^{\mc{I}}_{g} \cq \left\{O_{\theta}\cq\mathds{1}-2\proj\theta\right\}_{\theta\in \mc I}.
\end{equation}}
\end{problem}
Note that we also require that $\abs{\mc I}<\pi$, otherwise $\mc I$ we would require, as in last section, that $\ket{f_\theta}$ is $\pi$-periodic (as the oracle itself is $\pi$-periodic, we cannot distinguish two angles that are separated by an angle of $\pi$), and we can thus extend it on the full unit circle $S^1$ and recover the previous problem.

We prove the following proposition.
\begin{proposition}
    Amplitude estimation on the interval $\mc I$ has an upper bound on its complexity transduction complexity of 
    \begin{equation}
    W\left(\mathrm{Est}^{\mc{I}}_{g}\right)\leq \inf_{L,h\in C(-L,L)} \sum_{n\in\mathbb Z }\abs{\hat h_n} \qquad\text{s.t.}\begin{cases}
        h(\omega)=\lim_{\omega'\rightarrow\omega}\frac{1-g(\omega')}{2\sin\omega'}\quad\forall\omega\in[-\abs{\mc I},\abs{\mc I}]\\
        L\geq\abs{\mc I}\\
        h(-L)=h(L)=0\\
        h(\omega)=\sum_{n\in\mathbb Z}\hat h_ne^{i\omega n\frac{\pi}{L}},
    \end{cases}
\end{equation}
    and this complexity can be realized with $\mc{W}=\ell^2(2\mathbb Z+1)$, the following catalyst 
    \begin{equation}
    \ket{v_\theta}=\sum_{n\in\mathbb Z} \sqrt{\frac{\abs{\hat h_n}}{2}}e^{-in\theta\frac{\pi}{L}}\ket{n}\left(\ket{\theta}+e^{i\arg(\hat h_n)}\ket{\theta_\perp}\right)\quad\text{with }h(\omega)=\sum_{n\in\mathbb Z}\hat h_ne^{i\omega n\frac{\pi}{L}}.
\end{equation} 
\end{proposition}
\begin{proof}~
    \paragraph{Efficient catalyst} 
The kernels are the same as those in~\cref{eq:KerAmpEs} and the adversary bound can thus be expressed as~\cref{Eq:AdvAmpEs}. We can do the same decomposition as in~\cref{eq:DecCatAmpE}, and find again~\cref{Eq:AdvAmpEs2}.
In this case the problem does not possess the symmetry $U(1)$, so we cannot use it to prove the optimality of the catalyst, however we can find a feasible catalyst, that we can hope is good by assuming that for some value of $L$
\begin{equation}
    \begin{cases}
        \ket{v_\theta^\parallel}=\sum_{n\in\mathbb Z}a_ne^{-in\theta\frac{\Pi}{L}}\ket{n}\\
        \ket{v_\theta^\perp}=\sum_{n\in\mathbb Z}b_ne^{-in\theta\frac{\Pi}{L}}\ket{n}\\
    \end{cases}.
\end{equation}
Indeed as the function $g$ is not defined on the full circle but only on $\omega\in[-\abs{\mc{I}},\abs{\mc{I}}]$, we are not forced to consider the Fourier series on the full circle, but only on an interval $[-L,L]$, with $L\geq \abs{\mc I}$. This is why we considered $\abs{\mc I}<\pi$, otherwise the function $g$ is fixed on the full circle and we recover to previous problem. In that case, with the same calculations as in the last section, we can find that the optimal catalyst should be 
\begin{equation}
    \ket{v_\theta}=\sum_{n\in\mathbb Z} \sqrt{\frac{\abs{\hat h_n}}{2}}e^{-in\theta\frac{\pi}{L}}\ket{n}\left(\ket{\theta}+e^{i\arg(\hat h_n)}\ket{\theta_\perp}\right)\quad\text{with }h(\omega)=\sum_{n\in\mathbb Z}\hat h_ne^{i\omega n\frac{\pi}{L}},
\end{equation}
where $h$ is a continuous continuation of $\lim_{\omega'\rightarrow\omega}\frac{1-g(\omega')}{2\sin\omega'}$ on $[-L,L]$. Note that now we sum on all modes, and not only the odd ones, as the continuation does not have the property $h(\omega+L)=-h(\omega)$ we obtained in the last section. The adversary bound is thus
\begin{equation}
    W\left(\mathrm{Est}^{\mc{I}}_{g}\right)\leq\sum_{n\in\mathbb Z }\abs{\hat h_n}.
\end{equation}
One might ask why we should not just consider $L=\abs{\mc I}$, but this is not optimal. For instance if $h(\abs{\mc I})=-h^*(-\abs{\mc I})\neq0$, then the sum of the modulus of the Fourier coefficients diverges, whereas an analytical continuation that cancels in $\pm L$ would have prevented that problem and led to a finite norm solution.
In conclusion we have that 
\begin{equation}
    W\left(\mathrm{Est}^{\mc{I}}_{g}\right)\leq \inf_{L,h\in C(-L,L)} \sum_{n\in\mathbb Z }\abs{\hat h_n} \qquad\text{s.t.}\begin{cases}
        h(\omega)=\lim_{\omega'\rightarrow\omega}\frac{1-g(\omega')}{2\sin\omega'}\quad\forall\omega\in[-\abs{\mc I},\abs{\mc I}]\\
        L\geq\abs{\mc I}\\
        h(-L)=h(L)=0\\
        h(\omega)=\sum_{n\in2\mathbb Z+1}\hat h_ne^{i\omega n\frac{\pi}{L}},
    \end{cases}
\end{equation}
this upper bound of the transduction complexity is realized for a catalyst living in $\ell^2(2\mathbb Z+1)\otimes\mathbb C^2$
\begin{equation}
    \ket{v_\theta}=\sum_{n\in\mathbb Z} \sqrt{\frac{\abs{\hat h_n}}{2}}e^{-in\theta\frac{\pi}{L}}\ket{n}\left(\ket{\theta}+e^{i\arg(\hat h_n)}\ket{\theta_\perp}\right)\quad\text{with }h(\omega)=\sum_{n\in\mathbb Z}\hat h_ne^{i\omega n\frac{\pi}{L}}.
\end{equation}    
This concludes our proof.
\end{proof}
\subsection{Example of non-feasible strong covariance solution}
Let us show two small counter-examples where strong covariance is first not possible, second not optimal.
\subsubsection{No strongly covariant solution}
\begin{problem}[Search with cyclic scalar oracle]\label{prob:SearchCE1} Given $N\in \mathbb N$ the number of elements, $\mathrm{SearchC}^N$ is defined as the triple:
\begin{equation}
    \mathrm{SearchC}^N := (\mathcal{X}^N, \mathcal{T}^N, \mathcal{O}^N)
\end{equation}
over the input set $\mathcal{I} := \mathbb{Z}_n$ with $\mathcal{H} = \mathcal{M} = \mathbb{C}^N $; where the indexed families of input states, target states and oracles are defined as 
\begin{equation}
    \mathcal{X}^N := \left\lbrace \ket{\xi_j} = \ket{\psi} := \frac{1}{\sqrt{N}} \sum_{i=0}^{N-1} \ket{i}\right\rbrace_{j\in \mathcal{I}}  \quad \mathcal{T}^N := \left\lbrace \ket{\tau_j} = \ket{j}\right\rbrace_{j\in \mathcal{I}} \quad \mathcal{O}^N := \left\lbrace O_j=\omega_N^j\quad\text{s.t}\quad \omega = e^{2i\pi/n} \right\rbrace_{j \in \mathcal{I}}
\end{equation}
where $\{\ket{i}\}_{1 \le i \le N}$ is the canonical basis of $\mathbb C^N$.
\end{problem}
We will prove the following proposition
\begin{proposition}
    \emph{Search with a cyclic scalar oracle} has a no strongly covariant feasible point, whereas there exists a weakly covariant catalyst
    \begin{equation}
        \ket{v_j}=\sum_{r=0}^{N-2}\sqrt{\frac{N-1-k}{N}}\omega_N^{kj}\ket{k},\quad \norm{v_j}^2=\frac{N-1}{2}
    \end{equation}
     with $\mc W=\mathbb C^{N-1}$.
\end{proposition}
\begin{proof}~
    \paragraph{Derivation of the kernels, and the symmetry group of the problem}The kernels of this problem are
\begin{equation}
    \begin{cases}
        X(j,k)=1\\
        T(j,k)=\delta_{j,k}\\
        K(j,k)=\omega_N^{k-j}
    \end{cases},
\end{equation}
and thus the adversary bound is 
\begin{equation}\label{eq:Ineed2label}
    \operatorname{W}\left(\mathrm{SearchC}^N\right)=\min_{\ket{v_j}}\max_j\norm{v_j}^2\qquad\text{s.t.}\quad 1-\delta_{j,k}=\braket{v_j|v_k}(1-\omega_N^{k-j}) \quad \forall j,k\in \mathcal{I}.
\end{equation}
\paragraph{No strongly covariant feasible point}
The symmetry group of the problem is the cyclic group $\mc C_N$. In this particular case, the left and right representations are not necessarily uniquely defined, but the only possible choices for the representation $\phi_a^R$ and $\phi_a^L$ for an element $a\in \mathcal{C}_N$ are the $N$ pairs:
\begin{equation}
    \phi^R_a=\omega_N^{ra},\quad\phi^L_a=\omega_N^{(r+1)a}\qquad r\in\{0,\dots,N-1\}.
\end{equation}
For a particular choice of $r$, strong covariance implies
\begin{equation}
    \ket{v_k}=\phi^R_{k-j}\ket{v_j}=\omega_N^{r(k-j)}\ket{v_k}\qquad\text{and thus}\qquad\braket{v_j|v_k}=\omega_N^{r(k-j)}\norm{v_0}^2.
\end{equation}
Plugging that in the constraint of the adversary bound one finds
\begin{equation}
    1-\delta_{j,k}=\norm{v_0}^2\omega_N^{r(k-j)}(1-\omega_N^{k-j})=-2i\norm{v_0}^2\omega^{(r+1/2)(k-j)}\sin\left(\frac{\pi(k-j)}{N}\right).
\end{equation}
We can see that the left side is real whereas the right hand side has a non zero imaginary part, this is a contradiction therefore there is no strongly covariant solution. 
\paragraph{Feasible weakly covariant catalyst} 
We can choose as a feasible point the following catalyst
    \begin{equation}
        \ket{v_i}=\sum_{r=0}^{N-2}\sqrt{\frac{N-1-k}{N}}\omega_N^{ki}\ket{k},\quad \norm{v_i}^2=\frac{N-1}{2}.
    \end{equation}
Using similar techniques as in~\ref{sec:AmpEstimation}, one can further prove that this catalyst is optimal.
\end{proof}
\subsubsection{A non-optimal strongly covariant solution with a two-dimensional oracle}
We give a counter-example with three inputs and a two-dimensional oracle. Strongly covariant feasible points exist, but none is optimal.

We first define the constant $\theta = \frac{5\pi}{12}$, $c=\cos\theta$, $s=\sin \theta$ and $\omega = e^{2\pi i/3}$. Additionally, we set
\begin{equation}
    R=\begin{pmatrix}1&0\\0&\omega\end{pmatrix},\qquad
    U=\begin{pmatrix}c&-s\\s&c\end{pmatrix},\qquad
    L=U\begin{pmatrix}\omega&0\\0&\omega^2\end{pmatrix}U^\dagger
\end{equation}
With this mind, we provide the following problem.
\begin{problem}[Conversion to a cosine kernel]\label{prob:CosineCE}
$\mathrm{Con}$ is defined as the triple:
\begin{equation}
    \mathrm{Con} := (\mathcal{X}, \mathcal{T}, \mathcal{O})
\end{equation}
over the input set $\mathcal{I} := \mathbb{Z}_3$ with $\mc H=\mc M=\mathbb C^2$; where the indexed families of input states, target states and oracles are defined as
\begin{equation}
    \mathcal{X} := \left\lbrace \ket{\xi_j}=\ket{0}\right\rbrace_{j\in \mathcal{I}} \quad \mathcal{T} := \left\lbrace \ket{\tau_j}=\cos\left(\frac{2\pi j}{3}\right)\ket{0}
    +\sin\left(\frac{2\pi j}{3}\right)\ket{1}\right\rbrace_{j\in \mathcal{I}} \quad 
    \mathcal{O} := \left\lbrace O_j = L^jR^{-j}\right\rbrace_{j\in \mathcal{I}} 
\end{equation}.
\end{problem}
From the above definition, we have $X(j,k)=1$ and $T(j,k)=\cos(2\pi(j-k)/3)$. In particular, for any $a\in \mathcal{I}$,we have $O_{j+a}=L^aO_jR^{-a}$. So we may choose $\phi^R_a=R^a$ and $\phi^L_a=L^a$.
We will prove that
\begin{proposition}
The minimal norm of a strongly covariant catalyst over all admissible oracle representations is $2$, whereas
\begin{equation}
    \operatorname{W}(\mathrm{Con})\leq\frac{10+4\sqrt3}{10+\sqrt3}<\frac32<2.
\end{equation}
\end{proposition}
\begin{proof}~
\paragraph{Strongly covariant feasible points.}
For this problem, the adversary bound is
\begin{equation}
    W(\text{Con})=\min_{\ket{v_j}}\max_j\norm{v_j}^2\quad \text{s.t.}\quad1-\cos\left(\frac{2\pi(j-k)}{3}\right)=\braket{v_j|\mathds{1}_{\mc L}-(\mathds{1}_{\mc W}\otimes O_j^\dagger O_k)|v_k} \quad \forall j,k\in \mathcal{I}.
\end{equation}
For our choice of $\phi^L$ and $\phi^R$, strong covariance would imply that every catalyst $\ket{v_j}$ is defined as $\ket{v_j}=(\mathds{1}_{\mc W}\otimes R^j)\ket{v_0}$. Hence, with $d=k-j$, the constraint becomes 
\begin{equation}\label{eq:CosineStrong}
    1-\frac12(\omega^d+\omega^{2d})
    =\braket{v_0|\mathds{1}_{\mc W}\otimes(R^d-L^d)|v_0}.
\end{equation}
Let us write $\ket{v_0}=\ket{a}\otimes\ket{0}+\ket{b}\otimes\ket{1}$ where $\ket{a},\ket{b}\in \mathcal{W}$.
Comparing the coefficients of $1$ and $\omega^{2d}$ gives
\begin{equation}
    \norm{a}^2=1,\qquad \norm{-s a+c b}^2=\frac12.
\end{equation}
Since $s-c=1/\sqrt2$, the triangle inequality implies
\begin{equation}
    c\norm{b}\geq s\norm{a}-\norm{-s a+c b}
    =s-\frac1{\sqrt2}=c.
\end{equation}
Consequently $\norm{v}^2\geq2$. Equality is attained with $\mc W=\mathbb C$ and
\begin{equation}
    \ket{v_j}=R^j(\ket{0}+\ket{1})=\ket{0}+\omega^j\ket{1}.
\end{equation}
The remaining Fourier coefficient in~\eqref{eq:CosineStrong} follows because both sides vanish at $d=0$.

\paragraph{A better weakly covariant catalyst.}
Let $\ket{-}=(\ket{0}-\ket{1})/\sqrt2$ and set
\begin{equation}
    \alpha=\frac{10}{10+\sqrt3},\qquad
    \beta=\frac{4\sqrt3}{10+\sqrt3}.
\end{equation}
With $\mc W=\mathbb C^2$, take
\begin{equation}
    \ket{v_j}=\sqrt\alpha\ket{0}_{\mc W}\otimes\ket{0}_{\mc M}
    +\sqrt\beta\,\omega^{2j}\ket{1}_{\mc W}\otimes R^j\ket{-}_{\mc M}.
\end{equation}
This family is weakly covariant with $\phi^W_a=\operatorname{diag}(1,\omega^{2a})$ and has cost $\alpha+\beta<3/2$.
To check feasibility, put $z=\omega^{k-j}$. The two orthogonal workspace components contribute respectively
\begin{equation}
    f_0(z)=1-c^2z-s^2z^2,\qquad
    f_1(z)=\frac14-\frac34z+\frac12z^2.
\end{equation}
Indeed, $f_0(z)=\braket{0|R^{k-j}-L^{k-j}|0}$ and
$f_1(z)=z^2\braket{-|R^{k-j}-L^{k-j}|-}$, using $z^3=1$.
Since $c^2=(2-\sqrt3)/4$ and $s^2=(2+\sqrt3)/4$, direct substitution gives
\begin{equation}
    \alpha f_0(z)+\beta f_1(z)=1-\frac12(z+z^2)=X_{j,k}-T_{j,k}.
\end{equation}
Thus the strongly covariant feasible point is not optimal.

\paragraph{Other choices of representations.}
For completeness, the symmetry group is exactly $\mc C_3$: since
$\det K_{j,k}=\omega^{2(k-j)}$, preservation of $K$ up to conjugation forces every input permutation to be a translation.
Moreover, direct calculation shows that $O_1$ and $O_2$ do not commute. Their common commutant in dimension two therefore consists only of scalars. Any other admissible right representation must consequently be $\phi^{R\prime}_a=\omega^{ra}R^a$, for $r\in\mathbb Z_3$, with the same twist on the left representation.
For $r=1$, the constant Fourier coefficient of the right-hand side of~\eqref{eq:CosineStrong}, after multiplication by $\omega^{rd}$, is nonpositive and cannot equal $1$.
For $r=2$, the coefficient of $\omega^{2d}$ is nonnegative and cannot be equal to $-1/2$.
Only $r=0$ admits strongly covariant feasible points, and its minimum is $2$ as proved above.
\end{proof}

\section{Acknowledgment}
The authors thank Y. Zhang and S. Apers for fruitful discussions that motivated in particular the search for a transducer for the amplitude estimation problem as defined in~\cref{sec:AmpEstimation}. We also thank A. Darras and A. Flachs for their help and discussions around this paper.

This work was supported by the Belgian Fonds de la Recherche Scientifique - FNRS under Grant O.0013.22 (EoS CHEQS).
Benoît Dubus is a Research Fellow of the Fonds de la Recherche Scientifique - FNRS.
\section{Use of AI}

In this paper, AI was used for technical steps in different proofs, to find the counter-example to strong covariance, to generate the code for figures and for mathematical and language verifications. The authors assume the full responsibility for the text of the paper.
\newpage
\bibliography{references/ref2}
\newpage
\appendix
\section{Proof of Proposition 1.}\label{App:ProofPropEqProb}
\begin{proof}[Proof of $(\Leftarrow)$ from~\cref{Prop:EqProb}]
Assume $X=X'$, $T=T'$ and $[K]=[K']$. We construct $U^\xi, U^\tau\in U(\mc{H})$ and $U^L,U^R\in U(\mc{M})$ as in the definition of unitary equivalence.

\emph{States.} Let $S\cq\mathrm{Span}\{\ket{\xi_i}\}_{i\in\mc{I}}$ and $S'\cq\operatorname{Span}\{\ket{\xi_i'}\}_{i\in\mc{I}}$. Define $V:S\to S'$ on finite linear combinations by
\begin{equation}
    V\Big(\sum_i c_i\ket{\xi_i}\Big)\cq\sum_i c_i\ket{\xi_i'}.
\end{equation}
For any finitely supported $c\in\mathbb{C}^{\mc{I}}$,
\begin{equation}
    \Big\|\sum_i c_i\ket{\xi_i'}\Big\|^2=\sum_{i,j}\bar c_i c_j X'(i,j)=\sum_{i,j}\bar c_i c_j X(i,j)=\Big\|\sum_i c_i\ket{\xi_i}\Big\|^2 .
\end{equation}
Hence $V$ is well defined (a combination that vanishes is mapped to a vector of norm zero) and isometric, and it is surjective onto $S'$ by construction. It thus extends by continuity to a unitary $\bar S\to\bar S'$, where the bar denotes the closure. As stated in the assumption of the proposition we assume $\dim\bar S^\perp=\dim\bar S'^\perp$ and thus picking any unitary $W:\bar S^\perp\to\bar S'^\perp$, set $U^\xi\cq V\oplus W\in U(\mc{H})$. Then $U^\xi\ket{\xi_i}=\ket{\xi_i'}$ for all $i\in\mc{I}$. The same argument applied to $T=T'$ gives $U^\tau$.

\emph{Oracles.} Since $[K]=[K']$, there is $U^R\in U(\mc{M})$ with $K'(i,j)=U^RK(i,j)U^{R\dagger}$ for all $i,j\in\mc{I}$. Fix any $i_0\in\mc{I}$ and define
\begin{equation}
    U^L\cq O'_{i_0}\,U^R\,O_{i_0}^\dagger\in U(\mc{M}).
\end{equation}
Then, for every $j\in\mc{I}$,
\begin{equation}
    O'_j=O'_{i_0}\,O'^\dagger_{i_0}O'_j=O'_{i_0}K'(i_0,j)=O'_{i_0}U^R O_{i_0}^\dagger O_j\,U^{R\dagger}=U^L O_j U^{R\dagger},
\end{equation}
which concludes the proof.
\end{proof}
\section{Proof of symmetry theorems}
\subsection{Proof of~\texorpdfstring{~\cref{theorem:symmetryCat}}{Theorem 5}}\label{proof:SymmetryCat}
\begin{proof}[Proof of \cref{theorem:symmetryCat}] A necessary condition for the algorithm $\mc A=(S^\circ,\{\ket{v_i}\})$ to solve the problem $P\cong(X,T,[K])$ is that
\begin{equation}
    X(i,j)-T(i,j)=\braket{v_i|\mathds{1}_{\mc W}\otimes (\mathds{1}_{\mc M}-K(i,j)|v_j}\qquad\forall i,j\in\mc I,
\end{equation}
for a fixed $K\in[K]$. If this is true for any input pairs $i,j$, it is also true for any rotated pairs $g(i),g(j)$:
\begin{equation}
    X(g(i),g(j))-T(g(i),g(j))=\braket{v_{g(i)}|\mathds{1}_{\mc W}\otimes (\mathds{1}_{\mc M}-K(g(i),g(j))|v_{g(j)}}\qquad\forall i,j\in\mc I,~\forall g\in\mc{G},
\end{equation}
and thus using the transformations of the kernels defining $P$ in~\cref{eq:DefRepProb}, the rotation $g$ disappears from $X$ and $T$, and can be transferred from the oracle kernel $K$ onto the catalyst
\begin{equation}
    X(i,j)-T(i,j)=\left(\bra{v_{g(i)}}\left(\phi^{c}_g\otimes\phi^R_g\right)\right)\left(\mathds{1}_{\mc W}\otimes (\mathds{1}_{\mc M}-K(i,j)\right)\left(\left(\phi^{c}_g\otimes\phi^R_g\right)^\dagger\ket{v_{g(j)}}\right)\qquad\forall i,j\in\mc I,~\forall g\in\mc{G},
\end{equation}
for an arbitrary representation $\phi^c$ of $\mc{G}$ on $\mc{W}$. We thus directly see that a rotated catalyst $\ket{v_{i,g}}\cq\left(\phi^{c}_g\otimes\phi^R_g\right)^\dagger\ket{v_{g(i)}}$ is now a new valid catalyst, in the sense that it is a feasible point of the adversary bound, and thus can be transformed into a valid transducer.
One can note that an algorithm $\mc A_g\cq(S^\circ_g,\{\ket{v_{i,g}}\}$, for some unitary $S^\circ_g$ that as discussed in~\cref{Subsec:PreTra} necessarily exists, would have the same transduction complexity as $\mc A$, indeed
\begin{equation}
    W(\mc A_g)=\sup_{i\in\mc I}\norm{v_{i,g}}=\sup_{i\in\mc I}\norm{v_{g(i)}}=\sup_{i\in\mc I}\norm{v_{i}}=W(\mc A).
\end{equation}
Further, we define a symmetrized catalyst living in a larger private space $\mc L'=\mc H_{\mc G}\otimes\mc L$, where $\mc H_{\mc G}$ is the group Hilbert space introduced in~\cref{sec:PreSym}, as the controlled uniform superposition over all rotated catalysts:
\begin{equation}
    \ket{\tilde{v}_i}\cq \frac{1}{\sqrt{\mu(\mc{G})}}\int_{\mc{G}}dg\ket{g^{-1}}\otimes\ket{v_{i,g}}\qquad\forall i\in\mc{I}.
\end{equation}
Plugging the definitions of $\ket{\tilde{v}_i}$ and $\ket{v_{i,g}}$ and using~\cref{eq:canGroupBasis}, one can check that 
\begin{equation}
    X(i,j)-T(i,j)=\braket{\tilde v_i|\mathds 1_{\mc W}\otimes (\mathds 1_{\mc M}-K(i,j)|\tilde v_j}\qquad\forall i,j\in\mc I,
\end{equation}
therefore $\{\ket{\tilde v_i}\}_{i\in\mc I}$ is a valid family of catalysts. Once again, as discussed in~\cref{Subsec:PreTra}, we thus know that there exists an input-independent unitary $\tilde S^\circ$ such that $\tilde {\mc A}=(\tilde S^\circ,\{\ket{\tilde v_i}\})$ solves $P$. We can now check that $\{\ket{\tilde v_i}\}_{i\in\mc I}$, and by extension $\tilde {\mc A}$, are weakly covariant, and that it has a transduction complexity lower or equal than that of $\mc{A}$. 
First, developing again the definition of $\ket{\tilde{v}_i}$ and $\ket{v_{i,g}}$, and doing the change of variable $g\mapsto gh^{-1}$, one can check that 
\begin{equation}
    \ket{\tilde{v}_{h(i)}}=\left(L_h\otimes\phi^{c}_h\otimes\phi^R_h\right)\ket{\tilde{v}_i},
\end{equation}
with $L$ the left regular representation, and the catalysts (and the algorithm) are thus weakly covariant. Note that the representation $L_h\otimes\phi^{c}_h$ is unitary, and strongly continuous for compact Lie groups.
Further, we can also verify that the norm of the catalyst is bounded by the supremum of the catalysts on the orbit using max-bound on averages
\begin{equation}
    \norm{\tilde{v}_i}^2=\frac{\int_{\mc{G}}dg\norm{v_{g(i)}}^2}{\mu(\mc{G})}=\frac{\int_{\mc{G}i}dj \norm{v_j}^2}{\mu(\mc{G}i)}\leq \sup_{j\in\mc{G}i}\norm{v_j}^2.
\end{equation}
Since any $\ket{\tilde{v}_i}$ is upper-bound by some $\ket{{v}_i}$, we can conclude that globally
\begin{equation}
    W(\tilde{\mc A})=\sup_{j\in\mc{I}}\norm{\tilde{v}_j}^2\leq\sup_{j\in\mc{I}}\norm{v_j}^2=W(\mc A),
\end{equation}
and therefore the transduction complexity can only be reduced, which concludes the proof.
\end{proof}

\subsection{Proof of~\texorpdfstring{~\cref{theorem:symmetryCat}}{Theorem 7}}\label{proof:SymmetryTrans}
\begin{proof}[Proof of \cref{theorem:symmetryTrans}]
Given the input-independent unitary $S^\circ$ and the family of weakly covariant catalysts $\ket{v_i}$, we write $\mc{X}=\mathrm{Span}\left(\ket{\xi_i}\oplus\left(\mathds{1}_{\mc{W}}\otimes O_i\right)\ket{v_i},\forall i\in\mc{I}\right)$ and $\mc{T}=\mathrm{Span}\left(\ket{\tau_i}\oplus\ket{v_i},\forall i\in\mc{I}\right)=S^\circ\mc{X}$. The transduction condition gives that 
\begin{equation}\label{eq:ineedalabel}
    S^\circ\left(\ket{\xi_i}\oplus\left(\mathds{1}_{\mc{W}}\otimes O_i\right)\ket{v_i}\right)=\ket{\tau_i}\oplus\ket{v_i}\qquad\forall i\in\mc I.
\end{equation}
In particular, we can consider $i=g(j)$, with the transformations defined in~\cref{eq:DefRepProb} and in~\cref{def:weakCov}, and find that
\begin{equation}
    S^\circ\left(\phi_g^\xi\oplus\phi_g^{\mc{W}}\otimes\phi_g^L\right)\left(\ket{\xi_j}\oplus\left(\mathds{1}_{\mc{W}}\otimes O_j\right)\ket{v_j}\right)=\left(\phi_g^\tau\oplus\phi_g^{\mc{W}}\otimes\phi_g^R\right)\ket{\tau_j}\oplus\ket{v_j} \qquad\forall j\in\mc I, g\in\mc G.
\end{equation}
Further, reusing~\cref{eq:ineedalabel}, we obtain
\begin{equation}
    S^\circ\left(\phi_g^\xi\oplus\phi_g^{\mc{W}}\otimes\phi_g^L\right)\left(\ket{\xi_j}\oplus\left(\mathds{1}_{\mc{W}}\otimes O_j\right)\ket{v_j}\right)=\left(\phi_g^\tau\oplus\phi_g^{\mc{W}}\otimes\phi_g^R\right)S^\circ\left(\ket{\xi_j}\oplus\left(\mathds{1}_{\mc{W}}\otimes O_j\right)\ket{v_j}\right)\qquad \forall j\in\mc I, g\in\mc G.
\end{equation}
Since the vectors $\left(\ket{\xi_j}\oplus\left(\mathds{1}_{\mc{W}}\otimes O_j\right)\ket{v_j}\right)$ span $\mc X$, and $\mc T$ is the image of $\mc X$ by $S^\circ$, we directly get 
\begin{equation}
     S^\circ\left(\phi_g^\xi\oplus\phi_g^{\mc{W}}\otimes\phi_g^L\right)|_{\mc X}=\left(\phi_g^\tau\oplus\phi_g^{\mc{W}}\otimes\phi_g^R\right)|_{\mc T}S^\circ
\end{equation}
\end{proof}

\end{document}